\documentclass[11pt]{article}

\usepackage{amsmath}
\usepackage{amsthm}
\usepackage{amssymb}
\usepackage{adjustbox} 
\usepackage{algorithmic}
\usepackage{bbm}
\usepackage{color}
\usepackage{caption,subcaption}
\usepackage{dsfont} 
\usepackage{graphicx}
\usepackage{hyperref}
\usepackage{mathrsfs}
\usepackage[margin=1in]{geometry}
\usepackage{mathtools}
\usepackage{rotating}
\usepackage{tikz}
\usepackage{wrapfig}
\usepackage{setspace}
\usepackage{booktabs} 
\usepackage{accents} 
\usepackage{soul} 
\usepackage{relsize} 
\usepackage{lscape} 
\usepackage{xcolor} 
\setstcolor{red}
\usepackage{multirow}
\usepackage{booktabs,siunitx}
\usepackage[utf8]{inputenc}
\usepackage{float}
\restylefloat{table}
\usepackage{enumitem}
\newlist{Steps}{enumerate}{1}
\setlist[Steps,1]{label=Step~\arabic*.,leftmargin=*}

\newcommand{\indep}{\mathbin{\rotatebox[origin=c]{90}{$\models$}}}

\newcommand{\E}{\mathbb{E}}
\newcommand{\Var}{\text{Var}}

\newcommand{\X}{\mathbf{X}}

\newcommand{\C}{\mathbf{C}}

\newcommand{\Y}{\mathbf{Y}}

\DeclareMathAlphabet\mathbfcal{OMS}{cmsy}{b}{n}

\usepackage{accents}
\newcommand{\dbtilde}[1]{\accentset{\approx}{#1}}

\newcommand{\redsout}{\bgroup\markoverwith{\textcolor{red}{\rule[0.5ex]{2pt}{0.4pt}}}\ULon}

\usepackage{tikz}

\newtheorem{defn}{Definition}
\newtheorem*{defn*}{Definition}
\newtheorem{prop}{Proposition}

\newtheoremstyle{break}
  {\topsep}{\topsep}%
  {\itshape}{}%
  {\bfseries}{}%
  {\newline}{}%
\theoremstyle{break}

\usepackage[numbers,sort&compress]{natbib}

\title{State space model multiple imputation for missing data in non-stationary multivariate time series with application in digital Psychiatry}
\author{Xiaoxuan Cai, Xinru Wang, Li Zeng, Habiballah Rahimi Eichi, Dost Ongur, Lisa Dixon, \\ Justin T. Baker, Jukka-Pekka Onnela, Linda Valeri}

\begin{document}
\maketitle

\abstract{Mobile technology (e.g., mobile phones and wearable devices) provides effective and scalable methods for collecting physiological and behavioral biomarkers in patients' naturalistic settings, as well as opportunity for therapeutic advancements and scientific discoveries regarding the etiology of psychiatric illness. 
Continuous data collection yields a new type of data: entangled multivariate time series of outcome, exposure, and covariates. 
Missing data is a pervasive problem in biomedical and social science research, and Ecological Momentary Assessment (EMA) in psychiatric research via mobile devices is no exception. 
However, complex data structure of multivariate time series and non-stationarity make missing data a major challenge for proper inference.
Time series analyses typically include history information as explanatory variables to control for auto-correlation, exacerbating the missing data problem and potentially rendering unfeasible to adjust appropriately for confounding. 
The majority of available imputation methods are either designed for longitudinal data with limited follow-up times or for stationary time series. 
Limited work on non-stationary time series either focuses on missing exogenous information or ignores the complex relationship among outcome, exposure and covariates time series. 
How to handle missing data in complex non-stationary multivariate time series is a key problem that remains unresolved, and the performance of existing imputation methods remains to be evaluated in the context of non-stationary mobile device data. 
We propose a novel data imputation solution based on the state space model and multiple imputation to properly address missing data in non-stationary multivariate time series. 
We demonstrate its advantages over other widely used missing data imputation strategies by evaluating its theoretical properties and empirical performance in simulations of both stationary and non-stationary time series, subject to various missing mechanisms. 
We apply the proposed method to investigate the association between digital social interaction and negative mood in a multi-year smartphone observational study of bipolar and schizophrenia patients.
\\[1em]
}

\doublespace

\section{Introduction}
Mobile devices (e.g., mobile phones and wearable devices) have revolutionized the way we access information and provide convenient and scalable methods for collecting psychological and behavioral biomarkers in patients’ naturalistic settings \citep{world2011mhealth,ben2017mobile}.
An important application of mobile technology is psychiatric research to monitor patients' psychiatric symptoms, social interaction, life habits (e.g., sleeping, smoking and alcohol consumption), as well as other health-related conditions, continuously in their real life.
Intensive monitoring on participants increases the number of observations to hundreds or even thousands, resulting in longitudinal data as entangled multivariate time series.
Consequently, it greatly complicates research design, statistical analysis, and the handling of missing data.

Severe Mental Illness (SMI), including schizophrenia, bipolar disorder, schizoaffective disorder, and related conditions, has become one of the major burdens of disease, affecting over 13 million individuals in the United States \citep{NSDUH2019}.
Antipsychotics, albeit critical in controlling manic and depressive episodes, fall short in improving patient's quality of life \citep{johnson1999social,lieberman2005effectiveness}. 
On the other hand, social support, in particularly the size of social network \citep{corrigan2004social,hendryx2009social}, has been demonstrated to be critical for promoting sustained improvements in symptoms and social functioning \citep{johnson1999social,pevalin2003social,corrigan2004social,hendryx2009social}. 
Our research is motivated by the Bipolar Longitudinal Study (BLS), a multi-year observational smartphone study for bipolar and schizophrenia patients. 
Passive data, such as GPS, phone use, accelerometer data, and anonymized basic information of phone call and text logs, were collected continuously without any necessary direct involvement from study participants. 
Active data, such as self-reported symptoms, behaviors, and experiences, also known as Ecological Momentary Assessment (EMA) \citep{wichers2011momentary}, were obtained through daily reports from study participants. 
We aim to investigate the association between phone-based social network size and patient's self-reported negative mood, while controlling for other behavioral and environmental factors (e.g., physical activity captured by accelerometer data and outdoor temperature). 
Despite the BLS study being unique in its multiyear follow-up and wealth of data collected in the context of SMI, we anticipate that long-term close monitoring of patients will become the norm in the future as mobile technology becomes more prevalent in clinical settings. 

Missing data is a significant issue in mobile device data, due to long-term follow up and intensive information collection.
It can occur for several reasons, for instance, when devices lose power, are switched off (for both passive and active data collection), or lose connection to Wi-Fi (for passive data collection), or when participants decline or ignore scheduled surveys (for active data collection). The missing rate in the BLS study ranges from less than $10\%$ to more than $90\%$ across participants.
We highlight three unique challenges posed by missing data in mobile device time series research.
First, history information (e.g., lagged outcomes) are frequently included as explanatory variables in time series analysis to control for auto-correlation. Therefore, even if the missingness is restricted to outcomes only, this is not a univariate missing data problem and the missing rate for analysis can be significantly elevated.
Second, discarding incomplete records in time series can disrupt the temporal relationship between variables, introducing bias into statistical inference.
Finally, psychiatric patients exhibit significant variation in disease subtype, pharmacological and behavioral treatments, and recovery trajectory. 
For example, the BLS study recruited patients with Bipolar I and II disorder, and schizophrenia; participants experience varying number of relapses, frequent medication adjustments, and different types of psychotherapies.
As a result, collected time series are likely to differ remarkably across patients and be non-stationary over the lengthy follow-up period.
Non-stationarity in time series data is defined as variation in the mean or variance over time, induced by the changing effect of treatment or other covariates over time.
Proper missing data imputation for psychiatric research must account for participant heterogeneity and non-stationarity simultaneously.

Techniques for handling missing data have been expanded in recent decades with advancement in statistics, economics, and computer science. 
However, available approaches do not address our issue of simultaneous missing data in  response variable (dependent variable) and explanatory variables (independent variables) in non-stationary time series. 
While complete case analysis is frequently used to eliminate incomplete records, it reduces estimation efficiency, especially for data with high missing rate \citep{ibrahim2005missing,white2010bias,little2019statistical}; for time series data, it additionally breaks the temporal relationships among variables and may introduce further bias  \citep{bashir2018handling}. 
Commonly used imputation methods for cross-sectional and longitudinal data include mean imputation, last-observation-carried-forward (LOCF) imputation, and linear and spline interpolation. A considerable body of literature has demonstrated that these methods routinely produce biased estimation even under missing completely at random (MCAR) and should be used with caution \citep{honaker2010missing}.
Multiple imputation has been increasingly recommended for handling missing data 
and serves as the benchmark against which new methods are being evaluated \citep{twisk2002attrition,spratt2010strategies}. 
However, multiple imputation relies on independently and identically distributed (i.i.d) subjects and time-invariant models for  imputation, rendering it unsuitable for non-stationary time series centered around a single subject. 
Another class of methods for missing data that assumes a static process is weighted estimation equations (WEEs), which is undesirable for non-stationary time series \citep{ibrahim2005missing}. Furthermore, extensive simulations have cautioned against using WEEs in situations when propensity scores (or inverse weights) are close to zero \citep{ibrahim2005missing}, as is common in multivariate time series.
Various moving average techniques have been proposed for missing data in univariate time series; however, they fail to capture the relationships between multiple time series and thus are inappropriate for multivariate time series.
Recent progress in computer science introduced several machine-learning approaches for missing data imputation in multivariate time series (e.g., current neural networks \citep{bashir2018handling}, generative adversarial networks \citep{luo2018multivariate}). These approaches demand stationary data generation processes and a large number of i.i.d participants for training, and thus are not suitable for time series data centered around a single subject.

Model-based maximum likelihood (ML) methods provide more principled approaches to deal with missing data under missing completely at random (MCAR), missing at random (MAR) and missing not at random (MNAR) \citep{ibrahim2005missing,fitzmaurice2008longitudinal,fitzmaurice2012applied,little2019statistical,follmann1995approximate,little1995modeling,hogan1997mixture,ibrahim1999missing,ibrahim2001missing,ibrahim2009missing,sinha2012robust,fitzmaurice2000generalized,troxel1998analysis}, and are widely applied to longitudinal studies using (generalized) linear mixed-effects models \citep{follmann1995approximate,little1995modeling,hogan1997mixture,ibrahim1999missing,ibrahim2001missing,ibrahim2009missing,sinha2012robust} or generalized estimation equations \citep{fitzmaurice2000generalized,troxel1998analysis}. 
However, their contributions are limited to  static processes or stationary time series. 

A specialized model for potentially non-stationary multi-variate time series is the state space model (also known as dynamic linear model \citep{aoki2013state}), which has been widely applied in engineering, economics, statistics and many other fields \citep{schmidt1966application,harvey1990forecasting,scharf1991statistical,jones1993longitudinal,aoki2013state} to estimate time-varying parameters in dynamic systems, facilitated by a powerful estimation tool -- Kalman filter and smoothing algorithm \citep{kalman1960new,kalman1961new}.
Even though state space model tolerates missing values in response variable, it does not permit missing values in explanatory variables \citep{cipra1997kalman,petrics2009}: missing data in explanatory variables must either be imputed or be discarded as in complete case analysis.
Only a few missing data imputation approaches have been explored for multivariate time series using state space model. 
\citet{bashir2018handling} proposed a state-space-model formatted vector autoregressive model to impute missing values only in the response variable.
\citet{naranjo2013extending} proposed a state-space model algorithm to handle missing values in exogenous predictors that affect the response but not the other way around. 
Neither method addresses the commonly encountered missing data scenario in mobile device data --  missingness in outcome variable and explanatory variables induced by including lagged outcomes as covariates. 

Comprehensive assessments of existing imputation methods and the development of novel imputation strategies are in urgent need for missing data imputation in potentially non-stationary multivariate time series. N-of-1 study design is a patient-centered research design that is well suited for personalized mobile device studies or mental health research with high patient heterogeneity \citep{kumar2013mobile}. 
We consider an N-of-1 study design, where a single subject constitutes the entire study with recurrent interventions or exposures observed for the same subject.
We propose a model-based maximum likelihood strategy that combines state space modeling and multiple imputation, named as ``SSMmp,'' and its computationally improved version ``SSMimpute,'' in order to impute missing lagged outcomes used as regressors and to estimate quantities of interest in a non-stationary multivariate time series.
We prove the theoretical properties of the proposed method SSMmp in terms of producing consistent coefficient estimates for post-imputation analysis and evaluate their empirical performance in extensive simulations of both stationary and non-stationary time series.
Comprehensive comparisons to commonly used strategies for addressing missing data in terms of bias, estimation error, and coverage are conducted. 

The structure of the paper is as follows. Section 2 introduces notation of multivariate time series and the associated pattern of missing data. 
The state space model, as well as the Kalman filter and smoothing methods, are introduced in Section 3. Section 4 introduces the proposed SSMmp method and its theoretical foundations, followed by a more computationally efficient variant -- SSMimpute. Section 5 presents performance evaluation of the proposed method and compares it to other commonly used methods for handling missing data,  using simulations of both stationary and non-stationary time series, subject to various missing mechanisms. 
In section 6, we employ SSMimpute to investigate the association between phone-based social network size and patients' negative mood using a multi-year observational smartphone study of bipolar and schizophrenia patients. Novel insights on the time-varying association between network size of text messages and negative mood are presented. Section 7 concludes with discussions.

\section{Notations and missing data pattern}
We consider N-of-1 study where a single subject is followed over time at $t=1,2,3,\ldots,T$. At time $t$, denote the outcome as $Y_t$, the exposure(s) or treatment(s) as $A_t$, and other covariate(s) as $\C_t$. The observation series of $(Y_t,A_t,\C_t)$ at $t=1,2,3,\ldots,T$ constitute a multivariate time series. 

Missing data is prevalent in actively collected data (e.g., moods or psychatric symptoms requested in surveys), as they need participants' engagement; missing data may also be significant in passively collected data (e.g., telecommunication data, accelerometer data). 
We consider a patient with Bipolar I disorder who has been followed for 708 days in the BLS study as an example. The missing rate for the outcome -- actively collected daily negative mood -- is $23.31\%$, while the missing rate for the passively collected exposures (communication logs) and covariates (accelerometer data and temperature) is $0\%$.  
This exemplifies a simple, but common and critical, missing data scenario: missing data occurs only in the outcome $Y_t$, and exposures and covariates ($A_t$, $\C_t$) are fully observed for the entire follow-up.

We denote the q-lagged value of outcome $Y_t$ as $Y_{t-q}$ and the q-lagged values of the entire outcome time series $\Y=(Y_1,\dots,Y_t)$ as $\mathbf{Y_{t-q}}$ for $q>0$. Define a missingness indicator $M_t$ for the outcome $Y_t$ with $m_t=1 $ if $Y_t$ is missing and $m_t=0$ if $Y_t$ is observed. Denote the collection of time points at which the observation of outcome $Y_t$ is missing as $\mathbf{T}_{\text{mis}}=\{t: m_t=1,t=1,\ldots,T\}$. We can partition the entire outcome time series into observed outcomes $\mathbf{Y}_{\text{obs}} = \{Y_t: m_t=0 \text{ or } t \notin \mathbf{T}_{\text{mis}} \}$ and missing outcomes $\mathbf{Y}_{\text{mis}}=\{Y_t: m_t=1 \text{ or } t \in \mathbf{T}_{\text{mis}} \}$. 
As lagged outcomes are commonly included in time series analysis, missingness in outcome naturally results in missingness in the lagged outcomes as regressors: when $Y_t$ is missing at time points $t \in \mathbf{T}_{\text{mis}}=\{t: m_t=1, t=1,\ldots, T\}$, lagged outcome $Y_{t-q}$ is correspondingly missing at  $t+q \in \mathbf{T}_{\text{mis}}$.
Figure~\ref{fig:missing_pattern} shows this non-monotone missing data pattern when lagged outcomes with a time lag of 1 to 3 are included in the explanatory variables for analysis. 
In general, the more lagged outcomes are included, the higher the overall missing rate is induced for analysis.

\begin{figure}
\centering
\includegraphics[width=0.65\linewidth]{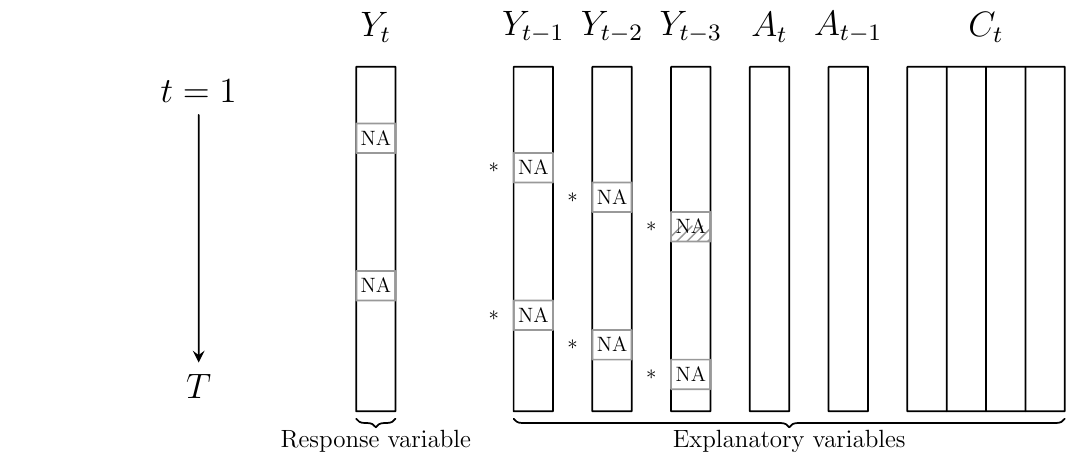}
\caption{Non-monotone missing data pattern in both response and explanatory variables for time series analysis caused by missing outcomes. $Y_t$ denotes outcome at time $t$, whereas ($Y_{t-1}$, $Y_{t-2}$, $Y_{t-3}$) denote three lagged outcomes included for the time series analysis. $A_t$ and $A_{t-1}$ denote exposures at time $t$ and $t-1$. All other covariates included for analysis are included in $\C_t$.}
\label{fig:missing_pattern}
\end{figure}

\section{State space model, Kalman filter and smoothing}

The state space model has been recognized as one of the most noticeable innovations in engineering \citep{gannot1998iterative,kiencke2000automotive}, econometrics \citep{lai2013adaptive}, mathematical finance \citep{manoliu2002energy}, neural networks \citep{haykin2004kalman}, and time series analysis \citep{harvey1990forecasting,scharf1991statistical}. 
Specifically, state space model provides an estimation framework for time-varying parameters in non-stationary time series via the powerful tools of the Kalman filter and smoothing \citep{kalman1960new,kalman1961new}.
State space model formulates time series $Y_t$, $t=1,2,\ldots$, as observations of a dynamic system's output up to additive Gaussian random noise. 
The ``observational equation'' describes the dependence of this time series on a latent process of hidden states, the evolution of which is described by another ``state equation.'' Denote $\theta_t$ as the $d \times 1$ latent state vector at $t=1,2,\ldots$, and assume it to be a Markov process such that $\theta_t \indep \theta_s | \theta_{t-1}$ for $\{\theta_s: s<t\}$.
\begin{defn}[Linear state space model] For $t=1,2,\ldots$, the state equation of linear state space model is
\begin{equation}
\theta_t = G_t \theta_{t-1} + w_t, \quad w_t \sim N_d(0,W_t)
\label{eq:state}	
\end{equation}
where $\theta_{t}$ denotes the $d \times 1$ state vector, $G_t$ is the $d \times d$ state transition matrix, and $w_t$ represents the $d \times 1$ independently and identically distributed noise vector, following distribution  $N_d(0,W_t)$. The observational equation of linear state space model is
\begin{equation}
Y_t=F_t \theta_t + v_t,  \quad v_t \sim N_n(0,V_t)
\label{eq:obs}
\end{equation}
where $Y_t$ is the $n \times 1$ vector of  outcomes, $F_t$ is the $n \times d $ observational matrix, and $v_t$ is the i.i.d observational noise vector, following distribution $N_n(0, V_t)$.
\end{defn}

\begin{figure}
\centering
\includegraphics[width=7cm]{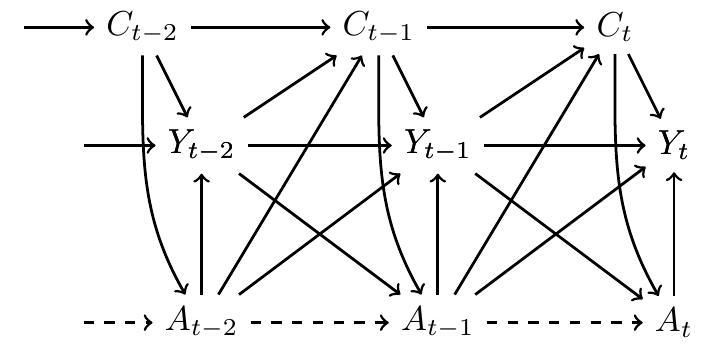}
\caption{Causal relationship of exposures $(A_t, A_{t-1},\ldots,)$, outcomes $(Y_t, Y_{t-1},\ldots,)$, and covariates $(C_t, C_{t-1},\ldots,)$ in entangled multivariate time series.}
\label{fig:dag1}	
\end{figure}
In the BLS study, the outcome variable is a one-dimensional ($n=1$) self-reported summary score for negative mood. 
The $F_t$ matrix may contain lagged outcomes and current and previous exposures and covariates, and hidden states $\theta_t$ represent their unknown regression coefficients at time $t$.
For example, consider a linear relationship of outcome $Y_t$ on its previous value $Y_{t-1}$, current and previous exposures $(A_t, A_{t-1})$, and current covariates $C_t$ as illustrated in the causal diagram Figure~\ref{fig:dag1}. Then the correctly specified statistical model (or the data generation process) for $Y_t$ is 
\begin{equation}
\begin{split}
	Y_t &=  \beta_{0,t} + \rho_t Y_{t-1} + \beta_{1,t} A_t +\beta_{2,t} A_{t-1} + \beta_{c,t} C_t + v_t	 \\
\end{split}
\label{eq:lm}
\end{equation}
which can be reformatted into a state space model representation $Y_t = F_t \theta_t +  v_t$, where $F_t =(1,Y_{t-1},A_t,A_{t-1},C_t)$, $\theta_t=(\beta_{0,t},\rho_t,\beta_{1,t},\beta_{2,t},\beta_{c,t})^T$, and $v_t \sim N(0,V_t)$. Note that coefficients $\theta_t$ and variance $V_t$ may change over time.

A stationary or static process happens when $\theta_t=\theta$ is time-invariant; if $A_t$ and $C_t$ are stationary time series, then the resulting time series $Y_t$ is also stationary. 
On the other hand, $Y_t$ may be generated by a non-stationary or dynamic process. Typical non-stationary examples occur when certain components of $\theta_t$ are time-varying (e.g., random walk, AR(p) process, or periodic stable) or when variance varies over time. 
Time-varying scenarios are prevalent and particularly suited for mental health research. 
As an illustration, for a complex disease (such as mental health) whose influencing factors are poorly understood, we may allow its baseline level (intercept component in $\theta_t$) to fluctuate as a random walk to capture the irregular gradual drift due to unknown factors; when treatment effect varies in response to other major life changes (e.g, a change of disease severity,  a move, a new job), we may model treatment effect as a periodic-stable process, which is time-invariant within periods but takes on different values across periods; when we switch to devices with improved measurement precision, we may model noise variance as also being time-varying. The state space model provides a flexible mathematical description of both stationary (or static) and non-stationary (or dynamic) time series, and its specification can be easily modified to accommodate more complex relationships than those illustrated in Figure~\ref{fig:dag1}. 
The Kalman filter and smoothing algorithm \citep{kalman1960new,kalman1961new} provides an optimal estimator and efficient estimation algorithms for the posterior distribution of coefficients $\theta_t|y_{1:t} \sim N(m_t, C_t)$ (given observations up to time $t$) and $\theta_t|y_{1:T} \sim N(s_t,S_t)$ (given all observations) for state space model \citep{anderson1983adaptive,ljung1990adaptation,guo1990estimating,guo1995exponential,cao2003exponential,welch2020kalman}.
\begin{defn}[Kalman Filter] Consider a dynamic linear model specified in \eqref{eq:state} and \eqref{eq:obs} with initial condition $\theta_0 \sim N(m_0,C_0)$. Let \[\theta_{t-1}|y_{1:t-1} \sim N(m_{t-1}, C_{t-1})\]
Then the following statements hold,
\begin{equation}
\begin{split}
	m_t &= \E[\theta_t|y_{1:t}] = G_t m_{t-1} + R_tF'_t Q_t^{-1}(Y_t - F_tG_tm_{t-1}) \\
	C_t &= \Var [\theta_t|y_{1:t}] = R_t-R_{t}F'_tQ_t^{-1}F_t R_t	\\
\end{split}	
\label{eq:filter}
\end{equation}
where $R_t=\text{Var}(\theta_{t}|y_{1:t})=G_tC_{t-1}G_t'+W_t$ and $Q_t=Var(Y_t|y_{1:t-1})=F_tR_tF'_t+V_t$.
\end{defn}
\begin{defn}[Kalman Smoothing]
Consider a dynamic linear model specified in \eqref{eq:state} and \eqref{eq:obs}. If $\theta_{t+1}|y_{1:T} \sim N(s_{t+1}, S_{t+1})$, then $\theta_t|y_{1:T} \sim N(s_{t}, S_{t})$, where
\begin{equation*}
\begin{split}
	s_t &= \E[\theta_t|y_{1:T}] = m_t + C_tG'_{t+1} R_{t+1}^{-1}(s_{t+1} - G_{t+1}m_{t}) \\
	S_t &= \Var [\theta_t|y_{1:T}] = C_t-C_{t}G'_{t+1}R_{t+1}^{-1}(R_{t+1}-S_{t+1})R_{t+1}^{-1}G_{t+1}C_t.	\\
\end{split}	
\end{equation*}
\end{defn}
Kalman filter and smoothing accommodate missing data in the response variable $Y_t$ by setting certain components in the recursion to $0$ or $\infty$, so that missing data in the response variable is intelligently skipped over and no information is updated for the posterior distribution of $\theta_t$.
For example, for a one-dimensional observation outcome $y_t$ missing at time $t$, we set $V_t=\infty$ or $F_t=0$ in Kalman Filter, so that the gain matrix $R_tF'_tQ_t^{-1}=0$, and consequently, $\theta_t|y_{1:t}=\theta_t|y_{1:t-1}$. 
However, Kalman filter and smoothing do not allow missing data in explanatory variables (equivalently, in $F_t$). 
A new imputation strategy is needed to account for missing data in explanatory variables caused by missing lagged outcomes in non-stationary time series.
\section{State space model multiple imputation for non-stationary time series }

By combining the concepts of likelihood-based missing data methods, multiple imputation, and dynamic modeling of time-varying systems, we developed a state-space-model-based multiple imputation algorithm (SSMmp) for addressing missing data in both the response variable and explanatory variables when missingness is restricted to the outcome in potentially non-stationary multivariate time series. 
Missing data can be generated under MCAR, in which there is no relationship between the missingness of the data and any values (observed or missing); under MAR, when the probability of being missing depends on observed data; and MNAR, when the probability of being missing depends on unobserved information. 
The missing mechanisms MCAR and MAR are referred to as being ignorable, as the modeling under which can be factored out and parameter inference can be derived solely from the ignorable likelihood with observed data \citep{ibrahim2009missing,little2019statistical}.
 
The SSMmp algorithm augments the ignorable likelihood with additional information from imputing missing lagged outcomes in regressors. 
The SSMmp is applicable to MCAR and MAR, with the potential to be extended to MNAR if modeling of the missing mechanism is also included in the likelihood.
Specifically, the SSMmp  algorithm employs an iterative procedure of ``multiple imputation'' and ``maximization,'' which has a theoretical basis similar to the Monte-Carlo Expectation Maximization (MCEM) used by a large number of likelihood-based methods in dealing with missing data \citep{ibrahim2005missing,little2019statistical}. SSMmp is applicable to both stationary and non-stationary time series. When the true   statistical model (or data generation process) is stationary, the SSMmp algorithm reduces to the multiple imputation technique used in longitudinal studies, which has been shown to be unbiased and more efficient in the case of missing covariates \citep{von20074,spratt2010strategies,white2010bias}. On the other hand, when the true statistical model (or data generation process) is dynamic and the resulting multivariate time series is non-stationary, the SSMmp algorithm is capable of handling complex non-stationary systems, when parameters of coefficients may follow random walk, AR(p) process, and periodic-stable process. 
Critical change points are also identified for periodic-stable parameters. 
This capability of identifying change points is particularly useful for psychiatric research in understanding how effects of exposures change over time in response to changes in patients' life events (e.g., a move or a new job, medication change, life habit change).


\subsection{State space model multiple imputation (SSMmp) algorithm}
\label{section:SSMmp}
Consider a general situation in which the outcome $Y_t$ at time $t$ depends on previous outcomes $(Y_{t-1},\ldots, Y_{t-q})$ with $q \ge 1$, current and previous exposures $(A_{t},\ldots,A_{t-p})$ with $p \ge 0$, and current and previous covariates $\C_t=(C_t,\ldots,C_{t-o})$ with $o \ge 0$. 
Then, a statistical model (or data generation process) for $Y_t$ could be specified as,
\begin{equation*}
\begin{split}
	Y_t &=  \beta_{0,t} + \rho_{1,t} Y_{t-1} + \ldots + \rho_q Y_{t-q} + \beta_{1,t} A_t +\ldots + \beta_{p,t} A_{t-p} + \beta_{c,t} \C_t + v_t\\
\end{split}
\end{equation*}
which can be reformatted as a state space model as 
\begin{equation*}
\begin{split}
	Y_t & = F_t \theta_t +  v_t
\end{split}
\end{equation*}
where $F_t = (1, Y_{t-1}, \ldots, Y_{t-q}, A_t ,\ldots, A_{t-p}, \C_t)'$ contains all information of lagged outcomes, current and lagged exposures and covariates at time $t$, and $\theta_t = (\beta_{0,t},\rho_{1,t}, \ldots ,\rho_{q,t},\beta_{1,t}, \ldots,\beta_{p,t}.\beta_{c,t})$ represents the corresponding unknown vector of coefficients. Parameters $\theta_t$ to be estimated can change over time as specified in 
\begin{equation*}	
\theta_t  = G_t \theta_{t-1} + w_t
\end{equation*}

Consider a scenario in which we have missing data for outcomes but not for exposures or covariates, as shown in Figure~\ref{fig:missing_pattern}. Denote the observed current and lagged values of exposures and covariates in the explanatory variables at $t$ as $\X_t=(A_t,\ldots,A_{t-q},\C_t)$, and denote their values at all times as $\X=(\X_1,\ldots,\X_T)$.
We propose the following iterative algorithm for dealing with missing outcomes in non-stationary multivariate time series. Each iteration includes a step of ``multiple imputation'' and a step of ``maximization''.
Similar to Bayesian multiple imputations, ``multiple imputation'' step generates multiple random samples for missing outcomes $\tilde{\mathbf{Y}}_{\text{mis}}=F_t \tilde{\theta}(t)+\tilde{v}_t$, where $\tilde{\theta}(t)$ are random draws from their estimated posterior distribution given the data and other nuisance parameters and $\tilde{v}_t \sim N(0, \hat{V_t})$; then substitutes these sampled missing outcomes into their corresponding lagged values of outcomes $(\mathbf{Y_{t-q}},\ldots,\mathbf{Y_{t-1}})$ in order to eliminate missing data in explanatory variables.
When missing lagged outcomes are substituted for and nuisance parameters of the state space model are taken from an initial specification or estimation from the previous iteration, the ``maximization'' step fits observed outcomes with completed explanatory variables in order to estimate time-varying coefficients. 
Simultaneously, information about the structure of state space models and other nuisance parameters is updated.
Rubin's Rule is the standard approach for analyzing multiply imputed datasets by combining within-imputation variance and between-imputation variance to correct the under-estimated variance of parameters \citep{barnard1999miscellanea}. We apply the Rubin's Rule to combine the results of multiple estimations performed on multiple copies of imputed data.
Below we provide a detailed description of the SSMmp algorithm.
\begin{Steps}
\item(Initialization) Make initial guess on the missing outcomes $\hat{\mathbf{Y}}_{\text{mis}}^{(0)}$ and other nuisance structural parameters of the state space model. Substitute guessed values $\hat{\mathbf{Y}}_{\text{mis}}^{(0)}$ into their correspondingly missing lagged values $(\mathbf{\hat{Y}_{t-q}}^{(0)},\ldots,\mathbf{\hat{Y}_{t-1}}^{(0)})$ to eliminate missing data in explanatory variables.
\item(Maximization step) For the $k$th iteration, $k=1,2,\ldots$, apply state space model to incomplete outcomes $\mathbf{Y}_{\text{obs}}$ and completed explanatory variables $(\mathbf{\hat{Y}_{t-q}}^{(k-1)},\ldots,\mathbf{\hat{Y}_{t-1}}^{(k-1)},\X)$, to obtain the maximum likelihood estimate (MLE) of the structural parameters of the state space model and the filtered or smoothed estimates of coefficients $\{\tilde{\theta}^{(k)}_t\}_{t=1,\ldots,T}$.
\item(Multiple imputation step) Get $r$ random draws of the coefficients, $\{\tilde{\theta}^{(k,j)}\}_{t=1,\ldots,T}$, $j=1,\ldots,r$, from the estimated posterior distribution $\{\tilde{\theta}^{(k)}_t\}_{t=1,\ldots,T}$. Get $r$ random draws of noise, $\tilde{v}_t^{(k,j)} \sim N(0,\hat{V}_t)$. 
For $j=1,\ldots,r$, calculate missing outcomes $\tilde{y}_t^{(k,j)}=F_t\tilde{\theta}^{(k,j)}_t+\tilde{v}_t^{(k,j)}$ for $t \in \mathbf{T}_{\text{mis}}$, substitute those guessed values into their corresponding lagged variables $(\mathbf{\tilde{Y}_{t-q}}^{(k,j)},\ldots,\mathbf{\tilde{Y}_{t-1}}^{(k,j)})$, refit the state space model, and finally obtain a new estimation for coefficients $\{\dbtilde{\theta}^{(k,j)}_t\}_{t=1,\ldots,T}$. Apply Rubin's rule to combine the $r$ newly estimated $\{\dbtilde{\theta}^{(k,j)}_t\}_{t=1,\ldots,T}$, $j=1,\ldots,r$ to be the final estimation of coefficients in the $k$th iteration, $\{\hat{\theta}^{(k)}_t\}_{t=1,\ldots,T}$.
\item (Update) Based on the $r$ groups of guessed missing outcome $\tilde{y}_t^{(k,j)}$ for $t \in \mathbf{T}_{\text{mis}}$, $j=1,\ldots,r$, update new guesses of the missing outcome $\hat{y}_t^{(k)}=\frac{1}{r} \sum_{j=1}^r \tilde{y}_t^{(k,j)}$ for $t \in \mathbf{T}_{\text{mis}}$. Substitute updated guesses for missing outcomes into  the explanatory variables $(\mathbf{\hat{Y}_{t-q}}^{(k)},\ldots,\mathbf{\hat{Y}_{t-1}}^{(k)})$ for their corresponding lagged values.
\item (Check convergence) Repeat steps 2-4 until convergence is achieved for likelihood, coefficient estimates, and other structural parameters of the state space model.
\end{Steps}
The SSMmp algorithm replicates the idea of a long-established approach for handling missing data in longitudinal settings: replace missing data with estimated values, estimate parameters of interest using completed data, and then repeat the previous two procedures by re-estimating missing data assuming newly-estimated parameters are correct and by re-estimating parameters of interest based on updated missing data, and so forth until convergence. 
The SSMmp algorithm also incorporates the idea of Bayesian multiple imputation to account for imputation uncertainty and to ensure adequate coverage for  confidence interval. 
Additionally, no constraints on the structure of the state space model used for imputation are imposed: parameters can be modeled as a random walk, an AR(p) process, or as periodic-stable or time-invariant; the ultimate structure of state space model, including unknown parameters structuring time-varying components (e.g, changepoints of period-stable process, or auto-correlation terms in AR(p) process) is gradually revealed during the iterative fitting process of SSMmp and stabilized upon convergence. Convergence is achieved when likelihood, estimated coefficients, and other structural parameters describing the time-varying behaviors of the system differ by less than a pre-specified criteria between iterations. A flowchart for the SSMmp algorithm is shown in the Appendix. The Appendix also contains an illustration of the convergence process for time-varying coefficients in the BLS application. Notably, we do not impute both missing outcomes in the response variable and missing lagged outcomes in explanatory variables simultaneously, this results in significantly biased parameter estimations \cite{sullivan2015bias,little2019statistical}.

\subsection{Statistical properties of Kalman filter, smoothing and SSMmp algorithm}
\label{section:statistical_properties}

It has long been established that the Kalman filter/smoothing estimator provides the optimal estimate for the hidden states $\theta_t$, if (i) the regression vector $F_t$ belongs to the $\sigma$-algebra $\mathcal{F}_t$, which contains outcomes, exposures, and covariates by time $t$ in our setting, and (ii) the noise $\{w_t,v_t\}$ for the state vector $\theta_t$ and outcome $y_t$ are Gaussian white noise, then $\hat{\theta}_t$ generated by \eqref{eq:filter} is the optimal estimate for $\theta_t$ in the sense that
\[
\hat{\theta_t}=\E[\theta_t|\mathcal{F}_t]
\]
and minimize the conditional expectation 
\[
\E[(\hat{\theta_t}- \theta_t)(\hat{\theta_t}- \theta_t)'|\mathcal{F}_t]
\]
with exponential convergence rate \citep{anderson1983adaptive,ljung1990adaptation,guo1990estimating,guo1995exponential,cao2003exponential}. The exponential convergence of Kalman Filter is of fundamental importance for its ability of tracking the (possibly time-varying) unknown parameters.
To be specific, when the unknown parameter is fixed in time as $\theta_t=\theta$, the Kalman filter estimator is an unbiased estimator with exponential convergence rate, along with its estimation covariance decreasing at rate of $O(1/t)$; when the unknown parameter vector is allowed to be time-varying (e.g., a random walk, a stationary process, or a bounded deterministic sequence), the estimate is exponentially convergent with error bounded if certain stochastic excitation conditions \citep{guo1990estimating,guo1995exponential} are satisfied for $F_t$.

We prove the reduced bias and improved estimation efficiency when missing lagged outcomes in the explanatory variables are imputed versus complete case analysis when they are not imputed and directly removed. 
Consider the scenario when outcome is observed at time $t$ but not at $t-1$. As a result, at time $t$, the response variable $Y_t$ is observed, but the lagged outcome $Y_{t-1}$ used in the explanatory variables is missing. 
We either need to impute the missing lagged outcomes or omit the entire incomplete observation for the state space model analysis.

First, consider the time-invariant case with $\theta_t=\dot{\theta}$, so that the resulting multivariate time series are stationary under constant variance. 
Complete case analysis discards the observation with missing data at $t$ and proceeds directly to $t+1$ for analysis, and thereby, modifies the original temporal relationship and generates a new timeline, in which the $t+1^\text{th}$ time point on the original timeline is substituted by $t^*=t+1$ on the new timeline.
The posterior distribution of the estimated coefficients at time $t^*$ using Kalman filter are denoted as $\hat{\theta}_{t^*} \sim N(\hat{m}_{t^*},\hat{W}_{t^*})$. 
Alternatively, we can maintain the original timeline and impute missing value of $Y_t$ using the proposed state space model multiple imputation method and include it as a lagged outcome for coefficient estimation at $t+1$, denoted as $\hat{\theta}_{t+1} \sim N(\hat{m}_{t+1},\hat{W}_{t+1})$.
\begin{prop}
When $\theta_t=\dot{\theta}$ for $t=1,\ldots, T$ is time-invariant and outcome $Y_t$ is missing at time $t$, the Kalman filter estimates at time $t+1$ from complete case analysis $\hat{\theta}_{t^*}$ (with $t^*=t+1$ on the modified timeline) and from the SSMmp strategy $\hat{\theta}_{t+1}$ are unbiased so that $\E[\hat{\theta}_{t^*}]=\E[\hat{\theta}_{t+1}]=\dot{\theta}$, and $\hat{W}_{t+1} \preceq \hat{W}_{t^*}$.
\label{prop:increasedefficiency}
\end{prop}
In short, when compared to complete case analysis, Kalman filter with missing data imputed manages to provide the same unbiased estimate for the unknown parameter, while increasing estimation efficiency by incorporating more observations. The proof of Proposition~\ref{prop:increasedefficiency} is shown in the Appendix. 
The rationale for the increased efficiency is comparable to the reason for the increased efficiency using multiple imputation to impute missing covariates in longitudinal studies: SSMmp tackles a type of incomplete data records -- those with missing values for lagged outcomes but observed values of response outcome, exposures and covariates (marked by * in Figure~\ref{fig:missing_pattern}), preserving partial information about the correlation between the response variable and exposures/covariates. This increased estimation efficiency can be particularly beneficial in scenarios with high missing rates.

Next, we consider the time-varying case, in which certain parameters vary over time and can be modeled as a random walk, an AR(p) process, or a periodic-stable process.
\begin{prop}
For time-varying $\theta_t$, the Kalman filter estimates at time $t+1$ from complete case analysis $\hat{\theta}_{t^*}$ (with $t^*=t+1$ on the modified timeline) is biased so that $\E[\hat{\theta}_{t^*}] \neq \theta_{t+1}$; the Kalman filter estimate from the SSMmp algorithm $\hat{\theta}_{t+1}$ is unbiased so that $\E[\hat{\theta}_{t+1}]=\theta_{t+1}$.
\label{prop:reducedbias}
\end{prop}
Complete case analysis has long been regarded as a safety net for dealing with missing data and has been demonstrated to be the optimal choice under a variety of circumstances for cross-sectional and longitudinal data \citep{von20074,spratt2010strategies,white2010bias}. 
In the case of time series, complete case analysis bypasses time points with incomplete record and splices time points with complete data for analysis. 
However, for dynamic systems or non-stationary time series, it is critical to restore the proper temporal relationship among variables; otherwise, it may introduce bias into parameter estimation. For example, if $Y_{t-1}$ is autocorrelated with its previous value $Y_{t-1}$ by a factor of 0.5; omitting $Y_{t-1}$ and splicing time points $t-2$ and $t$ together (correlation now is $0.25$) will bias the estimation of the autocorrelation. The proposed SSMmp algorithm restores the proper temporal structure between variables by imputing missing data properly for non-stationary system.
Proposition~\ref{prop:reducedbias} proved the presence of bias of using complete case analysis and the absence of bias when using the proposed SSMmp. The proof of Proposition~\ref{prop:reducedbias} is shown in the Appendix.

In addition to removing potential bias, SSMmp typically provides a more efficient estimation with a smaller variance, compared to the complete case analysis, due to the utilization of additional data. Nevertheless, this enhanced efficiency is not guaranteed, unlike the stationary scenario. 
For instance, Kalman filter monitors the true value of a random walk process, but its uncertainty does not decrease over time because the true value is continuously changing. Without decreasing variance over time, substantial variance in the posterior distribution of estimated parameters plus added noise in multiple imputation may cancel out the effect of increased sample size.

\subsection{``Computationally efficient'' SSMimpute algorithm}
\label{section:SSMimpute}

Despite SSMmp's favorable statistical properties, it can be computationally intensive, unstable, and exhibit convergence issues in practice. SSMmp uses estimation results merged from several random draws of missing outcomes in each iteration. Increasing the number of random draws improves stability but adds to the computational burden. 
While differences over iterations tend to decrease as approaching the maximization of the ignorable likelihood, there is no guarantee of convergence in practice due to the random nature of sample selection. 
We propose an alternative version of the state space model multiple imputation algorithm ``SSMimpute.'' 
This alternative approach iterates between model fitting and substituting until maximization is reached and then applies multiple imputation and Rubin's Rule for proper variance of the parameters.
We include the detailed algorithm for the ``SSMimpute'' in the Appendix. When compared to the SSMmp algorithm, the SSMimpute algorithm produces similar estimation results (up to differences due to randomness in multiple imputation), but is significantly more stable and converges faster. 
\section{Simulation}
\label{sec:simulation}

\subsection{Simulation setup}
We evaluate the performance of the SSMmp and SSMimpute algorithms using simulations of both stationary and non-stationary multivariate time series of 1000 time points, emulating the typical follow-up period of 2-4 years in the Bipolar Longitudinal Study (BLS). 
Following temporal causal relationship assumed in Figure~\ref{fig:dag1}, we consider that the outcome $Y_t$ is auto-correlated with its previous value $Y_{t-1}$ and depends on exposures on the same day (to represent a contemporaneous or short-term association) and the previous day (to represent a lagged or long-term association), in the presence of other confounders in $C_t$, as shown in \eqref{eq:dgp}. 
\begin{equation}
	Y_t = \beta_{0,t} + \rho_t Y_{t-1} + \beta_{1,t} A_t +\beta_{2,t} A_{t-1} + \beta_{c,t} C_t + v_t \text{, } \quad v_t \sim N(0,0.1)
	\label{eq:dgp}
\end{equation}
We consider two general  scenarios: 1) the stationary scenario, in which all covariate coefficients are fixed over time as $\beta_{0,t}=40$, $\rho_t=0.5$, $\beta_{1,t}=-1.5$, $\beta_{2,t}=-0.5$, and $\beta_{c,t}=-1$, so that the resulting outcome time series is stationary when $A_t$ and $C_t$ are stationary, and 2) a non-stationary scenario in which certain coefficients are time-varying, including a random walk intercept $\beta_{0,t} = 40 + \beta_{0,t-1} + w_{t}$, $w_t \sim N(0,1)$, three-pieced coefficient $\beta_{1,t}$ as
\begin{equation}
	\beta_{1,t}= \begin{cases} 
      -1 & t \leq 400 \\
      -2 & 400 < t \leq 700\\
      -1 & x > 700
   \end{cases}
\label{eq:3periodsA}
\end{equation}
and other time-invariant coefficients $\rho_t=0.5$, $\beta_{2,t}=-0.5$, and $\beta_{c,t}=-1$, so that the resulting outcome time series is non-stationary. 
Note that time-varying behaviors of coefficients are chosen to resemble those discovered in the BLS data application, and the magnitude of these coefficients are likewise chosen so that generated outcome, exposure and covariates match to those in the BLS application.

For each scenario, 500 simulations are generated to evaluate the performance of the SSMmp and SSMimpute algorithms for the estimation of coefficients, and to compare them to other widely used missing data strategies.
Widely used missing data strategies include complete case analysis (``cc''), mean imputation (``mean''), last-observation-carried-forward imputation (``locf''), linear (``linear'') and spline (``spline'') interpolation, multiple imputation (``mp''), and imputation using best-fitted ARIMA models (``ARIMA'').
The Multiple Imputation by Chained Equations (MICE) is implemented for the multiple imputation strategy (``mp''), using current and 1-lagged values of outcome, exposure, and covariates. 
Missing outcomes are imputed first using above strategies, then we apply both linear regression model (true model for the stationary scenario, denoted by the preceding abbreviation) and state space model (true model for the non-stationary scenario, denoted by similar abbreviations preceded by ``SSM\_'') for statistical inference.
We investigate imputation strategy performances under various missing mechanisms of MCAR, MAR, and MNAR and under varying missing rates of $25\%$, $50\%$, and $75\%$.

\subsection{Simulation result and comparisons to other missing data methods}

For the stationary case, simulation results confirm literature's conclusions  regarding performances of widely-used imputation approaches \citep{von20074,spratt2010strategies,white2010bias}.
Figure~\ref{fig:simulation1.1} illustrates the estimation result for $\beta_{2,t}$ (coefficient of $A_{t-1}$) at $t=1000$ in terms of bias, standard error, and coverage, over $500$ simulations using various missing data imputation strategies for missing rate of $50\%$ under MCAR, MAR, and MNAR. 
Similar estimation comparison plots for other coefficients are shown in the Appendix.
Simulation results confirm that complete case analysis and multiple imputation provide unbiased estimation and adequate coverage under the true model specification, whereas multiple imputation provides slightly increased estimation efficiency by incorporating additional observations. 
SSMmp and SSMimpute strategies provide unbiased estimation using the (in this case unnecessarily more complex) state space model for analysis, and provide increased estimation efficiency if compared to complete case analysis, as proved by Proposition~\ref{prop:increasedefficiency}.
All other imputation strategies introduce significant bias in the coefficient estimation, providing clear warnings of their use for time series data.  
It is worth noticing that the commonly used LOCF imputation and linear interpolation strategies introduces significant bias in estimation and almost $0\%$ coverage for the nominal $90\%$ confidence interval, which confirms the warning against their use in both longitudinal and time series analysis. 

For the non-stationary case, Figures~\ref{fig:simulation2.3} shows estimation results for the time-varying $\beta_{1,t}$ (coefficient of $A_t$, as in Equation~\eqref{eq:3periodsA}) at three time points t=400, 700, and 1000 within distinct periods 
under MCAR and missing rate of $50\%$; 
Figure~\ref{fig:simulation2.1} shows estimation result for time-invariant $\beta_{2,t}$ (coefficient of $A_{t-1}$) at $t=1000$ under MCAR, MAR, and MNAR and missing rate of $50\%$.
Estimation results for time-varying $\beta_{1,t}$ under MAR and MNAR, as well as the estimation results for remaining coefficients (random walk intercept $\beta_{0,t}$ and time-invariant coefficients $\rho_t$ and $\beta_{c,t}$) under MCAR, MAR, and MNAR are shown in the Appendix. Complete trajectory of estimated $\beta_{1,t}$ varying over time are also shown as an example in the Appendix. 

Figures~\ref{fig:simulation2.3} and \ref{fig:simulation2.1} demonstrate that for the non-stationary case, SSMimpute and SSMmp with correctly specified state space model provide unbiased estimation of the 3-periodic time-varying $\beta_{1,t}$ (coefficient of $A_t$) and time-invariant $\beta_{2,t}$ (coefficient of $A_{t-1}$) with adequate coverage of the $90\%$ CI. 
All other existing imputation methods are biased, for both the incorrectly specified linear model or with correctly specified state space model for post-imputation statistical analysis. 
To be specific, certain commonly employed imputation methods (e.g., mean imputation, multiple imputation, and LOCF) result in significant bias and a substantially under-covered $90\%$ CI. 
Complete case analysis has been a tempting option for researchers who are unsure of the appropriate imputation method to use. For time series data, it skips over time points with incomplete data and splices time points with complete data together for analysis. Consequently, data is misplaced in terms of number of time points that separate them and complete case analysis  will bias the estimation of the autocorrelation term. 
Our simulation demonstrates that complete case analysis is moderately biased for $\beta_{1,t}$ in periods 2 and 3 (shown in Figure~\ref{fig:simulation2.3}) and significantly biased for $\beta_{2,t}$ (shown in Figure~\ref{fig:simulation2.1}). 
We observe that the magnitude of bias using complete case analysis is greater for coefficients involving information across time points, such as $\beta_{2,t}$ (coefficient of $A_{t-1}$ on $Y_t$) and $\rho_t$ (coefficient of $Y_{t-1}$ on $Y_t$), than for coefficients involving information of the same time point, such as $\beta_{1,t}$ (coefficient of $A_t$ on $Y_t$) and $\beta_{c,t}$ (coefficient of $C_t$ on $Y_t$). 
This varying degrees of being biased for different coefficients can be explained using a similar argument put forth by \citet{carroll1985comparison,gleser1987limiting} in the context of measurement error bias: when $Y_{t-1}$ is measured with error (or in our case, complete records are spliced together and $Y_{t-1}$ is replaced by its autocorrelated surrogate $Y_{t-2}$), coefficients of $Y_{t-1}$ ($\rho_t$) and of variables correlated with $Y_{t-1}$ (e.g., $\beta_{2,t}$ of $A_{t-1}$) are biased; coefficients of variables uncorrelated with $Y_{t-1}$ (e.g., $\beta_{1,t}$ of $A_t$ and $\beta_{c,t}$ of $C_t$) are less affected.
 
Change points for three-pieced time varying $\beta_{1,t}$ (coefficient of $A_t$) are also identified, and their distribution  under various imputation imputation strategies are shown in the Appendix. 
All strategies with correctly specified state space model appear to identify change points with distributions centered around their true values. Complete case analysis underperforms other strategies with a more dispersive distribution due to the use of fewer data, which can be particularly apparent when missing data occurs near change points. Due to sampling uncertainty along the iterations, SSMmp has a more dispersive distribution of identified change points than SSMimpute. This can be improved by increasing the number of draws in its multiple imputation step.

For the non-stationary scenario, we also compare estimation performance of $\beta_{2,t}$ (coefficient of $A_{t-1}$) under various missing rate of $25\%$, $50\%$, and $75\%$, shown in the Appendix. Note that when the missing rate is low ($25\%$), the complete case analysis slightly underperforms our proposed methods of SSMmp and SSMimpute; however, as the missing rate increases, the estimation bias in complete case analysis becomes more apparent. 

\begin{figure}
    \centering 
    \includegraphics[width=\linewidth]{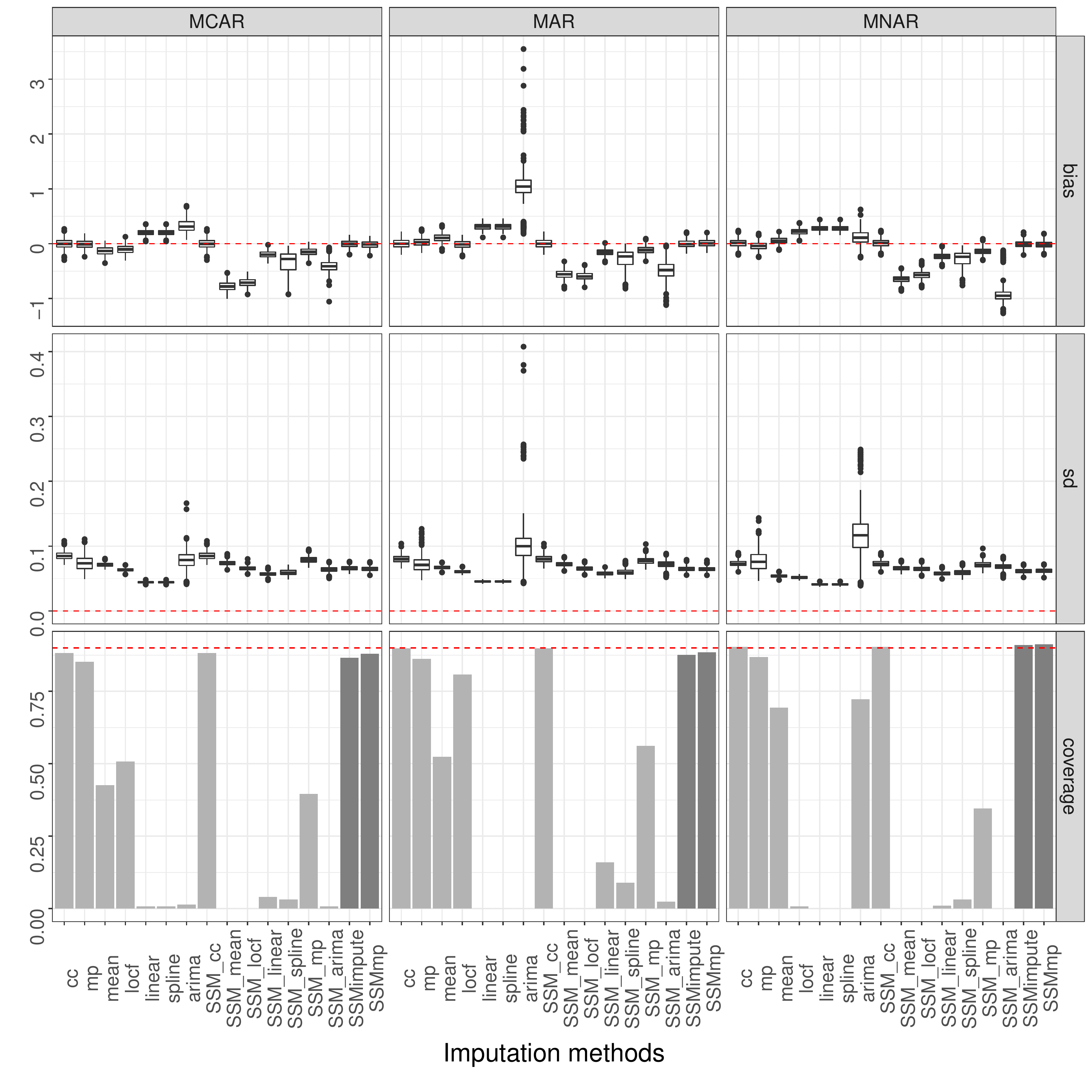}
\caption{Boxplots of the estimate (top), standard error (middle), and $90\%$ CI coverage (bottom) of the coefficient of $A_{t-1}$ for the stationary scenario, over 500 simulations under MCAR (left), MAR (middle), and MNAR (right) and missing rate of $50\%$.
Methods include complete case analysis (``cc''), multiple imputation (``mp''), mean imputation (``mean''), linear interpolation (``linear''), spline interpolation (``spline''), imputation with best ARIMA model (``arima'') under linear analytical models, and their corresponding versions under state space analytical models (with added ``SSM\_'' at the front), as well as the proposed state space model multiple imputation strategies -- SSMimpute and SSMmp.}
\label{fig:simulation1.1}
\end{figure} 

\begin{figure}
    \centering 
    \includegraphics[width=\linewidth]{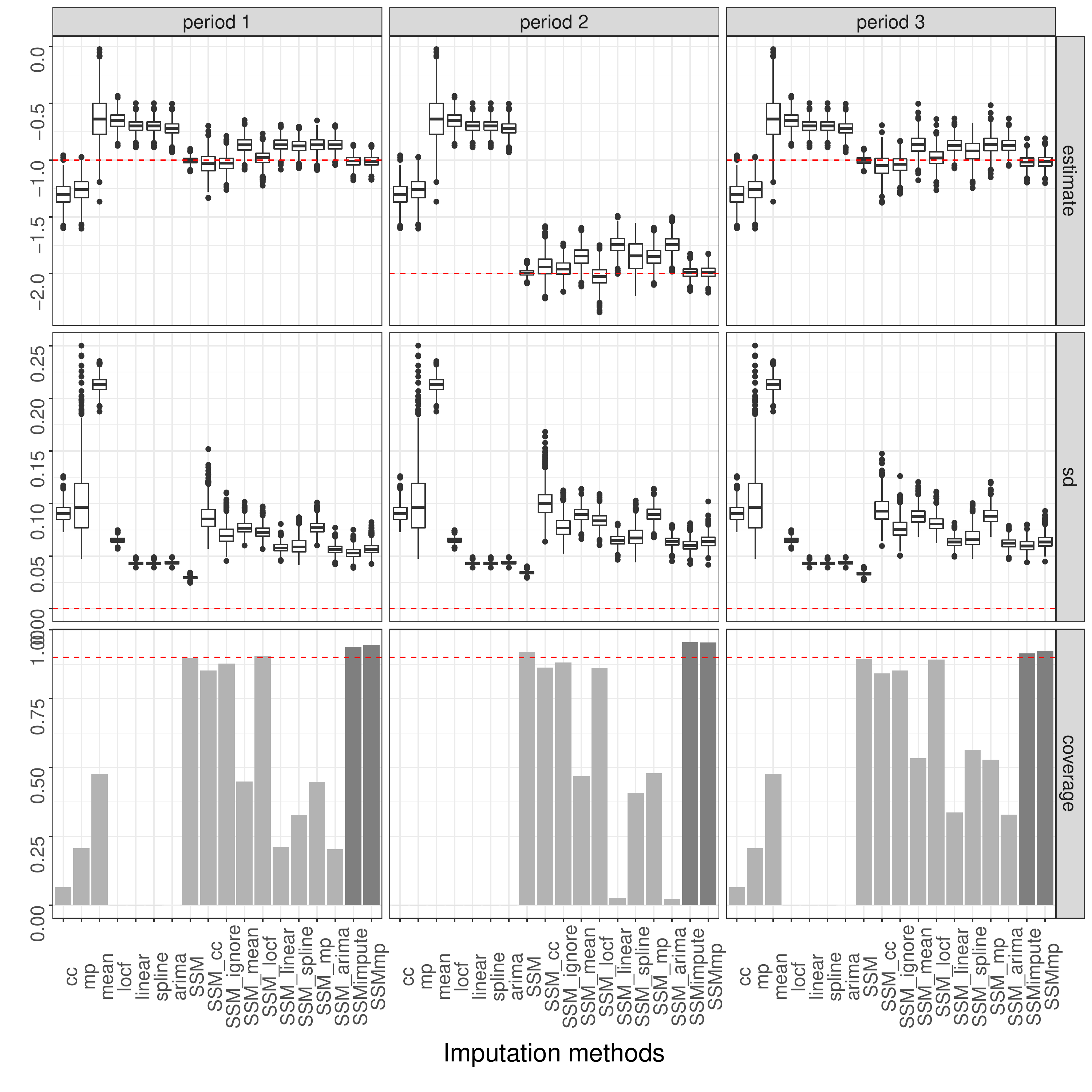}
\caption{Boxplots of the estimate (top), standard error (middle), and $90\%$ CI coverage (bottom) of the periodic-stable (time-varying) coefficient of $A_t$ over 3 distinct periods (from left to right) for the non-stationary scenario, over 500 simulations under MCAR and missing rate of $50\%$.
Methods include complete case analysis (``cc''), multiple imputation (``mp''), mean imputation (``mean''), linear interpolation (``linear''), spline interpolation (``spline''), imputation with best ARIMA model (``arima'') under linear analytical models, and their corresponding versions under state space analytical models (with added ``SSM\_'' at the front), as well as the proposed state space model multiple imputation strategies -- SSMimpute and SSMmp.}
\label{fig:simulation2.3}
\end{figure} 

\begin{figure}
    \centering 
    \includegraphics[width=\linewidth]{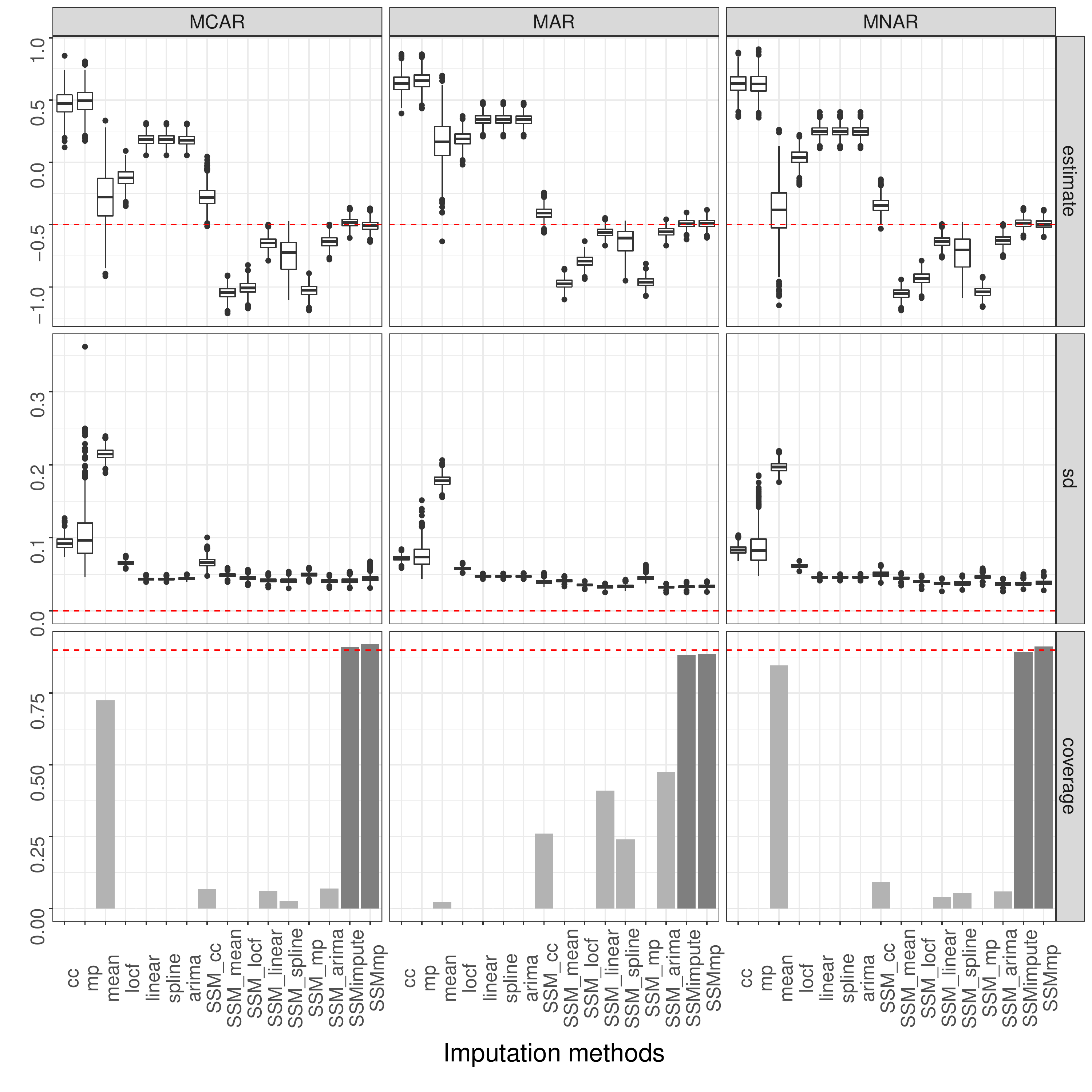}
\caption{Boxplots of the estimate (top), standard error (middle), and $90\%$ CI coverage (bottom) of the time-invariant coefficient of $A_{t-1}$ for the non-stationary scenario, over 500 simulations under MCAR (left), MAR (middle), and MNAR (right) and missing rate of $50\%$.
Methods include complete case analysis (``cc''), multiple imputation (``mp''), mean imputation (``mean''), linear interpolation (``linear''), spline interpolation (``spline''), imputation with best ARIMA model (``arima'') under linear analytical models, and their corresponding versions under state space analytical models (with added ``SSM\_'' at the front), as well as the proposed state space model multiple imputation strategies -- SSMimpute and SSMmp.}
\label{fig:simulation2.1}
\end{figure}

\section{Application}
\label{sec:application}
The Bipolar Longitudinal Study (BLS) is an ongoing mHealth cohort study that has recruited 74 patients with schizophrenia or bipolar illness from the Psychotic Disorders Division at McLean Hospital since February 2016. 
Once recruited, each participant underwent a comprehensive Diagnostic and Statistical Manual of Mental Disorders (DSM-IV) examination \citep{american2013diagnostic}. 
The study aimed at following each participant for at least 1 year (duration of mean(sd) 357.8(97.9) in days). A rich collection of data about physical activity, GPS locations, and anonymized basic information of calls and texts was passively collected using smartphones via the ``Beiwe'' platform \citep{huang2019activity,barnett2020inferring}. 
A customized 5-minute survey was sent to participants every day to inquire about their moods, sleep, social interactions, and psychotic symptoms.

It is worth noticing that individualized inference and a n-of-1 study design are implemented for each participants.
As shown in the left panel of  Figure~\ref{fig:selectedparticipants}, participants exhibit considerable heterogeneity regarding their enrollment time, follow-up duration, symptom trajectory, medication use, and behavioral therapy. As illustrated in the right panel of  Figure~\ref{fig:selectedparticipants}, their estimated time-varying coefficients $\beta_{12}$ (coefficient of $A_{\text{texts},t}$) also differ substantially in direction, magnitude, and segmentation, making it challenging to conduct pooled analysis for entire groups of patients.
In the remainder of the paper, we exemplify our proposed method considering a participant with bipolar I disorder who was followed for 708 days, and present the detailed analytical results for this participant.

As suggested by prior research and our collaborative psychiatric investigations, we are interested in quantifying the association of phone-based social network size with patients' self-reported negative mood.
The outcome of interest -- self-reported negative mood ($Y_t$) -- is a composite measure ranging from 0 (best) to 32 (worst), that captures negative feelings experienced by patients comprehensively, including fear, anxiety, embarrassment, hostility, stress, upset, irritated, and loneliness (all with a range from 0 to 4).
Among the numerous features captured by telecommunication, we choose the degree of contacts via outgoing calls ($A_{calls,t}$) and outgoing texts ($A_{texts,t}$) as the exposures of interest. Figure~\ref{fig:negative} shows the self-reported negative mood during the 708 days follow-up with 164 days missing, and the passively collected degree of contacts via texts and calls during the 708 days follow-up with no missing data.
Weather temperature ($\text{Temp}_t$) \citep{denissen2008effects,huibers2010does} and physical activity ($\text{PA}_t$) \citep{peluso2005physical,hamer2012physical} have been demonstrated to be associated with negative mood as well as the tendency for social interaction, and are thus included as confounders. 
The physical activity data are pre-processed following the strategy in \citet{bai2012movelets,bai2014normalization}. The trajectory of temperature and average level of physical activity in a moving window of past two weeks are shown in the Appendix. 
We include one-lagged outcome and one-lagged exposures to correct for confounding due to auto-correlation. Researchers may consider include additional information from the past. The analytical state space model is specified as follows.
\begin{equation}
\begin{split}
Y_t & = \beta_{0,t} + \rho_t Y_{t-1} + \beta_{11,t} A_{calls,t} + \beta_{12,t} A_{calls,t-1} + \beta_{21,t} A_{text,t} + \beta_{22,t} A_{texts,t-1} \\
	& \quad + \beta_{temp,t} \text{Temp}_t + \beta_{pa,t} \text{PA}_t + v_t
\end{split}
\end{equation}
This choice of model is determined by a standard time series model selection technique that compares one step ahead predictions \cite{rivers2002model}. Model selection details are shown in the Appendix. In the Appendix, we additionally present the $90\%$ confidence interval of the point-wise posterior distribution from which missing outcomes are imputed. The confidence interval band encompasses observed outcomes well and demonstrates satisfactory model fitting ability.

\begin{figure}[t]
\centering
\includegraphics[width=0.9\linewidth]{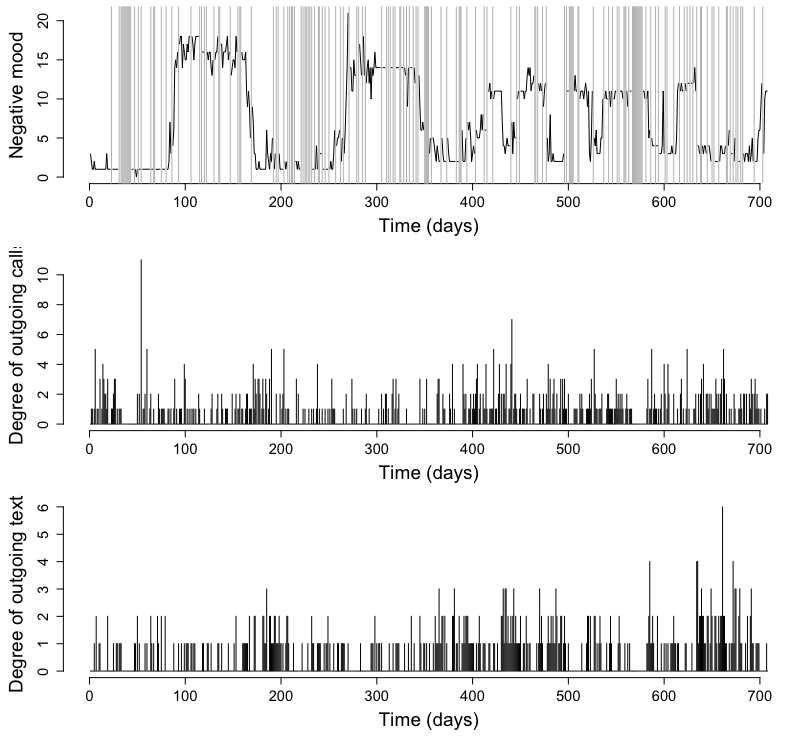}
\caption{The self-reported negative mood (top), degree of outgoing calls (middle), and degree of outgoing text (bottom) over 708 days of follow up for a bipolar patient enrolled in the BLS study. The gray vertical lines show the days with missing self-reports.}
\label{fig:negative}
\end{figure}

\begin{figure}
\centering
\begin{subfigure}[b]{\textwidth}
   \includegraphics[width=1\linewidth]{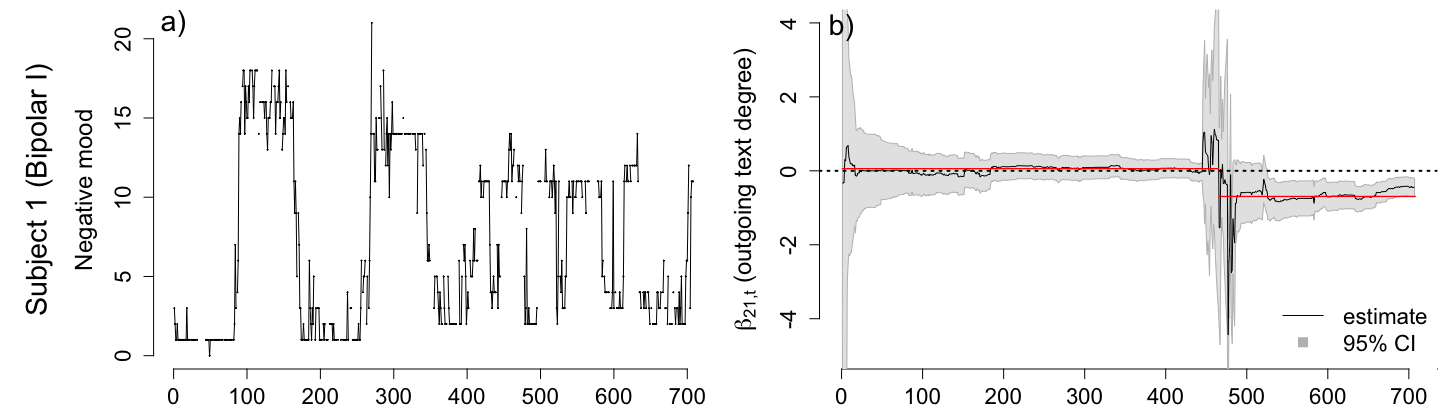}
\end{subfigure}
\begin{subfigure}[b]{\textwidth}
   \includegraphics[width=1\linewidth]{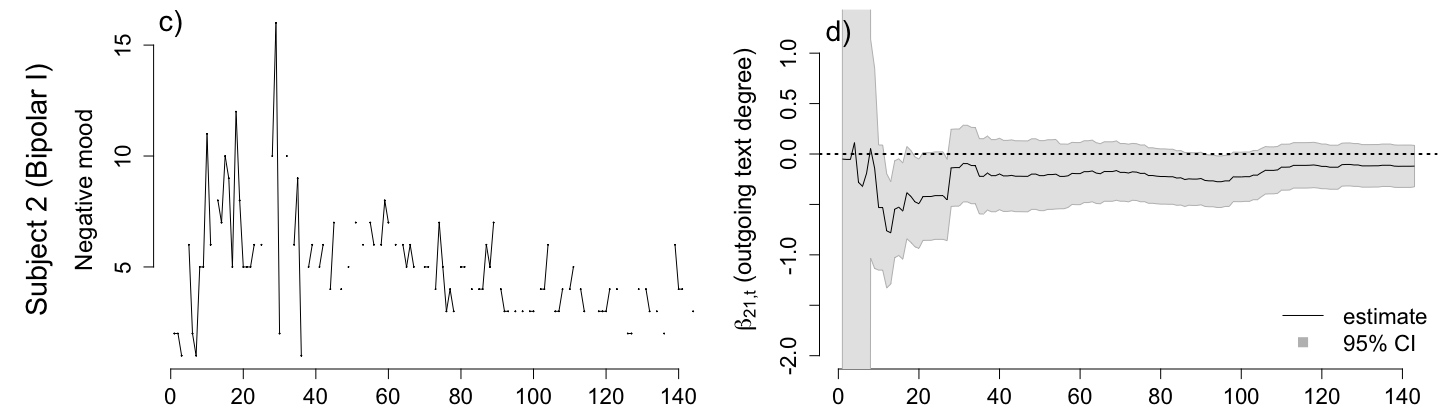}
\end{subfigure}
\begin{subfigure}[b]{\textwidth}
   \includegraphics[width=\linewidth]{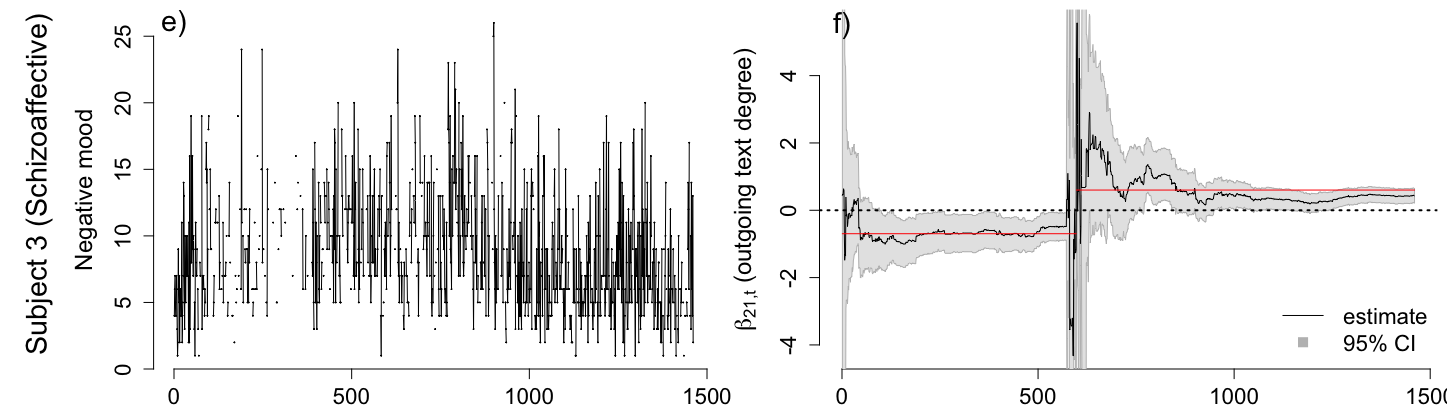}
\end{subfigure}
\begin{subfigure}[b]{\textwidth}
   \includegraphics[width=\linewidth]{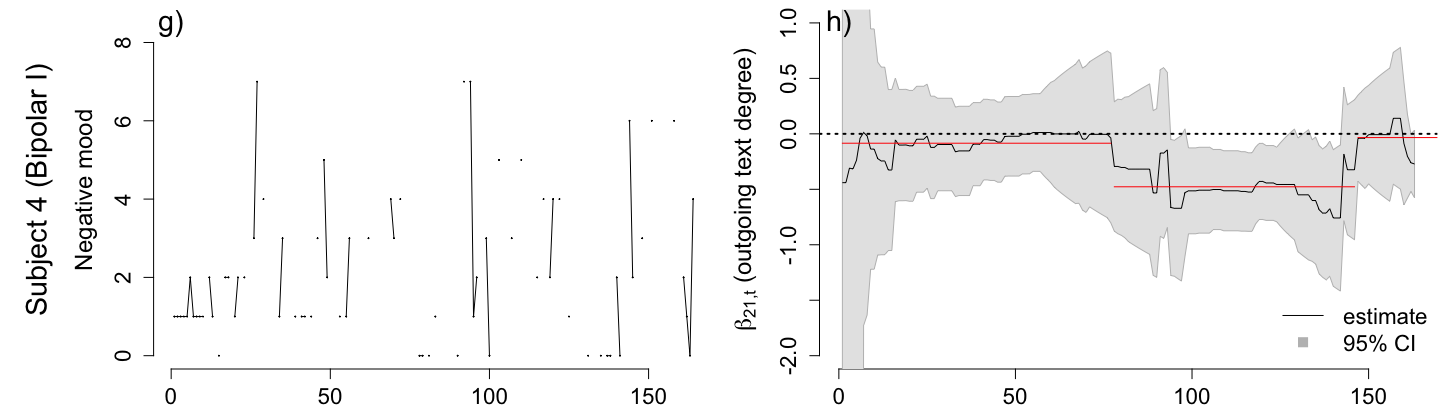}
\end{subfigure}
\caption{Self-reported negative moods (left) and estimated $\beta_{21,t}$ over time in days (right) for four subjects in BLS. For right panel, black lines represent the estimated coefficient over time, with $95\%$ confidence intervals represented in grey; the dashed horizontal lines at 0 indicates that there is no effect. Graph b) shows the estimated $\beta_{21,t}$ over time for the focal subject with Bipolar I disorder followed for 708 days; $\beta_{21,t}$ is found to be periodic-stable with two periods over time with a change point around day 466.
}
\label{fig:selectedparticipants}
\end{figure}

Table~\ref{tab:SSMimpute} presents the results of coefficient estimation using SSMimpute, compared to complete case analysis and multiple imputation under state space analytical model. The intercept is modeled as a random walk that continuously tracks the gradually drifting baseline level of negative mood. Negative mood is moderately auto-correlated with its previous value with an auto-correlation of $\rho_t=0.62$. 
The degree of outgoing calls ($A_{\text{calls},t}$) is significantly associated with a decrease in negative mood, whereas the degree of outgoing calls at the previous day ($A_{\text{calls},t-1}$) is not.
The estimated coefficient for degree of outgoing texts ($A_{\text{texts},t}$) is shown in Figure~\ref{fig:selectedparticipants}(b) and is found to be two-pieced over time, with a change point around day 466, prior to which there is no evidence of association and subsequent to which we uncover a significant negative association with negative mood. 
This transition point occurs concurrently with an increase in the medication Bupropion and a change in psychiatrist. 
The degree of outgoing text at the previous day ($A_{\text{texts},t-1}$) is not associated with decrease in negative mood.
In short, the absence of significant association between the degree of calls or texts received at the previous day with negative mood may indicate that social support is mostly associated with the same day's mood, whereas no association is detected for the following day.
Three distinct periods separated by day 162 and day 267 are identified for the estimated coefficient of average physical activity level ($\text{PA}_t$), with a significant negative correlation during the middle period, coinciding with the patient's second depression relapse. 
The outdoor temperature ($\text{Temp}_t$) is negatively correlated with the negative mood, but not significantly. 
The detailed plots of estimated coefficient trajectories over time of all variables are shown in the Appendix.
%

The estimation results using complete case analysis (CC) appear to be similar to those found using SSMimpute for variables including the degree of outgoing texts and calls on the same day and temperature; however, the estimated coefficients for variables across time points, such as the autocorrelation term ($\rho_t$), the coefficients for the 1-lagged degree of texts and calls ($\beta_{12,t}$ and $\beta_{22,t}$) differ, reflecting the potential bias induced by interrupting temporal relationship via complete case analysis. 
The majority of standard errors obtained with SSMimpute are lower than those obtained with complete case analysis. 
This improvement, however, is not significant given the low missing rate and the long follow-up. 
Improvement in estimation efficiency for cases with high missing rate are more pronounced, as shown by simulations.
The discrepancy in estimation between multiple imputation and SSMimpute confirms the substantial estimation bias introduced by multiple imputation, which is most obvious in the underestimated auto-correlation term of $0.14$. 
Simulations also confirm this underestimation of auto-correlation using multiple imputation in the non-stationary scenario (shown in the Appendix). 
\begin{table}
\centering
\begin{tabular}{lcccccc}
  \hline
\multirow{2}{*}{Variables} & \multicolumn{2}{c}{SSMimpute} & \multicolumn{2}{c}{complete case} & \multicolumn{2}{c}{multiple imputation (MICE)} \\ 
\cmidrule(lr){2-3} \cmidrule(lr){4-5} \cmidrule(lr){6-7} 
  & Estimate & Std.Error & Estimate & Std.Error & Estimate & Std.Error\\
  \hline
$\beta_{0,t}$ & \multicolumn{2}{c}{(random walk)} & \multicolumn{2}{c}{(random walk)} & \multicolumn{2}{c}{(random walk)} \\ 
$\rho_{t}$ & 0.62 & 0.05 & 0.65 & 0.04 & 0.14 & 0.08 \\ 
$\beta_{11,t}$ & -0.12 & 0.07 & -0.16 & 0.08 & -0.1 & 0.06 \\ 
$\beta_{12,t}$ & 0.02 & 0.06 & -0.08 & 0.07 & -0.03 & 0.06 \\ 
$\beta_{21,t}$(period1) & -0.00 & 0.14 & 0.01 & 0.15 & -0.03 & 0.13 \\ 
$\beta_{21,t}$(period2) & -0.42 & 0.15 & -0.42 & 0.17 & -0.34 & 0.14 \\ 
$\beta_{22,t}$ & -0.18 & 0.10 & -0.22 & 0.11 & -0.23 & 0.1 \\ 
$\beta_{pa,t}$(period1) & -7.04 & 5.79 & -8.49 & 5.70 & -4.49 & 7.53 \\ 
$\beta_{pa,t}$(period2) & -9.48 & 5.77 & -8.45 & 6.18 & -11.98 & 9.95 \\ 
$\beta_{pa,t}$(period3) & 1.78 & 1.80 & 1.88 & 1.90 & 1.17 & 2.83 \\ 
$\beta_{temp,t}$& -0.01 & 0.01 & -0.01 & 0.01 & -0.01 & 0.01 \\ 
  \hline
\end{tabular}
\caption{Estimated state space model coefficients in post-imputation analysis, using data imputation strategies of SSMimpute (left), complete case analysis (middle), and multiple imputation (right).}
\label{tab:SSMimpute}
\end{table}

\section{Discussion}
Despite the increased complexity in data structure and statistical analysis, multivariate time series from intense monitoring or measurement will inevitably become more prevalent and even standard for future mental health and medical research due to the widespread use of mobile technology and wearable devices. 
We propose novel imputation strategies to solve the missing data problem induced by missing outcomes used as both response and explanatory variables in non-stationary multivariate time series that arise in observational n-of-1 studies.
Commonly used imputation methods, including multiple imputation and last-observation-carried-forward, introduce significant bias into coefficient estimation. 
Those methods do not produce ``good'' imputation candidates for missing values as they are designed for static rather than dynamic processes, and certain methods based on univariate time series are unable to account for influences from exposure and covariate time series.
Complete case analysis also introduces bias into coefficient estimation by improperly skipping over missing data.
Our proposed approaches combine multiple imputation and the state space model, preserve the temporal structure between variables to aid in the estimation of time-varying coefficients, and allow the incorporation of exposure and covariate time series as well as lagged values for improved missing data imputation.
Additionally, our proposed strategies incorporate missing data imputation into the outcome modeling, thereby avoiding model incompatibility that arises when different models are employed for imputation and for analysis, as is the case with multiple imputation.

Leveraging the flexibility of multiple imputation and state space model, our proposed method also applies to missing data imputation for multiple subjects and can be extended to incorporate trend and seasonality, under the assumption that subjects follow the same dynamic model over time so that multivariate time series from those subjects can be grouped for joint statistical inference.
Individualized inference is a special instance and was specifically chosen, given the substantial heterogeneity of individuals during the extended follow-up in the BLS study motivating our work.
SSMimpute applied to groups of subjects with possible random coefficients may be appealing in circumstances where participants with similar conditions are recruited at roughly the same time and their respective multivariate time series exhibit comparable patterns.
Before using imputation techniques to groups of individuals, we suggest individualized imputation and inference to investigate the similarity of multivariate time series and potential dynamic outcome models. 
 Grouped imputation and inference may be preferable for enhanced signal and generalizability of the conclusion to a larger (sub) population if the dynamic outcome models demonstrate comparable patterns over time; however, individualized imputation and inference in an n-of-1 setting are more appropriate if the time series and dynamic outcome models are highly heterogeneous.
 In short, our proposed method advances the methodology for missing data imputation in multivariate non-stationary time series, and lays the groundwork for inferring causal relationships for non-stationary time series from N-of-1 studies.

The BLS study is a pioneer work in long-term smartphone monitoring of wealthy information in the context of Severe Mental Illness (SMI).
We apply the SSMimpute strategy to investigate the association between phone-based social interaction and negative mood in a bipolar patient followed for nearly two years, adjusting for confounding by weather temperature, physical activity, and previous values of outcome and exposures of interest. We observe a moderate auto-correlation between the outcome and its previous values.
Additionally, our analysis reveals a negative correlation between the degree of outgoing calls and texts contacts and negative mood, and also found that these effects may vary over time.
Critical change points are identified for the effect of degree of outgoing texts, which coincide to a change in psychiatrist and significant medication changes. 
These findings provide preliminary evidence  of the role of social network structure in affecting patient's mood over time in a complex fashion, depending on patient's symptoms and living environment.

The current implementation of SSMimpute is limited to  linear relationships between outcomes, exposures, covariates, and their lagged values.  
In order to address the issue of possible model misspecification, we will consider more flexible analytical models other than linear relationships in state space modeling for future research. 
Additionally, the missing outcomes considered in our proposed methods are restricted to be normally distributed; extensions to binary and categorical outcomes are being investigated. 
Furthermore, residual confounding may exist in real-world data analysis as a result of unmeasured individual and contextual confounders. However, the lagged outcomes are likely the strongest confounders in time series analysis, and together with the time-varying intercept, they also capture and help control some unmeasured individual and contextual information from the environment. 
As missing data could arise in other predictors (either actively or passively collected), we will also consider a similar imputation strategy to impute missing data in exposures and confounders in future work.
Our proposed imputation methods are based on ignorable likelihood, assuming MCAR or MAR; future work can incorporate MNAR by explicitly specifying missing mechanism models in the full likelihood.
Finally, our analysis reveals significant heterogeneity across participants regarding the association between phone-based social interaction and negative mood. This phenomenon may be attributable to unmeasured confounding variables or latent disease status that modifies the association over time. Future research will investigate how a latent disease modifies the effect of social engagement on mood improvement, providing directions for conducting group inferences for participants with heterogeneous severe mental diseases.


\section*{Acknowledgements}
This work is supported by K01MH118477, U01MH116925.

\singlespacing
\bibliographystyle{unsrtnat}

\bibliography{SSMimpute.bib}

\appendix

\section{SSMmp and SSMimpute algorithms}
\subsection{Flowchart for the SSMmp algorithm}
Figure~\ref{fig:flowchart} shows the flowchart the of SSMmp algorithm.
\begin{figure}
\centering
\includegraphics[width=\linewidth]{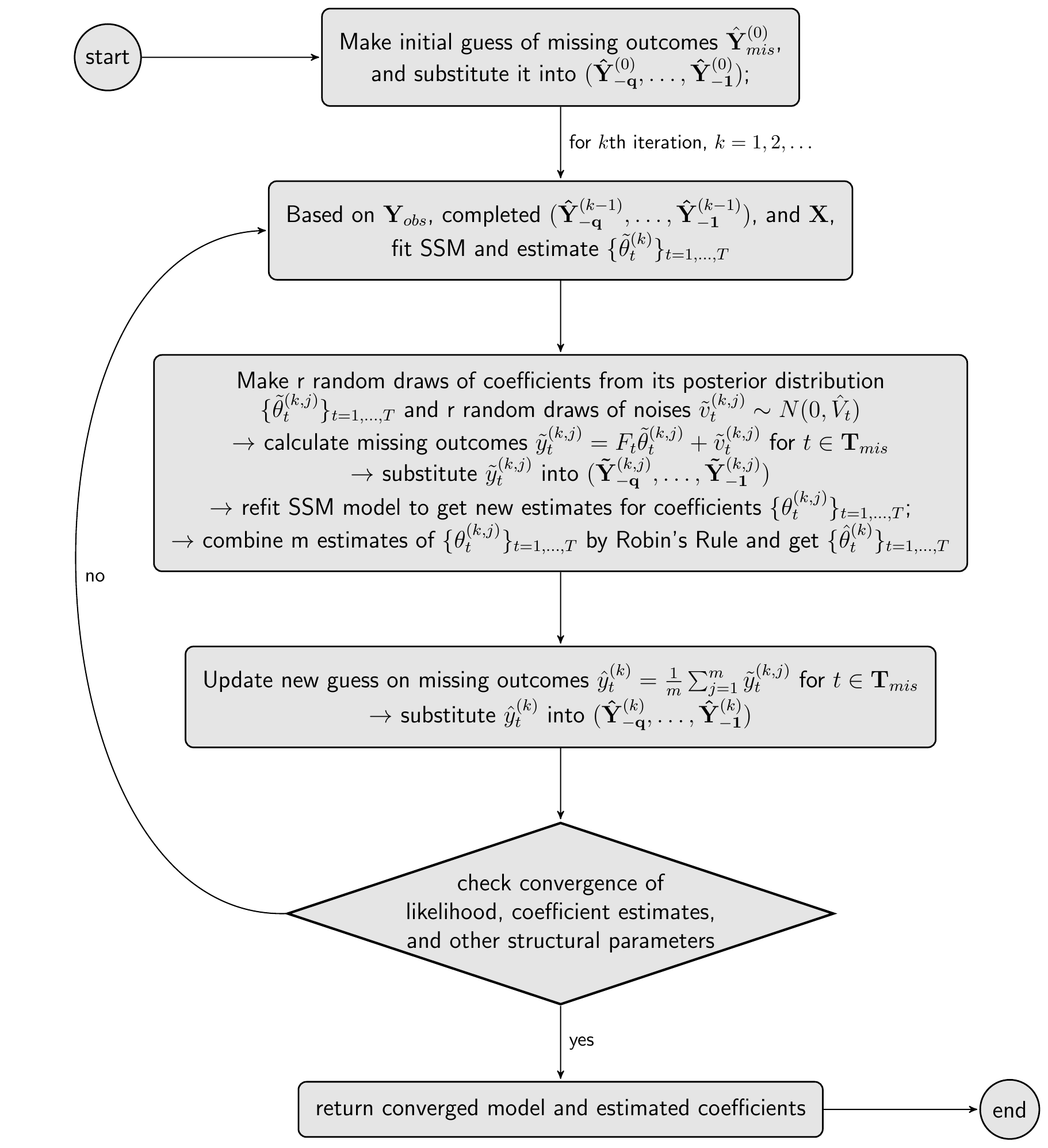}
\caption{The flowchart of the ``SSMmp'' algorithm for dealing with missing data in outcomes in non-stationary multivariate time series.}
\label{fig:flowchart}	
\end{figure}

\subsection{Steps of the SSMimpute algorithm}

``SSMimpute'' iterates between model fitting and substituting until getting closer to the maximization and then starts applying multiple imputation. 
We show the detailed algorithm for the ``SSMimpute'' as follows. \begin{Steps}
\item(Initialization) Make initial guess on missing outcomes $\mathbf{Y}_{\text{mis}}^{(0)}$ and other nuisance structural parameters of the state space model. Substitute guessed values $\mathbf{Y}_{\text{mis}}^{(0)}$ into their corresponding missing lagged values $(\mathbf{\hat{Y}_{t-q}}^{(0)},\ldots,\mathbf{\hat{Y}_{t-1}}^{(0)})$ to eliminate missing data in explanatory variables.
\item(Maximization) For the $k$th iteration, $k=1,2,\ldots$, apply state space model to incomplete outcomes $\mathbf{Y}_{\text{obs}}$ and completed explanatory variables $(\mathbf{Y_{t-q}}^{(k-1)},\ldots,\mathbf{Y_{t-1}}^{(k-1)},\X)$ to obtain the maximum likelihood estimate (MLE) of coefficients $\{\hat{\theta}^{(k)}_t\}_{t=1,\ldots,T}$ and other nuisance parameters of the state space model.
\item(Substitution) Calculate imputations of missing outcomes $\hat{y}_t^{(k)}=F_t \hat{\theta}^{(k)}_t$ from the estimated posterior distribution of $\{\hat{\theta}^{(k)}_t\}_{t=1,\ldots,T}$, substitute those guessed values into their corresponding lagged variables $(\mathbf{\tilde{Y}_{t-q}}^{(k,j)},\ldots,\mathbf{\tilde{Y}_{t-1}}^{(k,j)})$.
\item (Check convergence) Repeat steps 2 to 3 until convergence is achieved for likelihood, coefficient estimation, and other structural parameters of the state space model.
\item (Multiple imputation) After reaching convergence at iteration $K$, make $r$ random draws of estimated coefficients $\{\hat{\theta}^{(K,j)}\}_{t=1,\ldots,T}$ from its posterior distribution of $\{\hat{\theta}^{(K)}_t\}_{t=1,\ldots,T}$ and $r$ random draws of noises from $\tilde{v}_t \sim N(0,\hat{V}_t)$. For $j=1,\ldots,r$, calculate missing outcomes $\tilde{y}_t^{(K,j)}=F_t\hat{\theta}^{(k,j)}_t+\tilde{v}_t $ for $t \in \mathbf{T}_{\text{mis}}$, substitute those guessed values into their corresponding lagged variables $(\mathbf{\tilde{Y}_{t-q}}^{(K,j)},\ldots,\mathbf{\tilde{Y}_{t-1}}^{(K,j)})$, refit the state space model, and finally obtain new estimations for coefficients $\{\dbtilde{\theta}^{(K,j)}_t\}_{t=1,\ldots,T}$. Apply Rubin's rule to combine the $r$ newly estimated $\{\dbtilde{\theta}^{(K,j)}_t\}_{t=1,\ldots,T}$, $j=1,\ldots,m$ to be the ultimate estimation of coefficients, $\{\hat{\theta}_t\}_{t=1,\ldots,T}$.
\end{Steps}

\section{Proofs}

\begin{proof}[Proof of Proposition \ref{prop:increasedefficiency}]
Denote the estimate from complete case analysis by discarding observation at time $t$ as $\hat{\theta}_{t^*} \sim N(\hat{m}_{t^*},\hat{W}_{t^*})$. Denote the estimate from ``SSMmp'' at time $t$ and $t+1$ by imputing the missing lagged outcomes as $\hat{\theta}_t \sim N(\hat{m}_t, \hat{W}_t)$ and $\hat{\theta}_{t+1} \sim N(\hat{m}_{t+1}, \hat{W}_{t+1})$. The true value $\theta_t=\dot{\theta}$ for $t=1,2,\ldots$ is time-invariate. Assume Kalman filter or smoothing provides unbiased estimate $\hat{\theta}_{t-1} \sim N(\hat{m}_{t-1}, \hat{W}_{t-1})$ for the true value of $\theta_{t-1}$ so that $\E[\hat{m}_{t-1}]=\dot{\theta}$. 

According to the Kalman filter, we have,
\begin{equation*}
\begin{split}
	\hat{m}_{t^*} &= G_{t^*}\hat{m}_{t-1}+R_{t^*} F'_{t^*} Q_{t^*}^{-1} (Y_{t^*}-F_{t^*}G_{t^*}\hat{m}_{t-1}) \\
	& = G_{t+1}\hat{m}_{t-1}+R_{t+1} F'_{t+1} Q_{t+1}^{-1} (Y_{t+1}-F_{t+1}G_{t+1}\hat{m}_{t-1})  \\
\intertext{plug in $Y_{t+1}=F_{t+1}\theta_{t+1}+v_{t+1}$}
	& = G_{t+1}\hat{m}_{t-1}+R_{t+1} F'_{t+1} Q_{t+1}^{-1} (F_{t+1}\theta_{t+1}+v_{t+1}-F_{t+1}G_{t+1}\hat{m}_{t-1})  \\
\intertext{plug in $\theta_t=G_t\theta_{t-1}+w_{t}$}
	& = G_{t+1}\hat{m}_{t-1}+R_{t+1} F'_{t+1} Q_{t+1}^{-1} (F_{t+1} G_{t+1} \theta_{t} + F_{t+1} W_{t+1} +v_{t+1}-F_{t+1}G_{t+1}\hat{m}_{t-1})  \\
	& = G_{t+1}\hat{m}_{t-1}+R_{t+1} F'_{t+1} Q_{t+1}^{-1} (F_{t+1} G_{t+1} G_t \theta_{t-1} + F_{t+1} G_{t+1} W_t + F_{t+1} W_{t+1} +v_{t+1}-F_{t+1}G_{t+1}\hat{m}_{t-1})  \\
\intertext{plug in $G_t=I$ and $W_t=0$ for time-invariate paramters}
	& = \hat{m}_{t-1}+R_{t+1} F'_{t+1} Q_{t+1}^{-1} (F_{t+1} \theta_{t-1} +v_{t+1}-F_{t+1}\hat{m}_{t-1})  \\
\end{split}
\end{equation*}
\begin{equation*}
\begin{split}
\E[\hat{m}_{t^*}] & = \E[\hat{m}_{t-1}] + R_{t+1} F'_{t+1} Q_{t+1}^{-1} F_{t+1}(\theta_{t-1} -\E[\hat{m}_{t-1}]) \\
\intertext{by $\E[\hat{m}_{t-1}]=\theta_{t-1}$} 
& = \E[\hat{m}_{t-1}]=\theta_{t-1}= \dot{\theta}\\
\end{split}
\end{equation*}
Therefore, the complete case analysis provides unbiased estimate for unknown coefficients for time-invariate system.

Alternatively, if we impute missing values of lagged outcomes in $F_t$ as $\tilde{F}_t$ with the proposed imputation method, we can update the parameter estimation at both time $t$ and $t+1$. 

For time $t$, we have:
\begin{equation}
\begin{split}
	\hat{m}_t &= G_t \hat{m}_{t-1}+R_t \tilde{F}'_t Q_t^{-1} (Y_t-\tilde{F}_{t} G_t \hat{m}_{t-1})\\
\intertext{plug in $Y_{t}=F_{t}\theta_{t}+v_{t}$}
	&= G_t \hat{m}_{t-1}+ R_t \tilde{F}'_t Q_t^{-1} (F_{t}\theta_t +v_t -\tilde{F}_{t} G_t \hat{m}_{t-1})\\
\intertext{plug in $\theta_t=G_t\theta_{t-1}+w_{t}$}
		&= G_t \hat{m}_{t-1}+ R_t \tilde{F}'_t Q_t^{-1} (F_{t}G_t\theta_{t-1}+F_{t}w_t +v_t -\tilde{F}_{t} G_t \hat{m}_{t-1})\\
\intertext{plug in $G_t=I$ and $W_t=0$ for time-invariate paramters}
		& = \hat{m}_{t-1}+ R_t \tilde{F}'_t Q_t^{-1} (F_{t}\theta_{t-1}+v_t -\tilde{F}_{t}\hat{m}_{t-1})\\
\end{split}
\end{equation}
\begin{equation}
\begin{split}
	\E[\hat{m}_t] & = \E[\hat{m}_{t-1}] +R_t \E[\tilde{F}'_t] Q_t^{-1} ( F_{t}\theta_{t-1}  - \E[\tilde{F}'_t]\E[\hat{m}_{t-1}]) \\
	 & = \E[\hat{m}_{t-1}] +R_t \E[\tilde{F}'_t] Q_t^{-1} ( F_{t}\theta_{t-1}  - F_{t}\E[\hat{m}_{t-1}] + F_{t}\E[\hat{m}_{t-1}] -\E[\tilde{F}'_t]\E[\hat{m}_{t-1}]) \\
	 &= \E[\hat{m}_{t-1}] +R_t \E[\tilde{F}'_t] Q_t^{-1} F_{t} (\theta_{t-1}-\E[\hat{m}_{t-1}) + R_t \E[\tilde{F}'_t] Q_t^{-1} ( F_{t}-\E[\tilde{F}'_t])\E[\hat{m}_{t-1}]\\
\intertext{by $\E[\hat{m}_{t-1}]=\theta_{t-1}$} 
	& = \E[\hat{m}_{t-1}] +R_t \E[\tilde{F}'_t]  Q_t^{-1}( F_{t}-\E[\tilde{F}'_t])\E[\hat{m}_{t-1}] \\
\intertext{since $\tilde{F}_{t}$ and $F_t$ only differ on $\hat{Y}_{t-1},\ldots,\hat{Y}_{t-q}$, whose expectation equals the true value of $Y_{t-1},\ldots,Y_{t-q}$, therefore $\E[\tilde{F}'_t]=F_{t}$}
	& = \E[\hat{m}_{t-1}] = \theta_{t-1} = \dot{\theta} \\
\end{split}
\end{equation}
Similarly, for time $t+1$, we have:
\begin{equation}
\begin{split}
	\hat{m}_{t+1} & = G_{t+1} \hat{m}_{t}+R_{t+1} F'_{t+1} Q_{t+1}^{-1} (Y_{t+1}-F_{t+1} G_{t+1}\hat{m}_{t})\\
\intertext{plug in $Y_{t+1}=F_{t+1}\theta_{t+1}+v_{t+1}$}
	& = G_{t+1} \hat{m}_{t}+R_{t+1} F'_{t+1} Q_{t+1}^{-1} (F_{t+1} \theta_{t+1}+v_{t+1}-F_{t+1} G_{t+1}m_{t})\\
\intertext{plug in $G_{t}=I$ and $W_t=0$ for time-invariate parameters}
	& = \hat{m}_{t}+R_{t+1} F'_{t+1} Q_{t+1}^{-1} (F_{t+1}\theta_t +v_{t+1}-F_{t+1} \hat{m}_{t})\\
	\E[\hat{m}_{t+1} ] & = \E[\hat{m}_{t} ] + R_{t+1} F'_{t+1} Q_{t+1}^{-1} F_{t+1} (\theta_t -\E[\hat{m}_{t}]) \\	
	& = \E[\hat{m}_{t} ] = \dot{\theta} \\
\end{split}
\end{equation}
Therefore, imputation using state space model provides unbiased estimate for unknown coefficients for the time-invariate system if outcome is missing at $t$. Similar logic applies to missing outcomes at other time points. Therefore, for multiple unbiased estimation from the same posterior distribution, $\hat{m}_{t+1}^{(k)}$ for $k=1,\ldots,K$, SSMmp estimate, $\frac{1}{K}\sum_{k=1}^k \hat{m}_{t+1}^{(k)}$, provides unbiased estimation for unknown time-invariant coefficients.

For the standard error in estimation, we compare $\hat{W}_{t^*}$ from complete case analysis and the $\hat{W}_{t+1}$ after state space model imputation. 
\begin{equation}
\begin{split}
\tilde{C}_{t^*}&=R_{t^*} - R_{t^*} F_{t^*} Q_{t^*}^{-1} F_{t^*} R_{t^*} \\
& = R_{t+1} - R_{t+1} F_{t+1} Q_{t+1}^{-1} F_{t+1} R_{t+1}\\
\intertext{plug in $R_t=G_tC_{t-1}G_t+W_t$ and $G_t=I$ and $W_t=0$ for time-invariate system}
	& = C_{t-1} - C_{t-1}  F_{t+1} Q_{t+1}^{-1} F_{t+1} C_{t-1} \\
\intertext{plut in $Q_t=F_tR_tF'_t+V_t$ and $R_t=C_{t-1}$ for time-invariate system}
	& = C_{t-1} - C_{t-1}  F_{t+1} (F_{t+1}C_{t-1} F'_{t+1}+V_{t+1})^{-1} F_{t+1} C_{t-1} \\
	& = (F'_{t+1}V_{t+1}F_{t+1} + C^{-1}_{t-1})^{-1}
\end{split}
\end{equation}
\begin{equation}
\begin{split}
	\hat{W}_{t+1} & =R_{t+1} - R_{t+1} F_{t+1} Q_{t+1}^{-1} F_{t+1} R_{t+1} 
	\intertext{plug in $R_t=G_tC_{t-1}G_t+W_t$ and $G_t=I$ and $W_t=0$ for time-invariate system}
	& = C_{t} - C_{t}  F_{t+1} Q_{t+1}^{-1} F_{t+1} C_{t} \\
	\intertext{plug in $Q_t=F_tR_tF'_t+V_t$ and $R_t=C_{t-1}$ for time-invariate system}
	& = C_{t} - C_{t}  F_{t+1} (F_{t+1}C_t F'_{t+1}+V_{t+1})^{-1} F_{t+1} C_{t} \\
	& = (F'_{t+1}V_{t+1}F_{t+1} + C^{-1}_{t})^{-1} \\
\end{split}
\end{equation}

By Loewner theory, for two real symmetric and positive definite square matrix $A$ and $B$, if $A \succeq B$ then $A^{-1} \preceq B^{-1}$. 
Therefore, since $C_t \preceq C_{t-1}$, thus $C^{-1}_t \succeq C^{-1}_{t-1}$. By adding a same symmetric and positive definite square matrix $F'_{t+1}V_{t+1}F_{t+1}$, we have $F'_{t+1}V_{t+1}F_{t+1} + C^{-1}_{t} \succeq F'_{t+1}V_{t+1}F_{t+1} + C^{-1}_{t-1}$. Applying the Leowver theory again, we have
\begin{equation}
\hat{W}_{t+1}= (F'_{t+1}V_{t+1}F_{t+1} + C^{-1}_{t})^{-1} \preceq  (F'_{t+1}V_{t+1}F_{t+1} + C^{-1}_{t-1})^{-1} = \tilde{C}_{t^*}
\end{equation}
In short, estimation standard error for time $t+1$ is reduced by including the observation at time $t$. SSMmp returns the estimates at the last time point $T$ as the final estimation for the unknown time-invariant parameter. As shown by \citet{meng1994multiple} and \citet{white2010bias}, for static analytical models, the increased variance from multiple imputation by Robin's rule is approximately $1+ \text{missing rate}/r$, which is ignorable when $r \to \infty$. We prove the improved efficiency after multiple imputation in SSMmp.
\end{proof}

\begin{proof}[Proof of Proposition \ref{prop:reducedbias}]
For time-varying system, denote the estimate of $\theta_{t+1}$ without missing data in $F_t$ as $\mathring{\theta}_{t+1} \sim N(\mathring{m}_{t+1}, \mathring{C}_{t+1})$. 

\begin{equation}
\begin{split}
	\mathring{m}_{t+1} & = G_{t+1} \mathring{m}_{t}+R_{t+1} F'_{t+1} Q_{t+1}^{-1} (Y_{t+1}-F_{t+1} G_{t+1}\mathring{m}_{t})\\
\intertext{plug in $Y_{t+1}=F_{t+1}\theta_{t+1}+v_{t+1}$}
	& = G_{t+1} \mathring{m}_{t}+R_{t+1} F'_{t+1} Q_{t+1}^{-1} (F_{t+1} \theta_{t+1}+v_{t+1}-F_{t+1} G_{t+1}\mathring{m}_{t})\\
\label{eq:dynamictheta1}
\end{split}
\end{equation}

Similarly,
\begin{equation}
\begin{split}
\mathring{m}_{t} & = G_{t} \mathring{m}_{t-1}+R_{t} F'_{t} Q_{t}^{-1} (F_{t} \theta_{t+1}+v_{t+1}-F_{t+1} G_{t+1}\mathring{m}_{t-1}) \\
\label{eq:dynamictheta2}
\end{split}
\end{equation}
It has been proven by massive literature that $\mathring{m}_{t+1}$ and $\mathring{m}_{t}$ is exponentially convergent with tracking error bounded under certain stochastic conditions.

We consider the ``SSMmp'' method. the Kalmen filter estimation follow the same procedure as in \eqref{eq:dynamictheta1}--\eqref{eq:dynamictheta2} excet that $F_t$ is replace with $\tilde{F}_t$ with missing lagged outcomes replaced by their imputed values. Since $\E[\tilde{F}_t]=F_t$, then $\E[\hat{m_t}]=\mathring{m_t}$ and $\E[\hat{m_{t+1}}]=\mathring{m_{t+1}}$, which are exponentially convergent in expectation with bounded tracking error under certain stochastic conditions.

Alternatively, if we consider the complete case analysis when we skip time $t$, then
\begin{equation}
\begin{split}
	\hat{m}_{t^*} & = G_{t+1} \hat{m}_{t-1}+R_{t+1} F'_{t+1} Q_{t+1}^{-1} (F_{t+1} \theta_{t+1}+v_{t+1}-F_{t+1} G_{t+1}\hat{m}_{t-1})\\
	\label{eq:dynamictheta3}
\end{split}
\end{equation}

Compare \eqref{eq:dynamictheta2} and \eqref{eq:dynamictheta3}, we see systematic difference may exist on $F_t$ and $F_{t+1}$, since they represent exogenous information of exposures and other covariates, which may further induce bias in the estimation of $\hat{m}_{t^*}$. In short, there is no guarantee on unbiasedness of $\hat{m}_{t^*}$ when $F_{t+1}$ is mis-allied with $\hat{m}_{t-1}$.
\end{proof}

\section{Additional simulation results}
\subsection{Stationary scenario}
Figures~\ref{fig:simulation1.2}--\ref{fig:simulation1.4} illustrates the estimation results, standard errors, and
coverages for the estimated time-invariant coefficients of $Y_{t_1}$, $A_t$, and $C_t$ in the stationary scenario, respectively, over 500 simulations under MCAR, MAR, and MNAR and under missing rate of $50\%$.  

\begin{figure}
    \centering 
    \includegraphics[width=\linewidth]{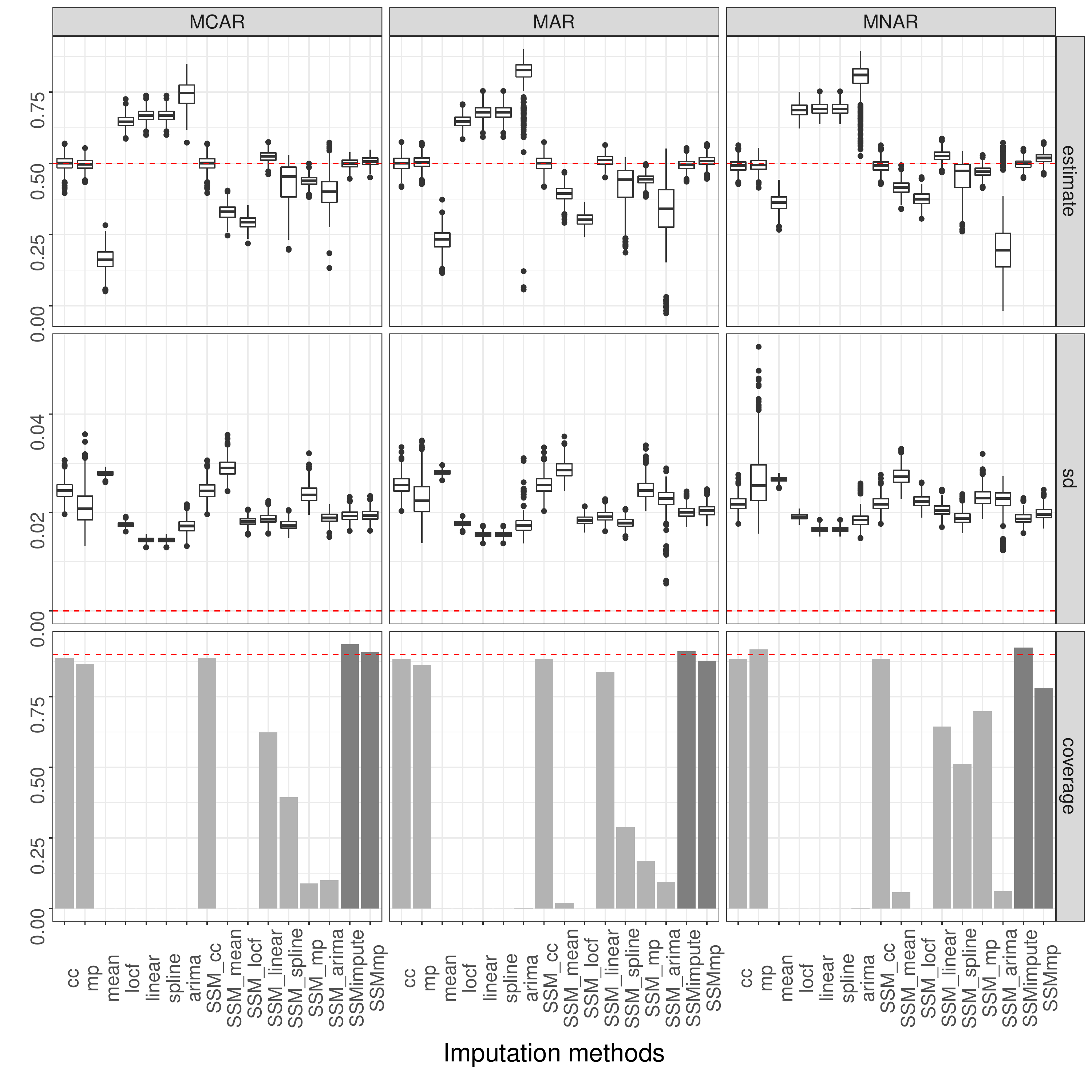}
\caption{Boxplots of the estimate (top), standard error (middle), and $90\%$ CI coverage (bottom) of the auto-correlation term of $Y_{t-1}$ for the stationary scenario, over 500 simulations under MCAR (left), MAR (middle) and MNAR (right) and missing rate of $50\%$.
Methods to be compared include complete case analysis (``cc''), multiple imputation (``mp''), mean imputation (``mean''), linear interpolation (``linear''), spline interpolation (``spline''), imputation with best ARIMA model (``arima'') under linear analytical models, and their corresponding versions under state space analytical models (with added ``SSM\_'' at the front), as well as the proposed state space model multiple imputation strategies -- SSMimpute and SSMmp.}
\label{fig:simulation1.2}
\end{figure} 

\begin{figure}
    \centering 
    \includegraphics[width=\linewidth]{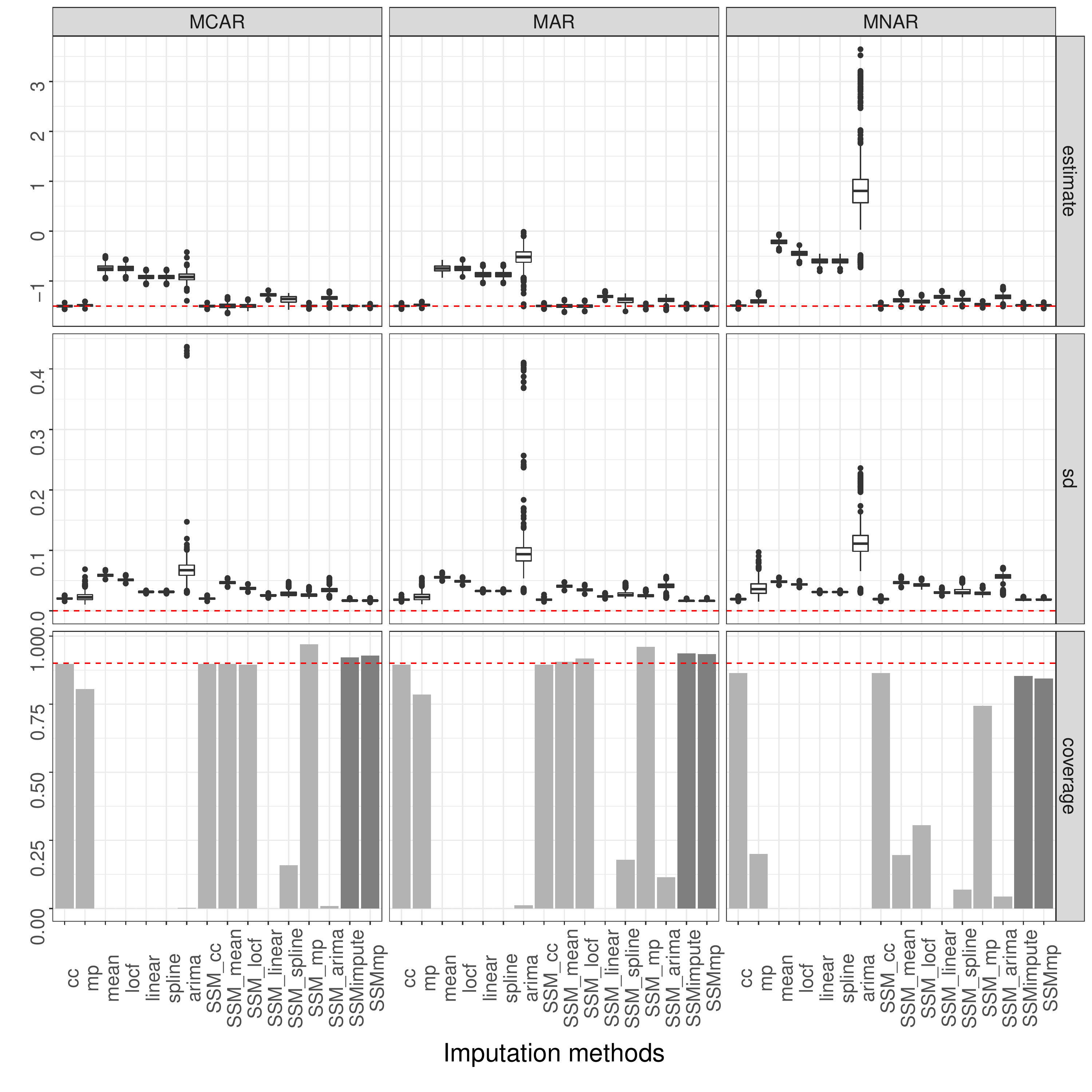}
\caption{Boxplots of the estimate (top), standard error (middle), and $90\%$ CI coverage (bottom) of the coefficient of $A_{t}$ for the stationary scenario, over 500 simulations under MCAR (left), MAR (middle) and MNAR (right) and missing rate of $50\%$. 
Methods to be compared include complete case analysis (``cc''), multiple imputation (``mp''), mean imputation (``mean''), linear interpolation (``linear''), spline interpolation (``spline''), imputation with best ARIMA model (``arima'') under linear analytical models, and their corresponding versions under state space analytical models (with added ``SSM\_'' at the front), as well as the proposed state space model multiple imputation strategies -- SSMimpute and SSMmp.}
\label{fig:simulation1.3}
\end{figure}

\begin{figure}
    \centering 
    \includegraphics[width=\linewidth]{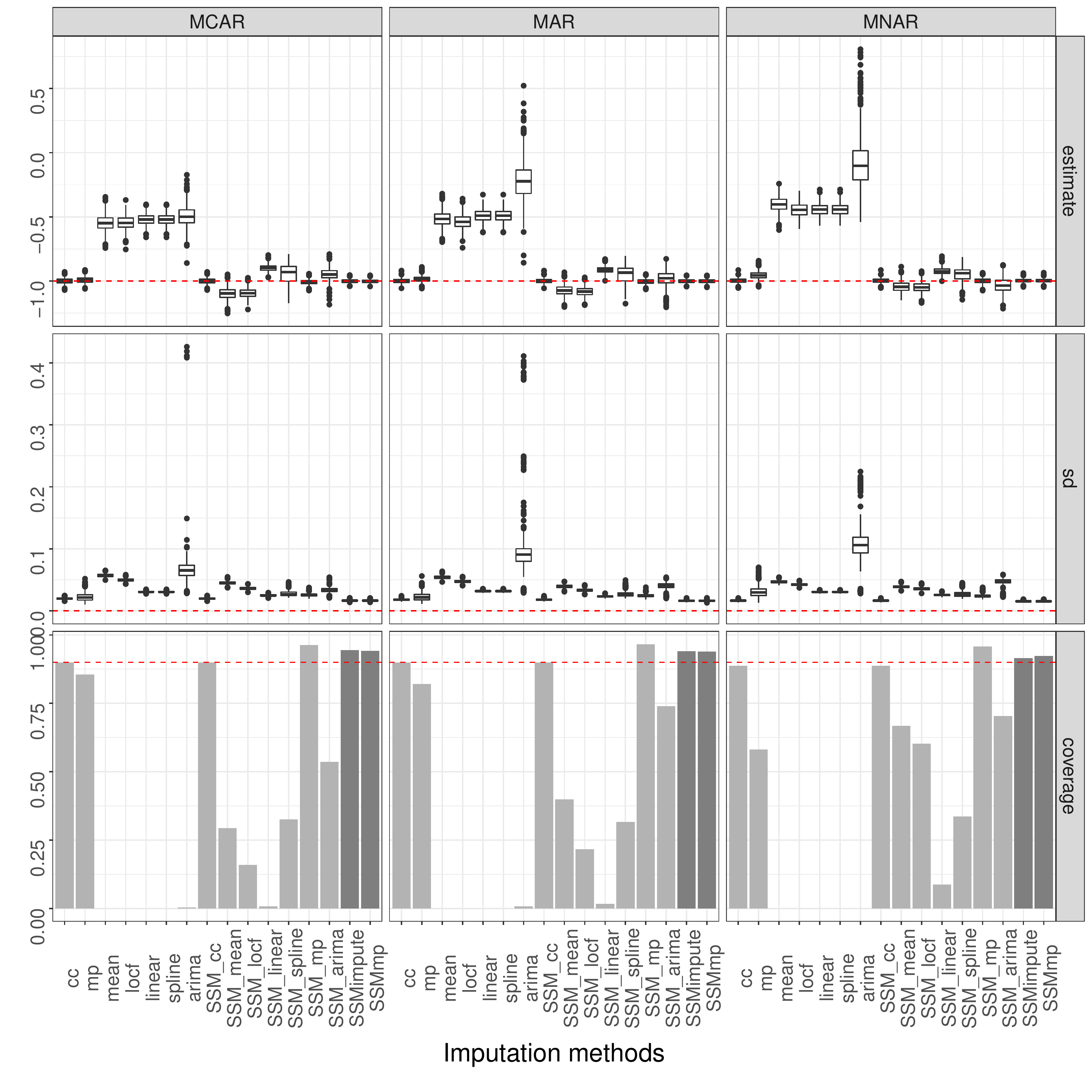}
\caption{Boxplots of the estimate (top), standard error (middle), and $90\%$ CI coverage (bottom) of the coefficient of $C_{t}$ for the stationary scenario, over 500 simulations under MCAR (left), MAR (middle) and MNAR (right) and missing rate of $50\%$.
Methods to be compared include complete case analysis (``cc''), multiple imputation (``mp''), mean imputation (``mean''), linear interpolation (``linear''), spline interpolation (``spline''), imputation with best ARIMA model (``arima'') under linear analytical models, and their corresponding versions under state space analytical models (with added ``SSM\_'' at the front), as well as the proposed state space model multiple imputation strategies -- SSMimpute and SSMmp.}
\label{fig:simulation1.4}
\end{figure} 

\subsection{Non-stationary scenario}

\begin{figure}[ht]
    \centering 
\begin{subfigure}{0.32\textwidth}
\includegraphics[width=\linewidth]{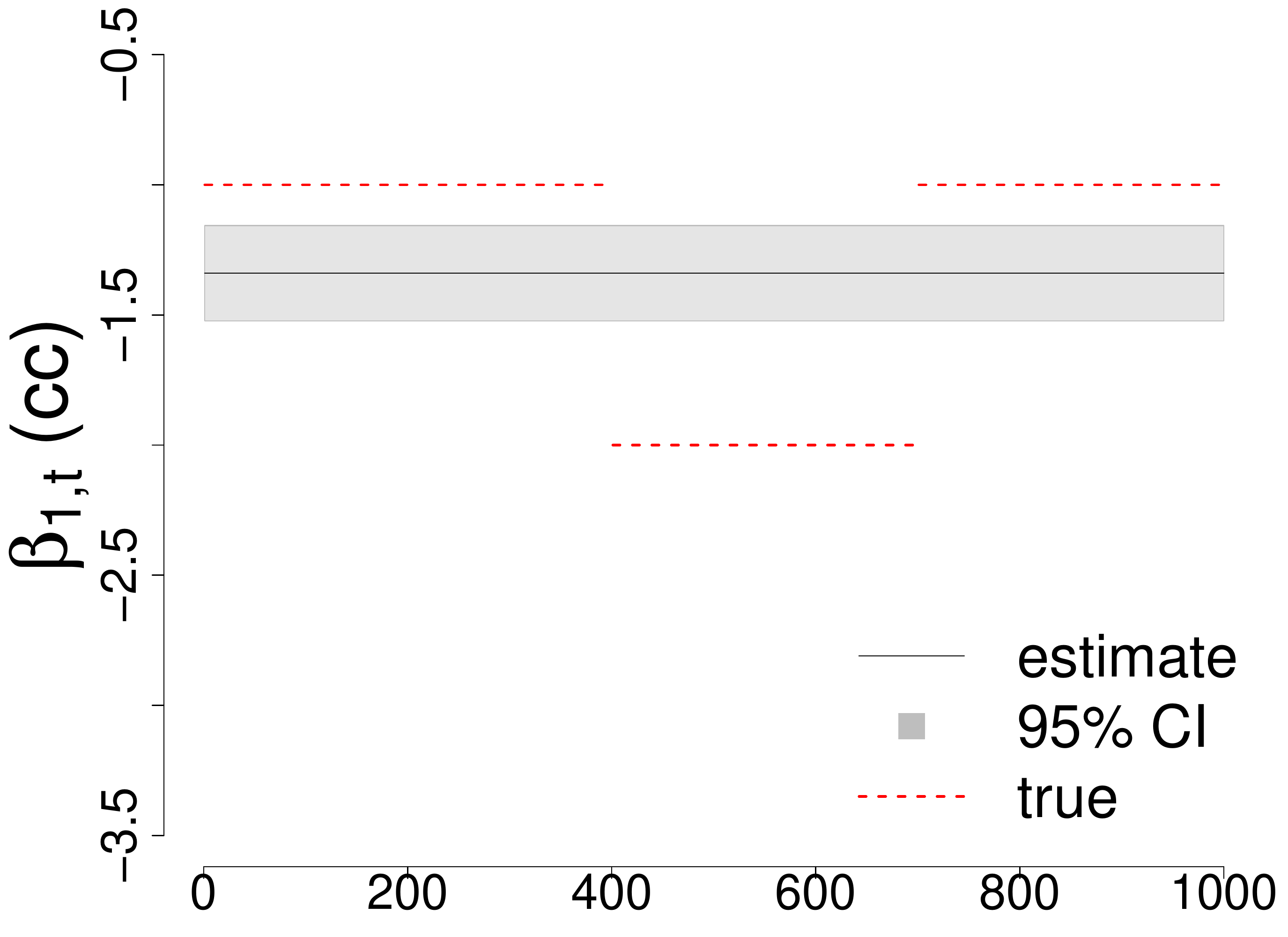}
\end{subfigure} 
\begin{subfigure}{0.32\textwidth}
  \includegraphics[width=\linewidth]{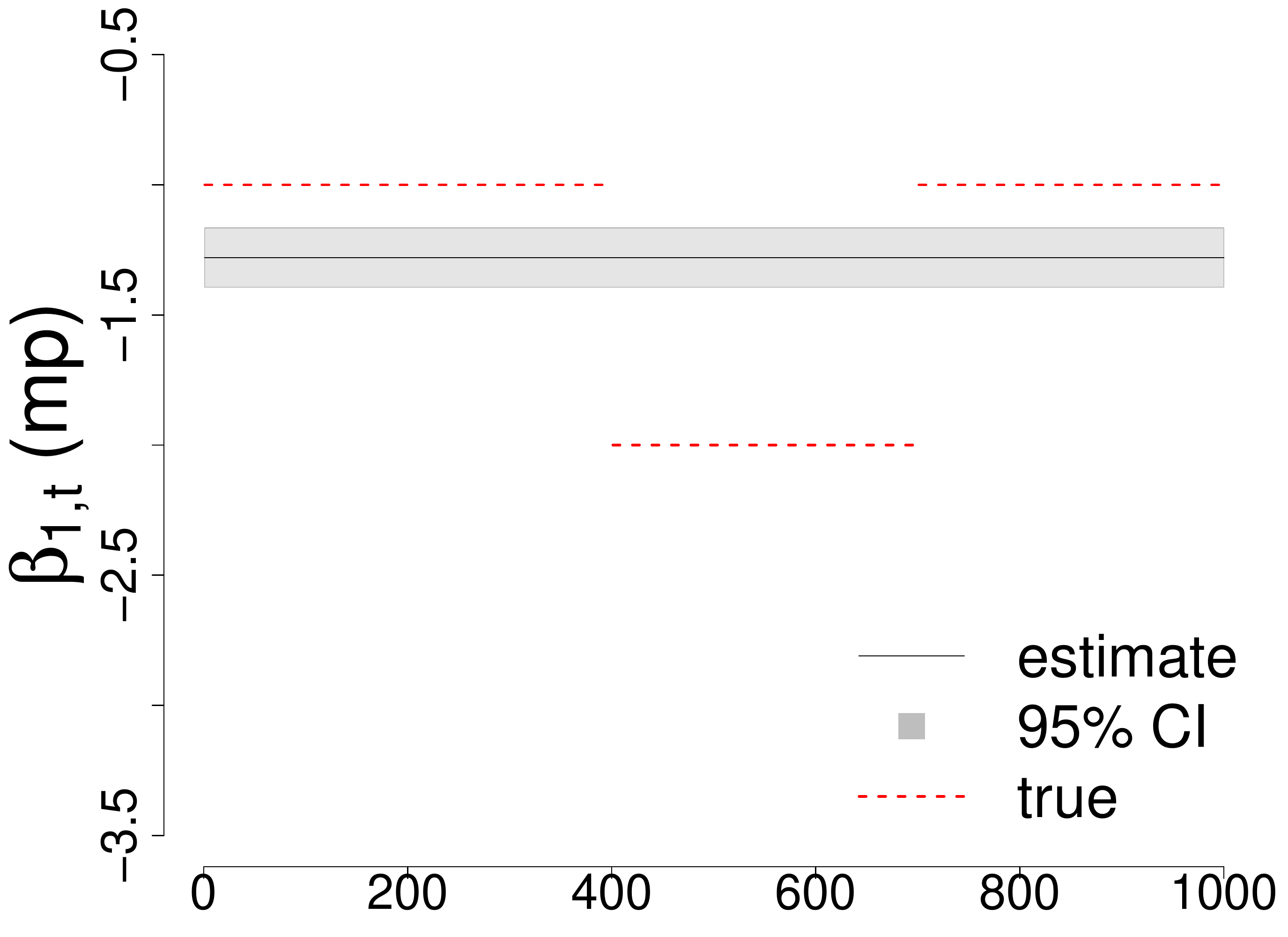}
\end{subfigure} 
\begin{subfigure}{0.32\textwidth}
  \includegraphics[width=\linewidth]{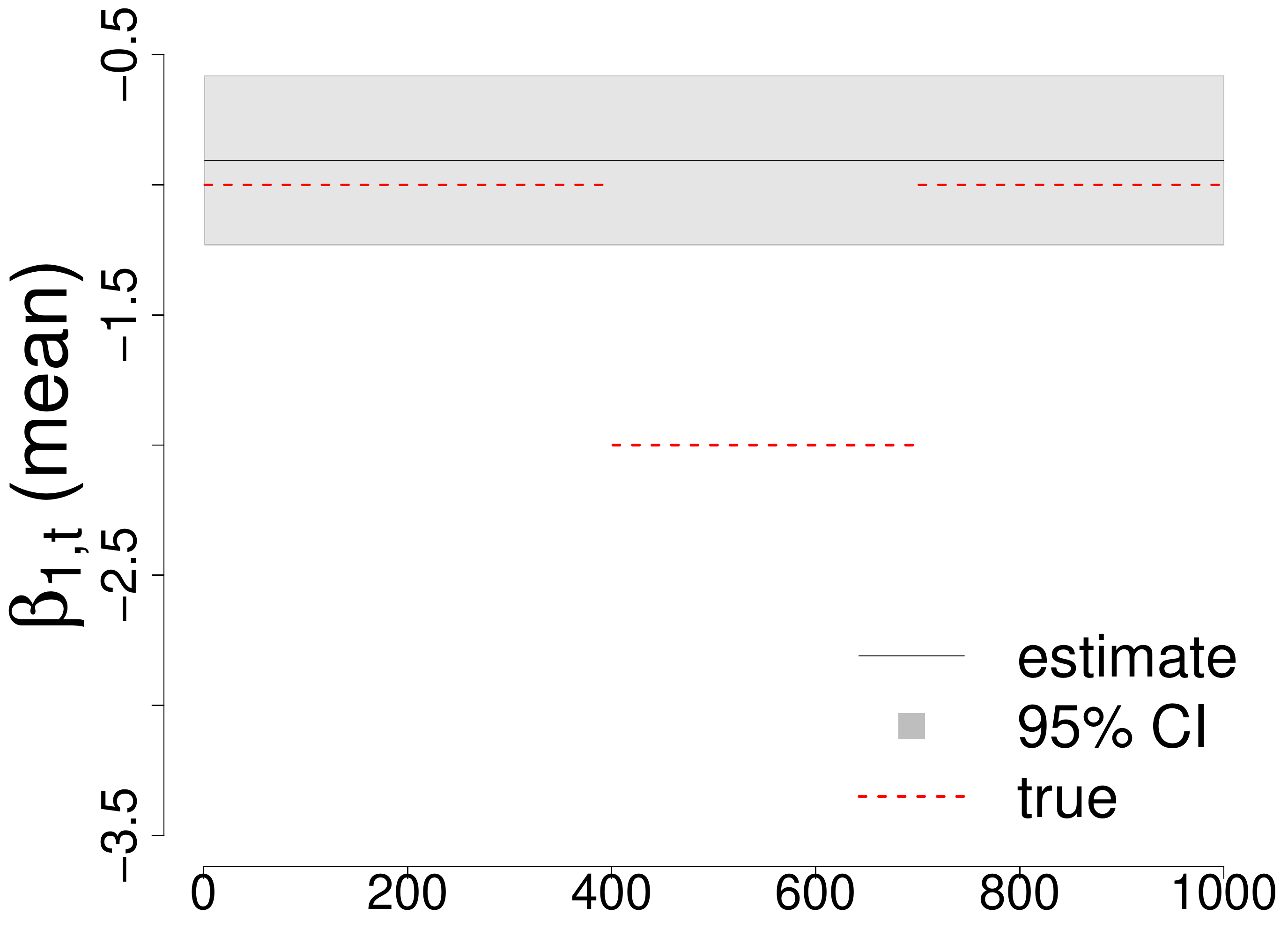}
\end{subfigure} \hfil 
\begin{subfigure}{0.32\textwidth}
  \includegraphics[width=\linewidth]{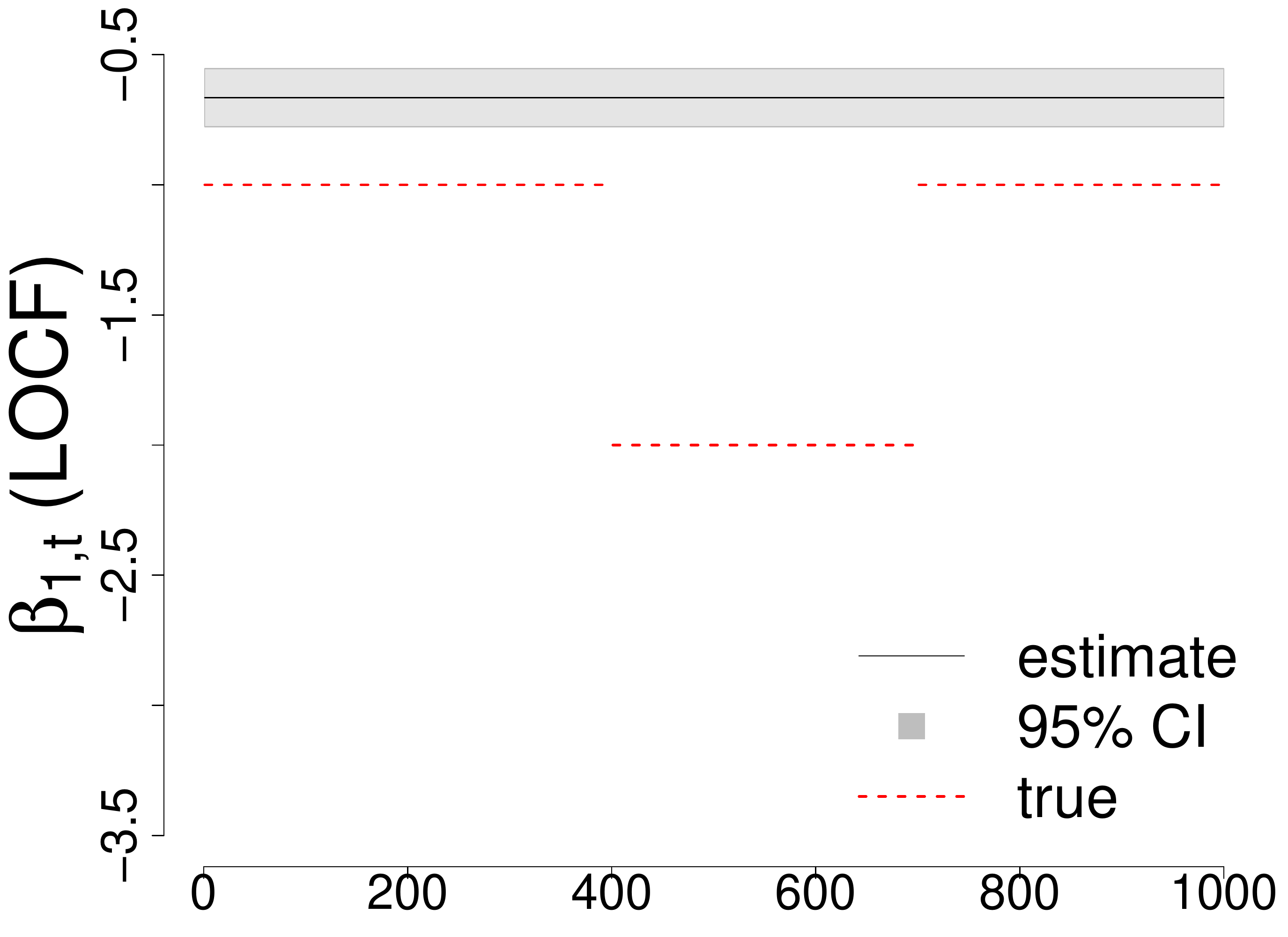}
\end{subfigure} 
\begin{subfigure}{0.32\textwidth}
  \includegraphics[width=\linewidth]{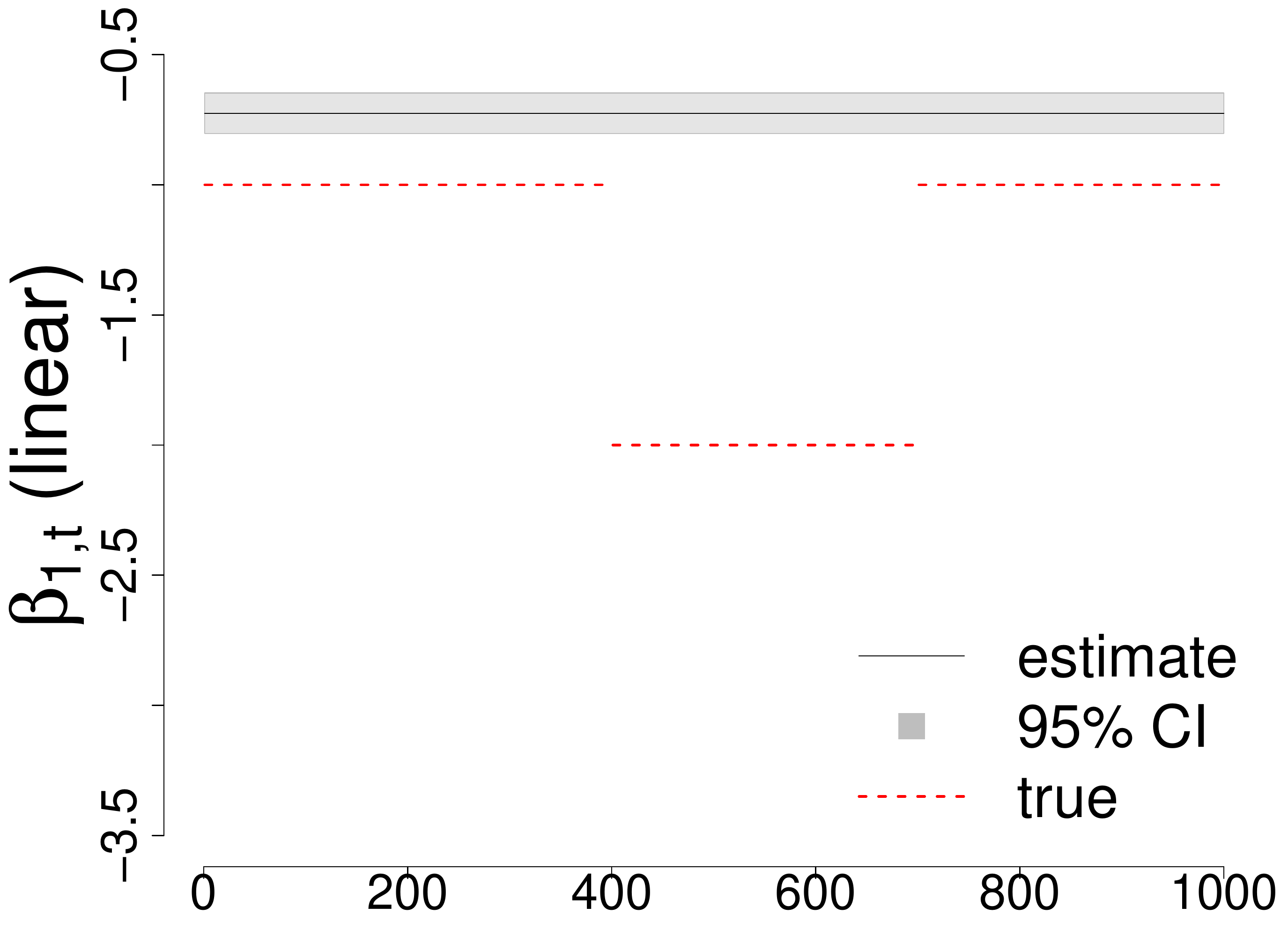}
\end{subfigure} 
\begin{subfigure}{0.32\textwidth}
  \includegraphics[width=\linewidth]{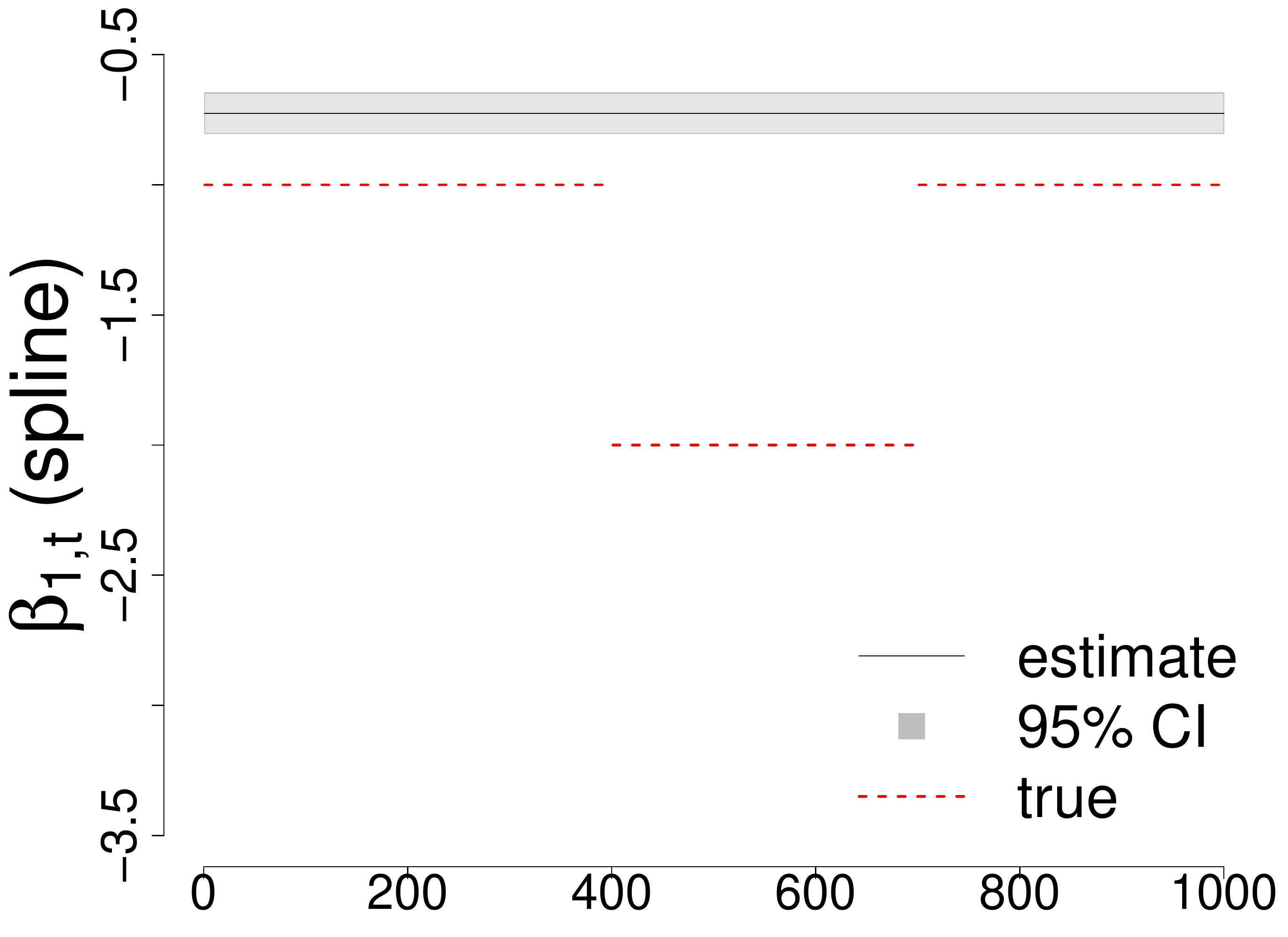}
\end{subfigure}\hfil 
\begin{subfigure}{0.32\textwidth}
  \includegraphics[width=\linewidth]{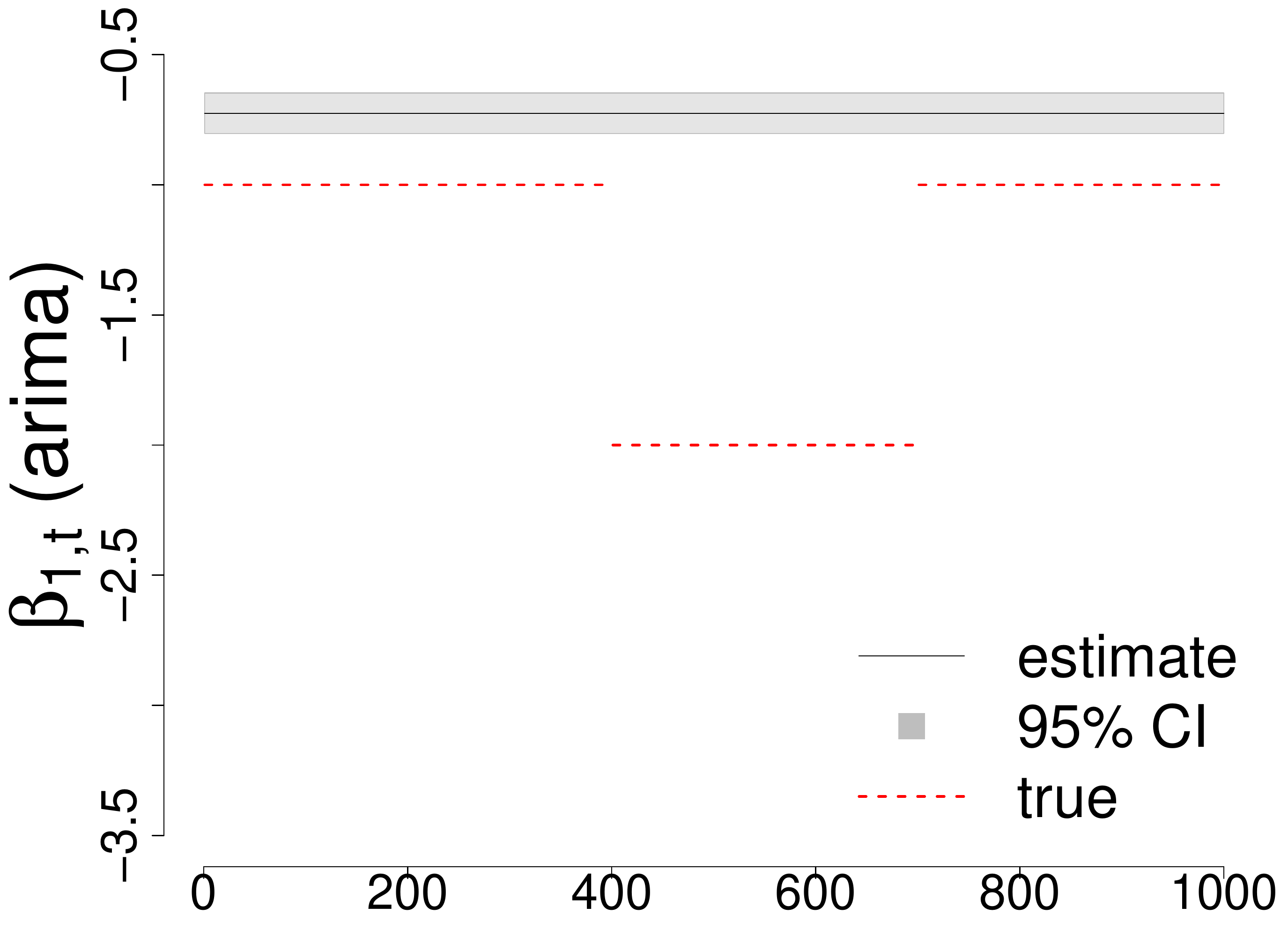}
\end{subfigure} 
\begin{subfigure}{0.32\textwidth}
  \includegraphics[width=\linewidth]{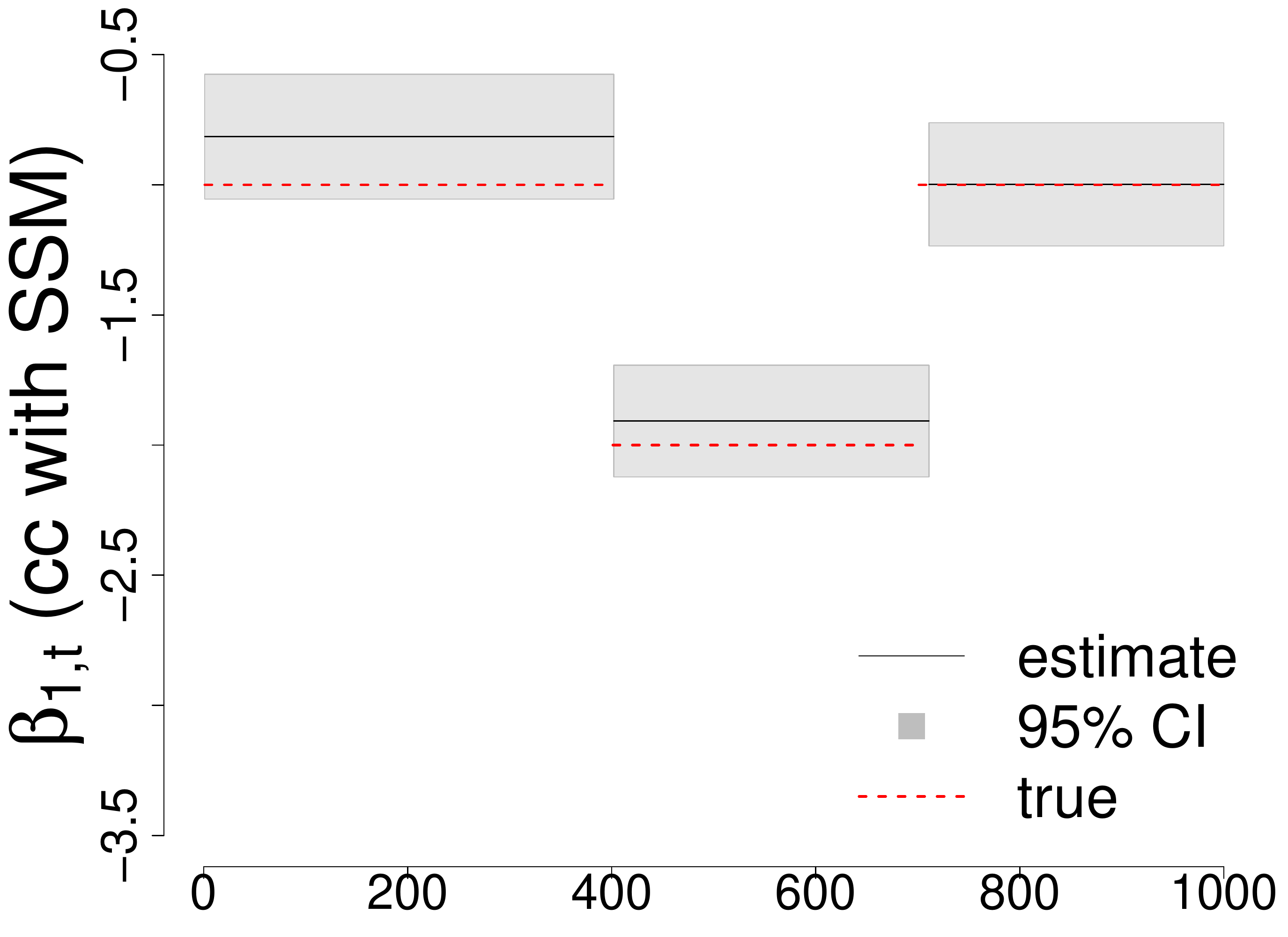}
\end{subfigure}  
\begin{subfigure}{0.32\textwidth}
  \includegraphics[width=\linewidth]{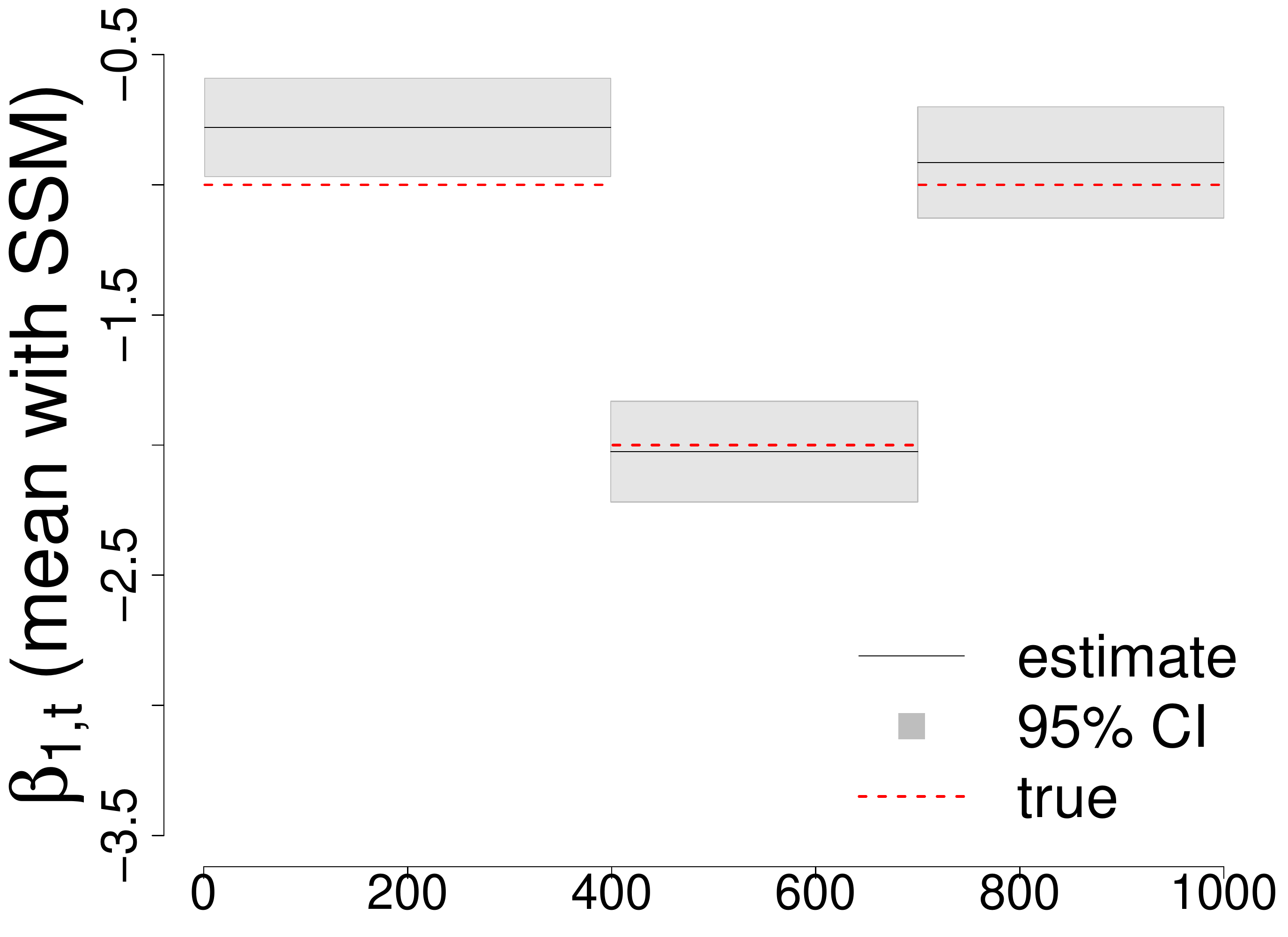}
\end{subfigure} \hfil 
\begin{subfigure}{0.32\textwidth}
  \includegraphics[width=\linewidth]{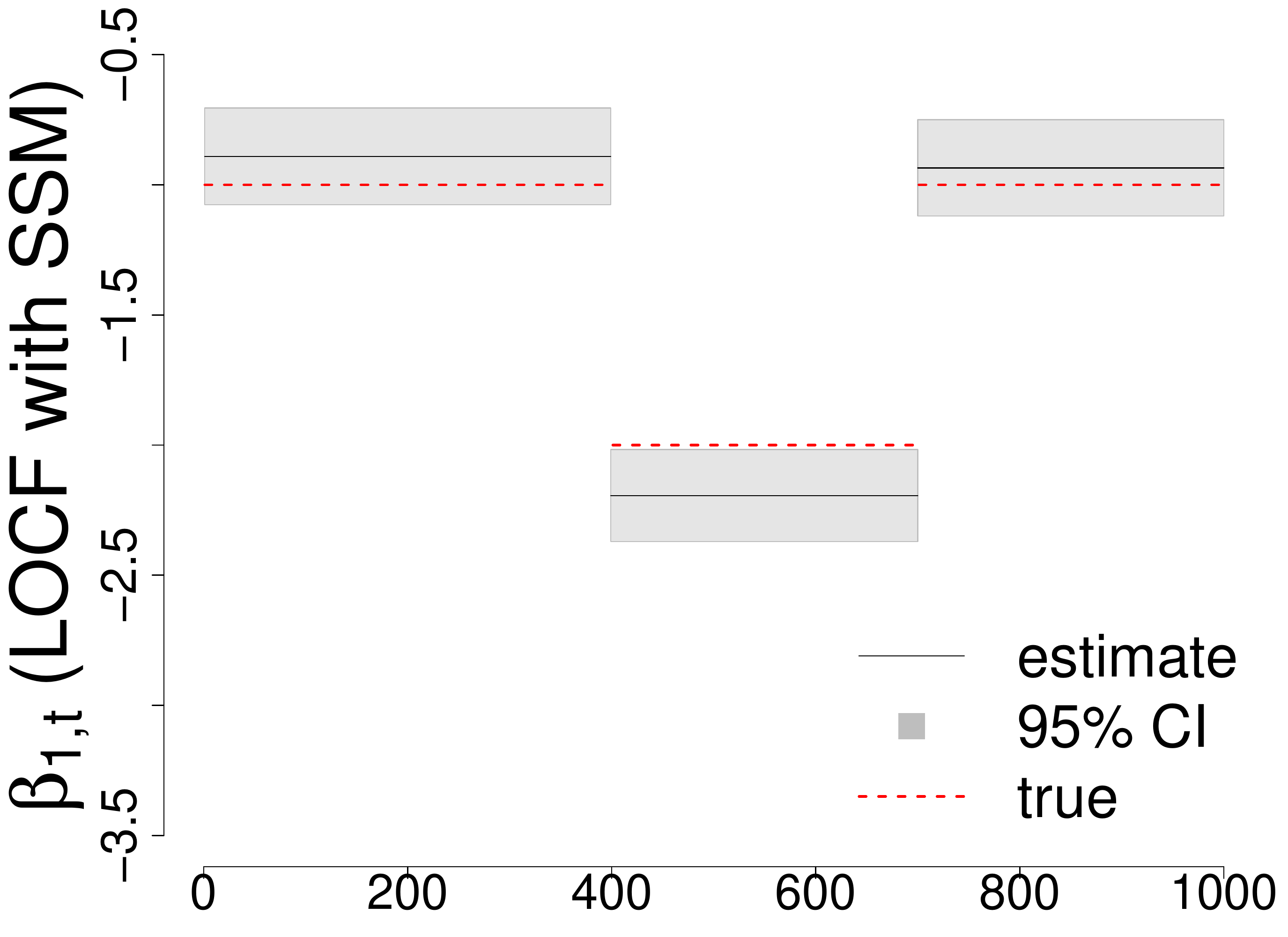}
\end{subfigure} 
\begin{subfigure}{0.32\textwidth}
  \includegraphics[width=\linewidth]{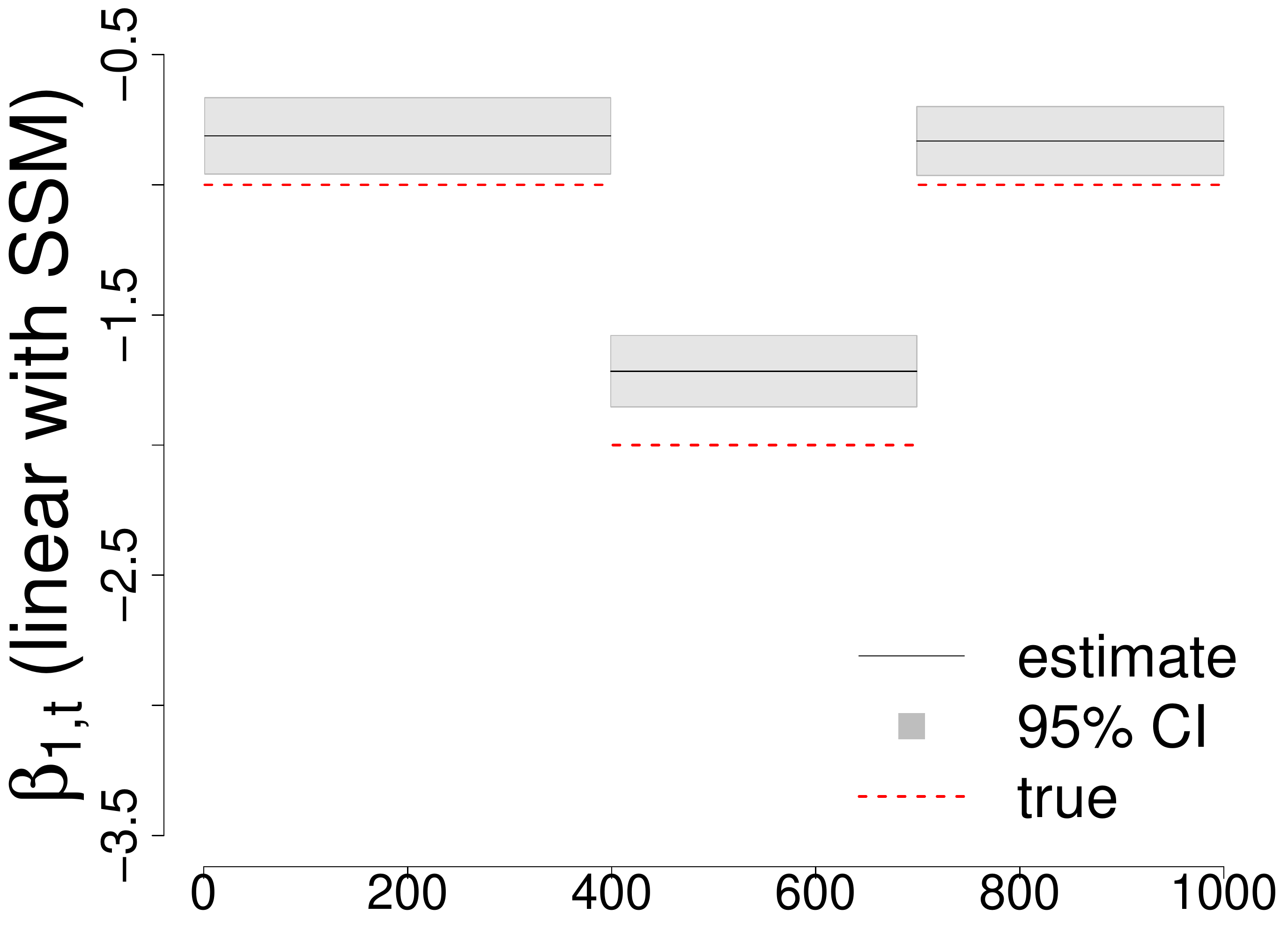}
\end{subfigure} 
\begin{subfigure}{0.32\textwidth}
  \includegraphics[width=\linewidth]{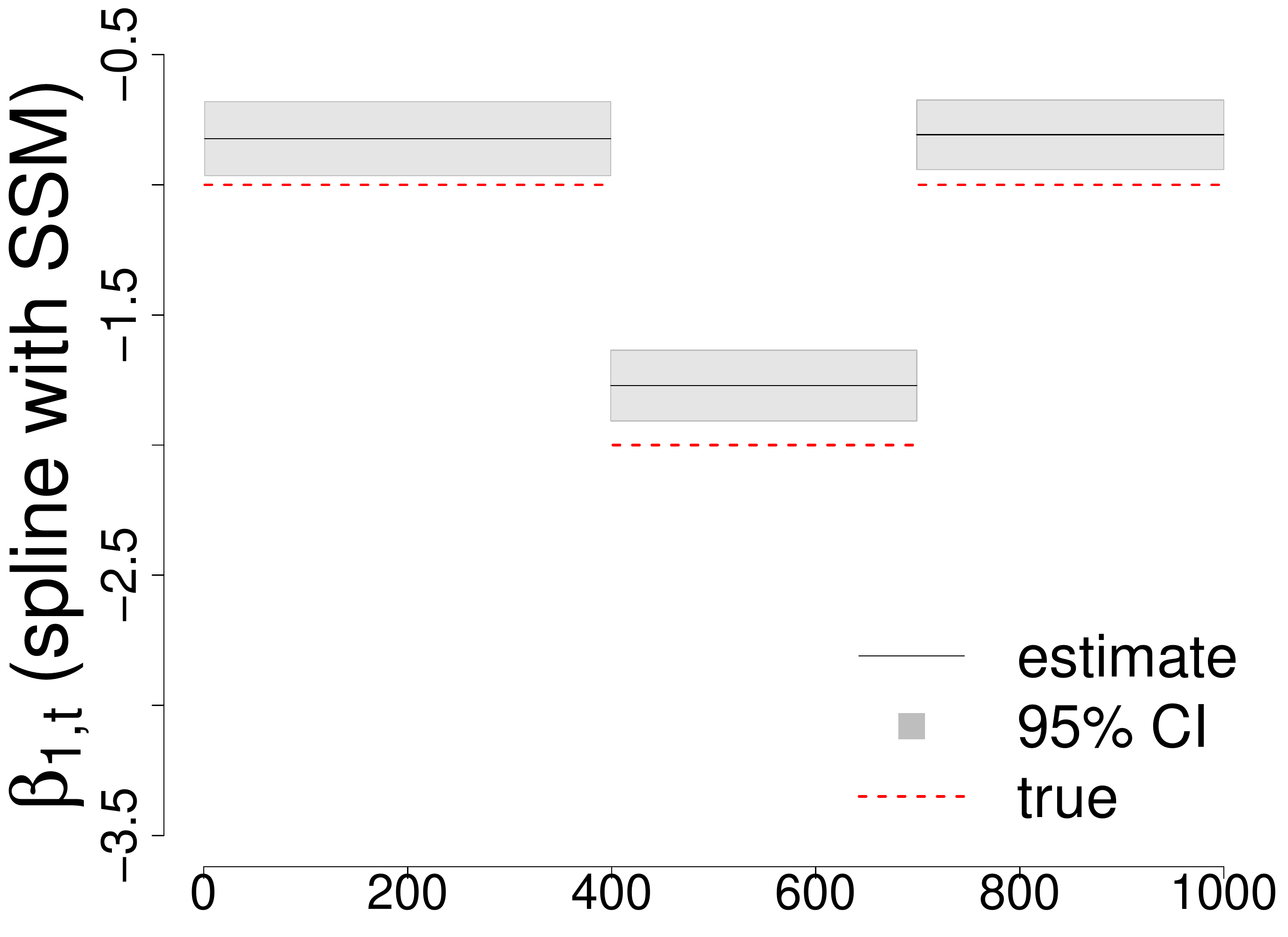}
\end{subfigure}\hfil 
\begin{subfigure}{0.32\textwidth}
  \includegraphics[width=\linewidth]{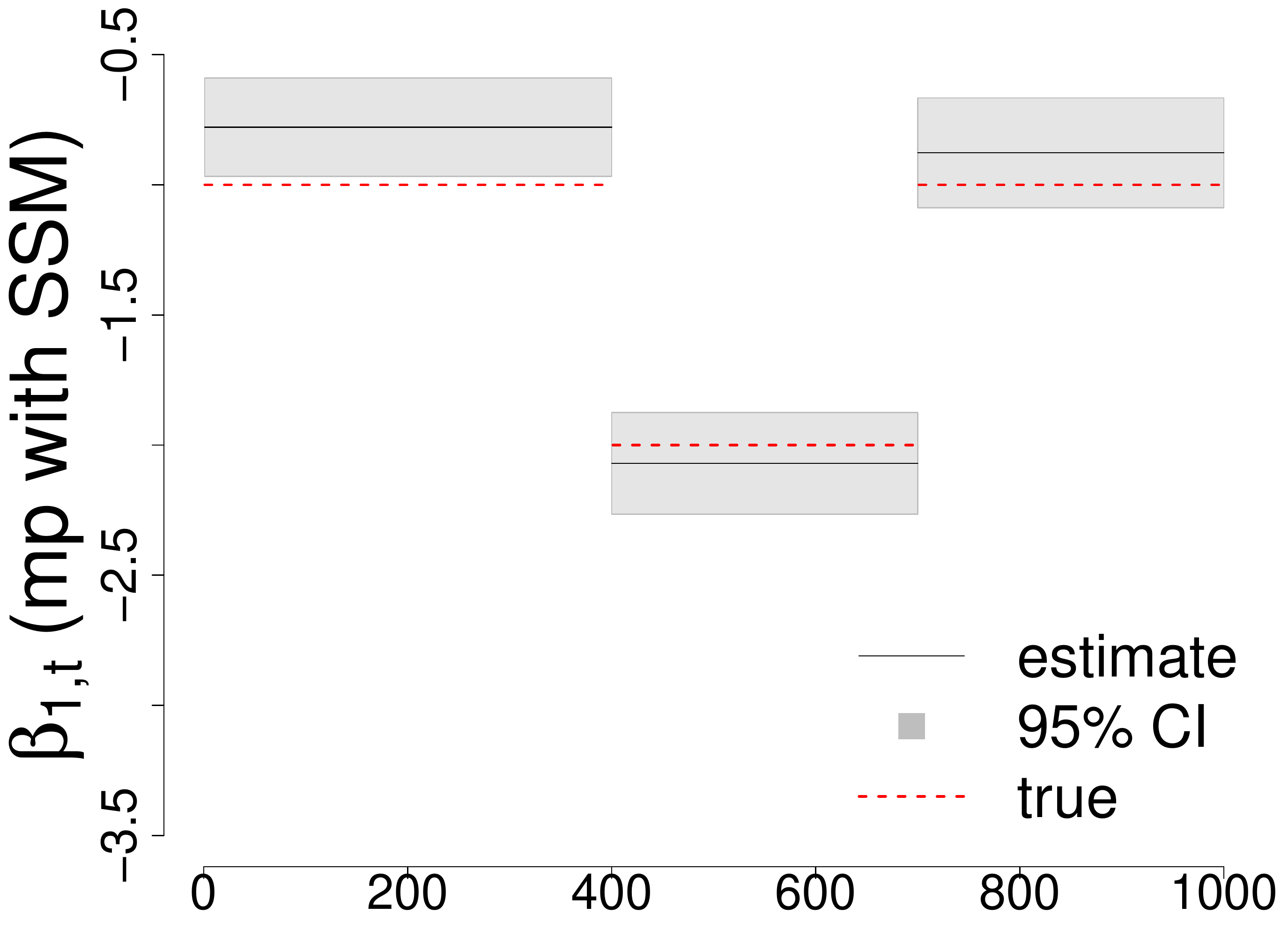}
\end{subfigure} 
\begin{subfigure}{0.32\textwidth}
  \includegraphics[width=\linewidth]{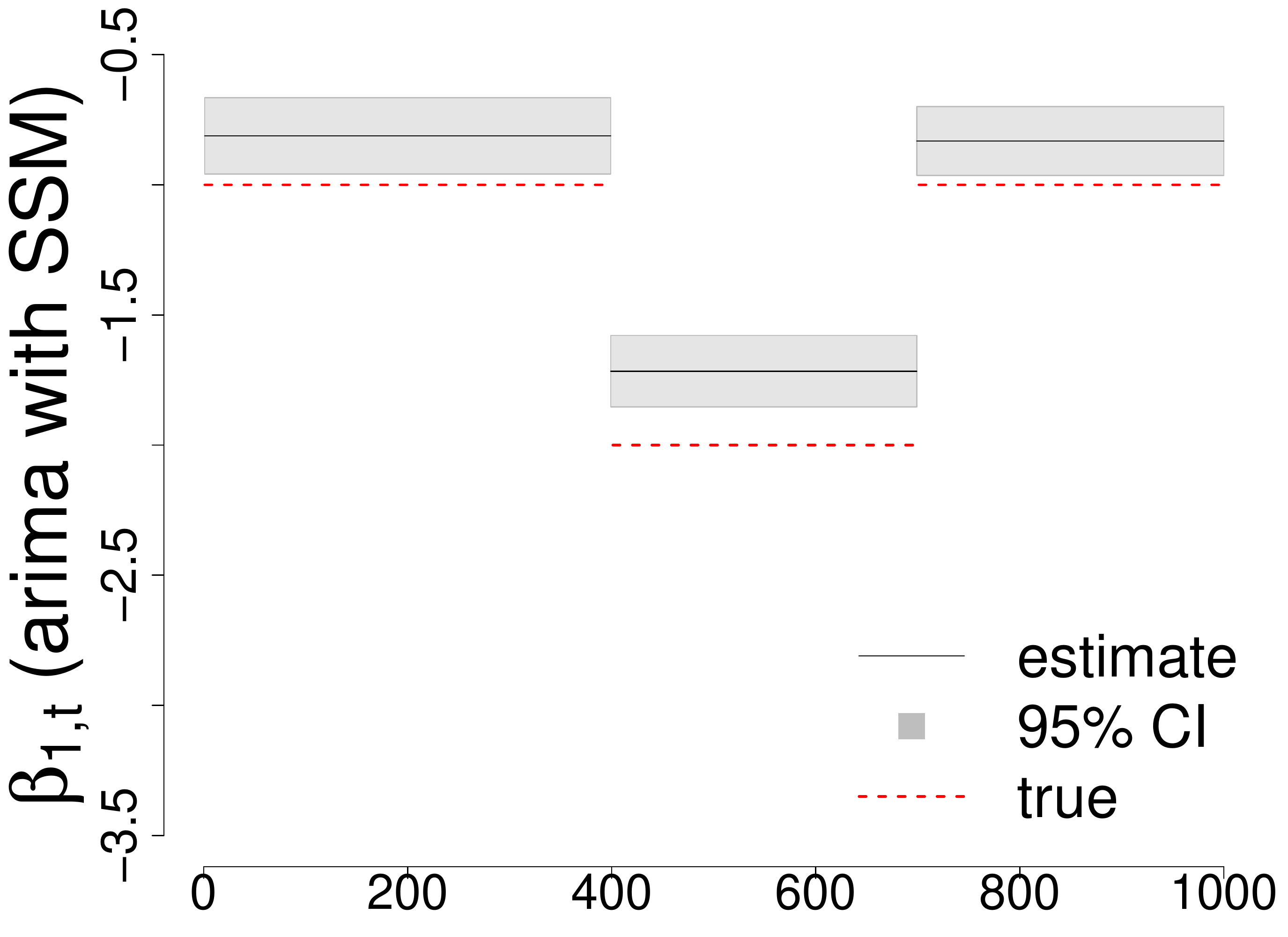}
\end{subfigure} 
\begin{subfigure}{0.32\textwidth}
  \includegraphics[width=\linewidth]{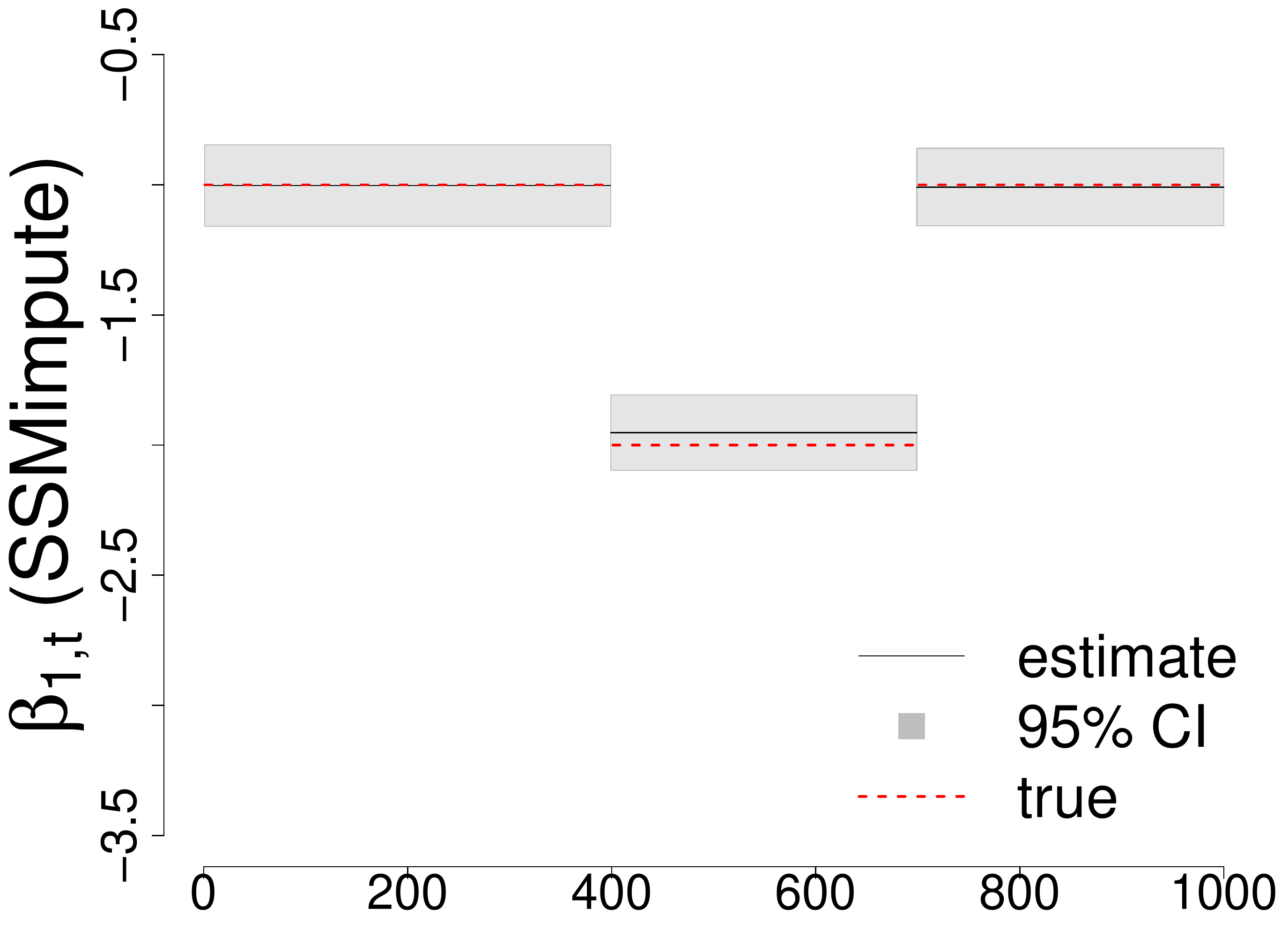}
\end{subfigure} \hfil 
\caption{Estimation of the time-varying coefficient of $A_{t}$ ($\beta_{1,t}$) over time under various missing data imputation methods and various analytical statistical models. Missing data are generated under MCAR for $25\%$. 
Methods to be compared include complete case analysis (``cc''), multiple imputation (``mp''), mean imputation (``mean''), linear interpolation (``linear''), spline interpolation (``spline''), imputation with best ARIMA model (``arima'') under linear analytical models, and their corresponding versions under state space analytical models, as well as the proposed state space model multiple imputation strategie -- SSMimpute.}
\label{fig:simulation2.2}
\vspace{-5.5pt}
\end{figure}

\begin{figure}
    \centering 
    \includegraphics[width=\linewidth]{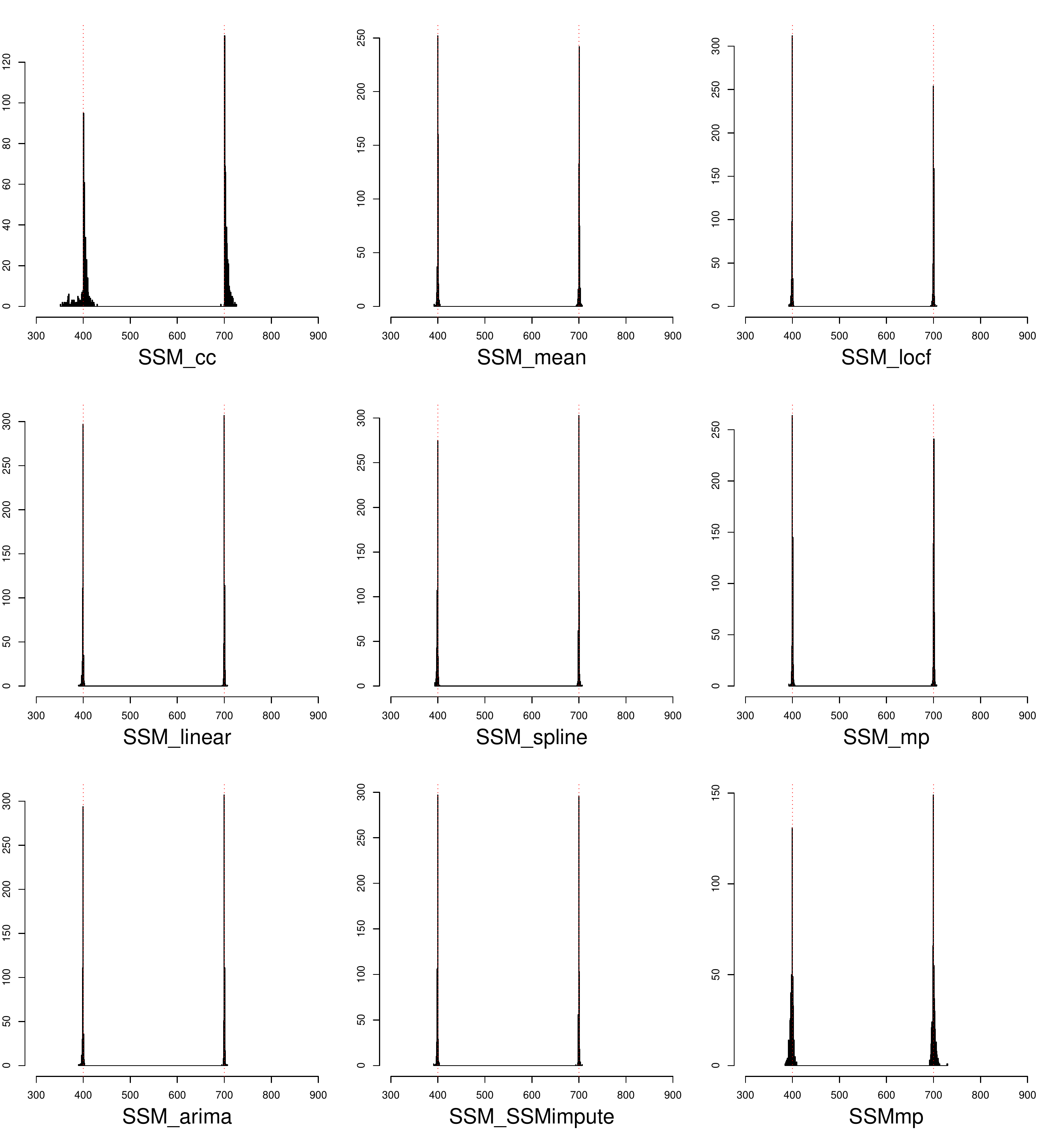}
\caption{Distribution of the estimated change points identified for the periodic-stable coefficient for $A_t$, over 500 simulations under MCAR and missing rate of $50\%$ in non-stationary scenario.
Methods to be compared include complete case analysis (``cc''), multiple imputation (``mp''), mean imputation (``mean''), linear interpolation (``linear''), spline interpolation (``spline''), imputation with best ARIMA model (``arima'') under linear models and the their corresponding versions under state space model (with added ``SSM\_'' at the front), as well as the proposed state space model multiple imputation strategies -- ``SSMimpute'' and ``SSMmp.''}
\label{fig:simulation2.8}
\end{figure} 

Figure~\ref{fig:simulation2.2} illustrates the complete trajectory of estimated three-pieced $\beta_{1,t}$ varying over time in one of the simulations shown in Figure~\ref{fig:simulation2.3} as an example. All imputation strategies with linear analytical models assumes time-invariant coefficients and fail to capture the three piece nature of $\beta_{1,t}$; imputation strategies with correctly specified state space model recover the three piece nature of $\beta_{1,t}$ with varying degree of bias, except our proposed SSMimpute imputation method. 
Figure~\ref{fig:simulation2.8} additionally illustrates distributions of identified change points for three-pieced $\beta_{1,t}$ in Figure~\ref{fig:simulation2.3} using various imputation strategies, over 500 simulations under MCAR and under missing rate of $50\%$. 
Estimation results, standard errors, and coverages for the estimated three-pieced $\beta_{1,t}$ over its 3 periods in the non-stationary scenario under MAR and MNAR are shown in Figure~\ref{fig:simulation2.4} and Figure~\ref{fig:simulation2.5}, respectively, under missing rate of $50\%$ and over 500 simulations.

\begin{figure}
    \centering 
    \includegraphics[width=\linewidth]{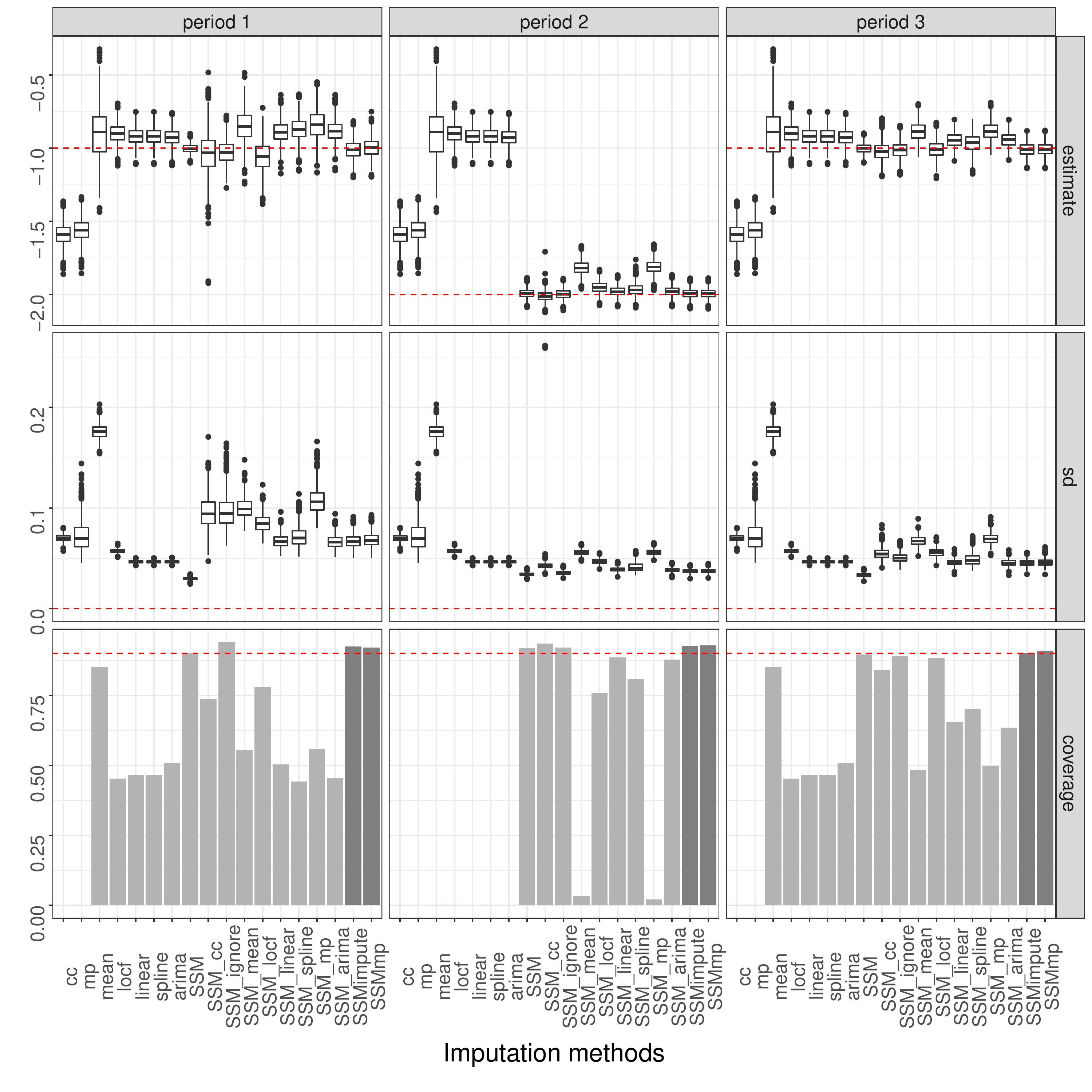}
\caption{Boxplots of the estimate (top), standard error (middle), and $90\%$ CI coverage (bottom) of the periodic-stable coefficient of $A_t$ over 3 distinct periods (from left to right) for the non-stationary scenario, over 500 simulations under MAR and missing rate of $50\%$.
Methods to be compared include complete case analysis (``cc''), multiple imputation (``mp''), mean imputation (``mean''), linear interpolation (``linear''), spline interpolation (``spline''), imputation with best ARIMA model (``arima'') under linear analytical models, and their corresponding versions under state space analytical models (with added ``SSM\_'' at the front), as well as the proposed state space model multiple imputation strategies -- SSMimpute and SSMmp.}
\label{fig:simulation2.4}
\end{figure} 
\begin{figure}
    \centering 
    \includegraphics[width=\linewidth]{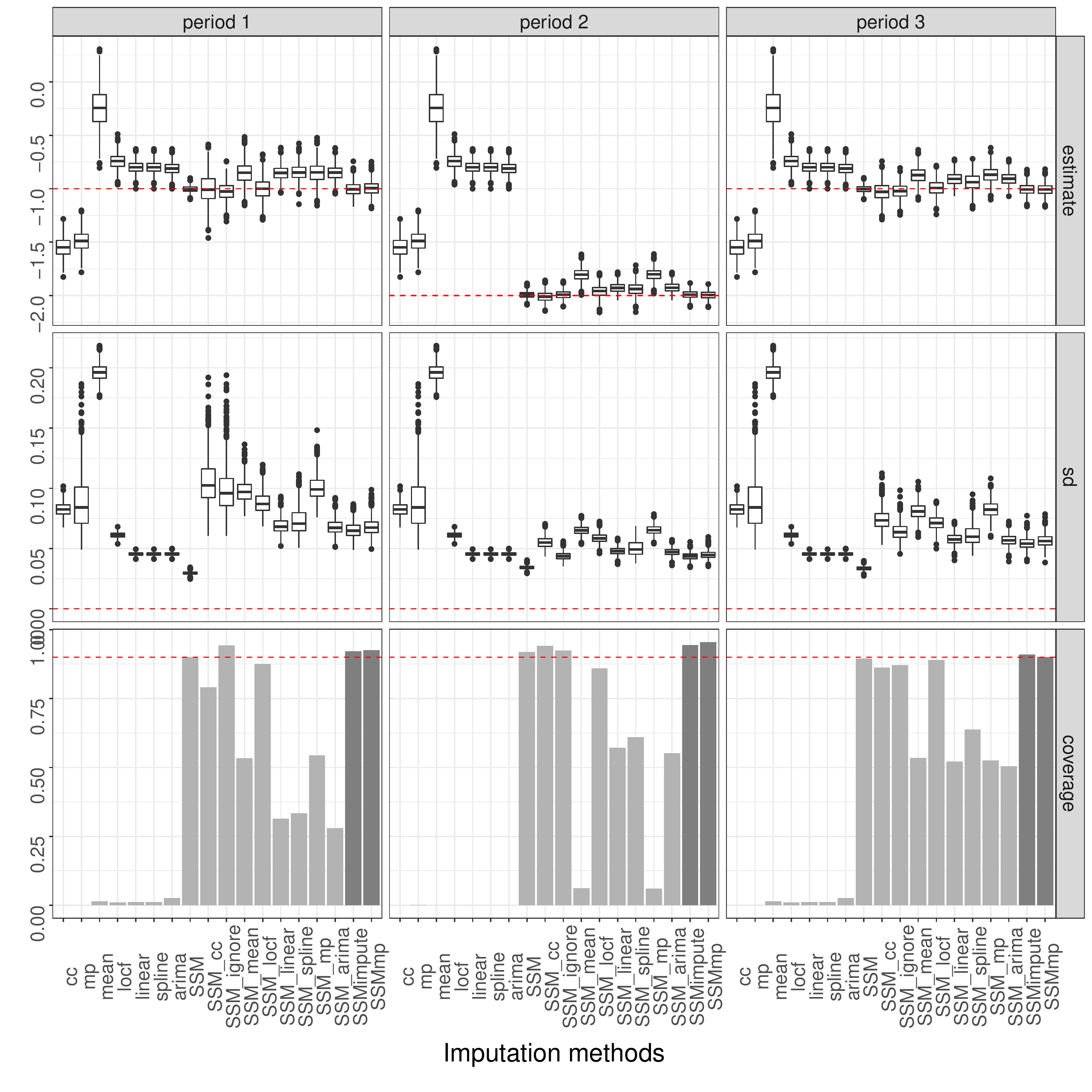}
\caption{Boxplots of the estimate (top), standard error (middle), and $90\%$ CI coverage (bottom) of the periodic-stable coefficient of $A_t$ over 3 distinct periods (from left to right) for the non-stationary scenario, over 500 simulations under MNAR and missing rate of $50\%$.
Methods to be compared include complete case analysis (``cc''), multiple imputation (``mp''), mean imputation (``mean''), linear interpolation (``linear''), spline interpolation (``spline''), imputation with best ARIMA model (``arima'') under linear analytical models, and their corresponding versions under state space analytical models (with added ``SSM\_'' at the front), as well as the proposed state space model multiple imputation strategies -- SSMimpute and SSMmp.}
\label{fig:simulation2.5}
\end{figure} 

Figures~\ref{fig:simulation2.6}--\ref{fig:simulation2.7} illustrate the estimation results, standard errors, and coverages for time-invariant coefficients of $Y_{t-1}$ and $C_t$ in the non-stationary scenario, respectively, using various imputation strategies over 500 simulations under MCAR, MAR, and MNAR and under missing rate of $50\%$.  Figures~\ref{fig:simulation2.9} illustrates the estimation results, standard errors, and coverages for time-invariant coefficients of $A_{t-1}$ in the non-stationary scenario under missing rate of $25\%$, $50\%$, and $75\%$ and MCAR, using all strategies over 500 simulation.

\begin{figure}
    \centering 
    \includegraphics[width=\linewidth]{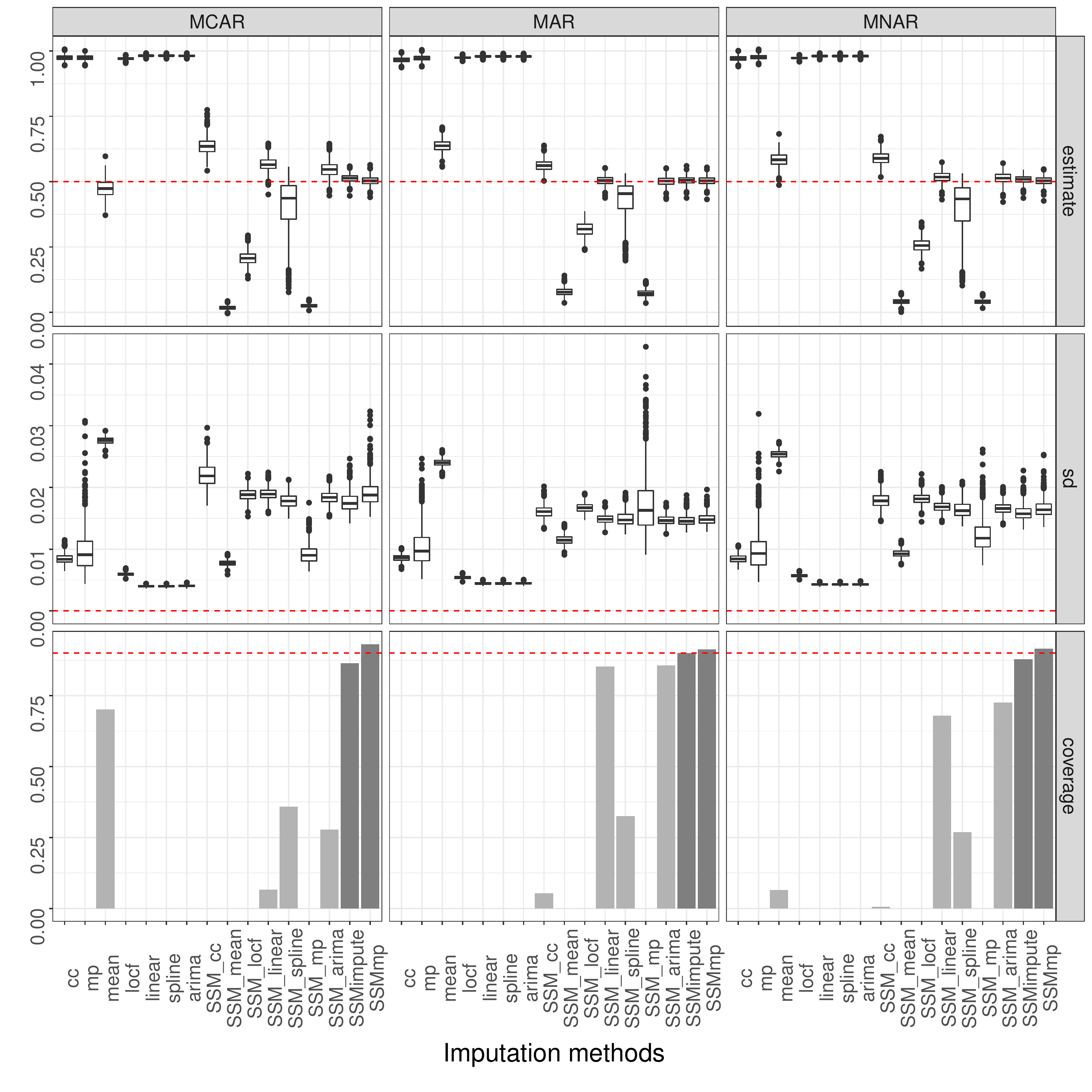}
\caption{Boxplots of the estimate (top), standard error (middle), and $90\%$ CI coverage (bottom) of the auto-correlation term of $Y_{t-1}$ for the non-stationary scenario, over 500 simulations under MCAR (left), MAR (middle), and MNAR (right), and under missing rate of $50\%$.
Methods to be compared include complete case analysis (``cc''), multiple imputation (``mp''), mean imputation (``mean''), linear interpolation (``linear''), spline interpolation (``spline''), imputation with best ARIMA model (``arima'') under linear analytical models, and their corresponding versions under state space analytical models (with added ``SSM\_'' at the front), as well as the proposed state space model multiple imputation strategies -- SSMimpute and SSMmp.}
\label{fig:simulation2.6}
\end{figure} 

\begin{figure}
    \centering 
    \includegraphics[width=\linewidth]{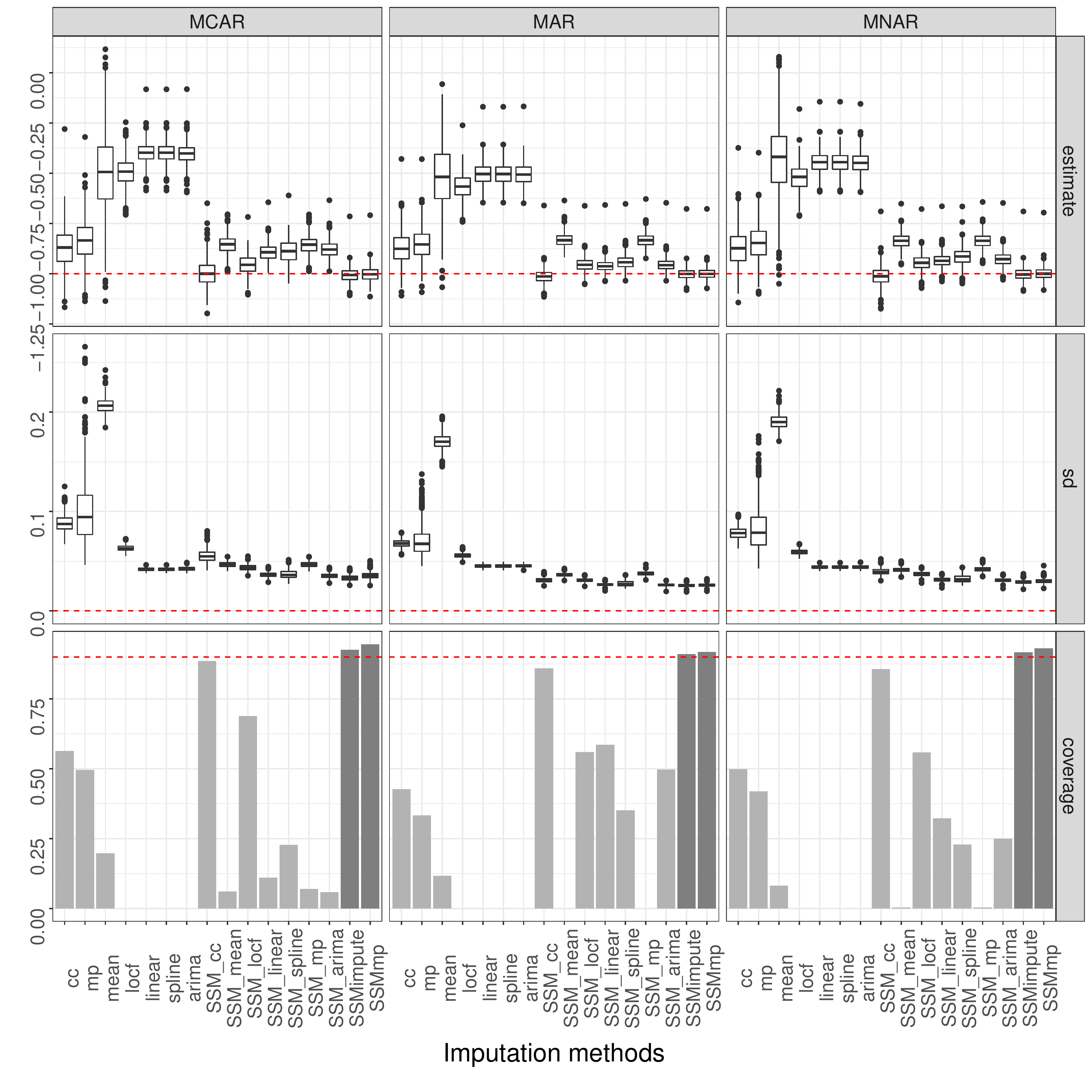}
\caption{Boxplots of the estimate (top), standard error (middle), and $90\%$ CI coverage (bottom) of the coefficient of $C_{t}$ for the non-stationary scenario, over 500 simulations under MCAR (left), MAR (middle), and MNAR (right), and under missing rate of $50\%$. 
Methods to be compared include complete case analysis (``cc''), multiple imputation (``mp''), mean imputation (``mean''), linear interpolation (``linear''), spline interpolation (``spline''), imputation with best ARIMA model (``arima'') under linear models and the their corresponding versions under state space model (with added ``SSM\_'' at the front), as well as the proposed state space model multiple imputation strategies -- ``SSMimpute'' and ``SSMmp.''}
\label{fig:simulation2.7}
\end{figure} 

\begin{figure}
    \centering 
    \includegraphics[width=\linewidth]{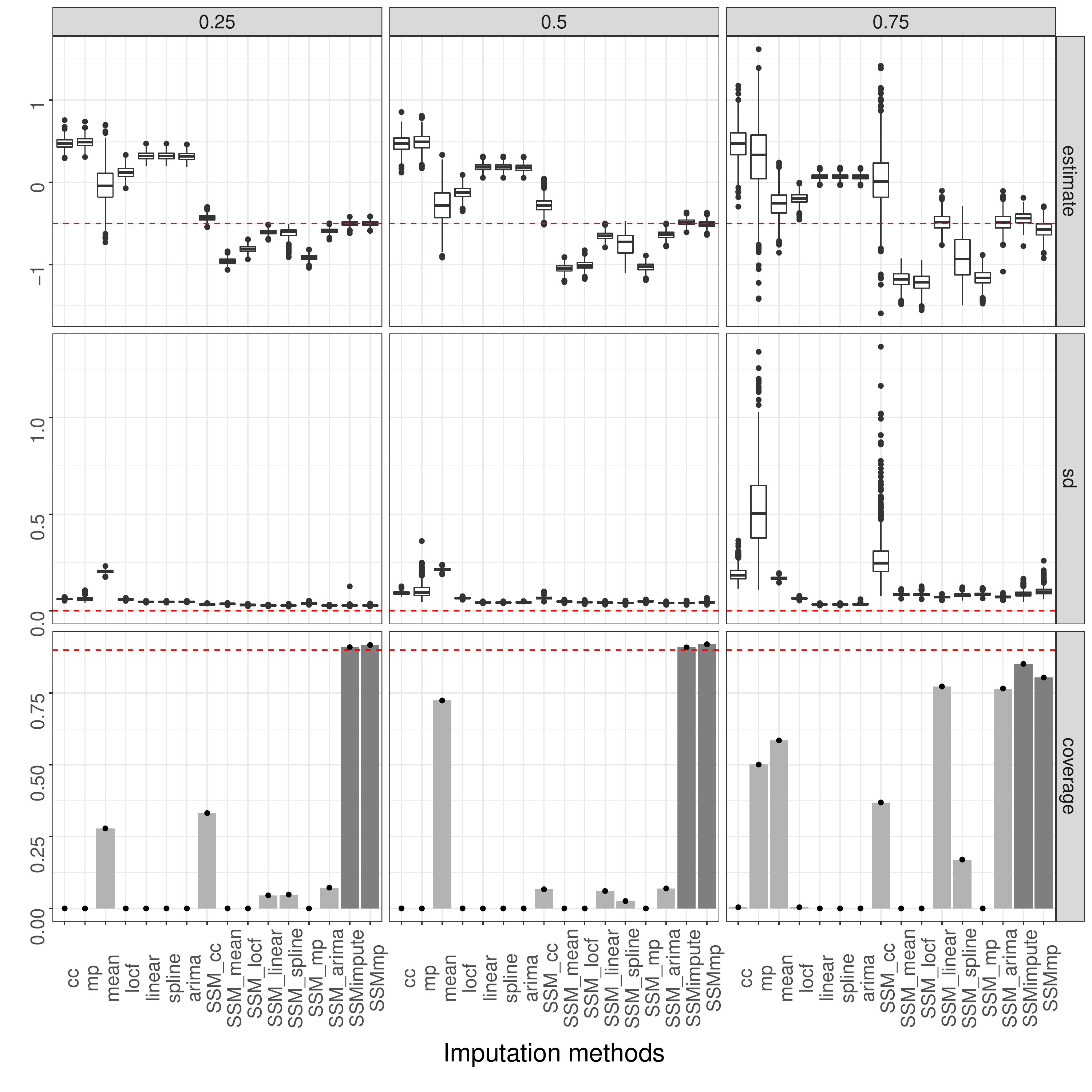}
\caption{Boxplots of the estimate (top), standard error (middle), and $90\%$ CI coverage (bottom) of the coefficient of $A_{t-1}$ for the non-stationary scenario, over 500 simulations under MCAR and missing rate of $20\%$ (left), $50\%$ (middle) and $75\%$ (right).
Methods to be compared include complete case analysis (``cc''), multiple imputation (``mp''), mean imputation (``mean''), linear interpolation (``linear''), spline interpolation (``spline''), imputation with best ARIMA model (``arima'') under linear models and the their corresponding versions under state space model (with added ``SSM\_'' at the front), as well as the proposed state space model multiple imputation strategies -- ``SSMimpute'' and ``SSMmp.''}
\label{fig:simulation2.9}
\end{figure} 

\section{Additional BLS application estimation result}
\subsection{Temperature and average physical activity during the entire follow up}

Figure~\ref{fig:pa} depicted the pre-processed physical activity level \citep{bai2012movelets,bai2014normalization}, averaged over a moving window of 15 days of the 708-day follow-up. 
Figure~\ref{fig:temp} depicts the average outdoor temperature in the patient's neighborhood over the course of 708 days of follow up.

\begin{figure}
\centering
\includegraphics[width=\linewidth]{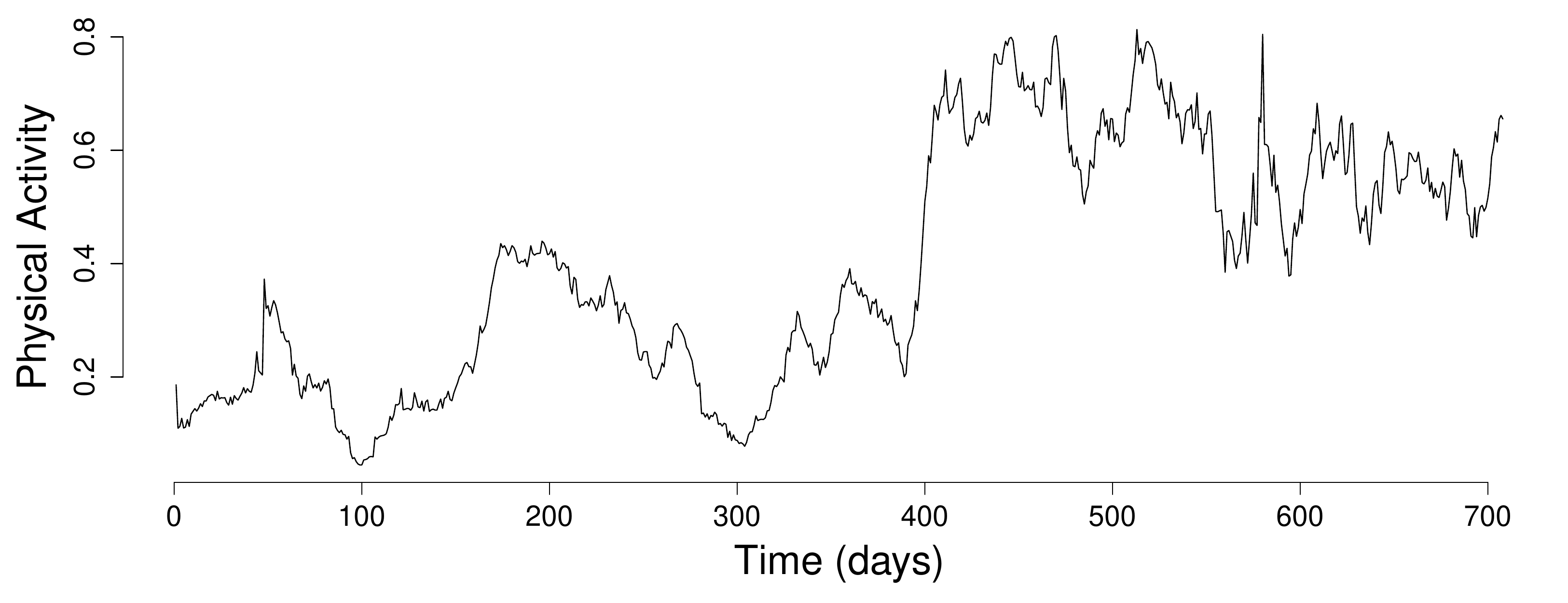}
\caption{Pre-processed average physical activity level over the 708 days of follow up}
\label{fig:pa}
\end{figure}

\begin{figure}
\centering
\includegraphics[width=\linewidth]{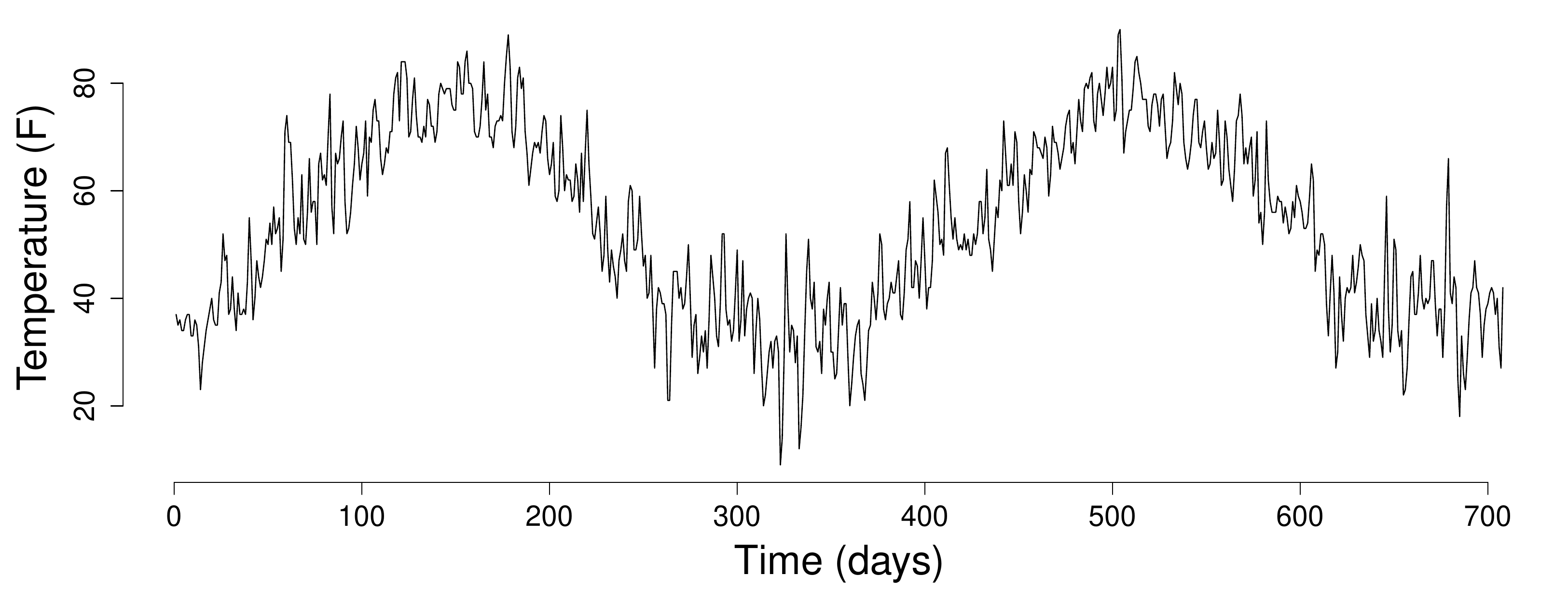}
\caption{Outdoor temperature over the 708 days of follow up}
\label{fig:temp}
\end{figure}

\subsection{State space analytical model selection}

We consider a concise model including one auto-correlation term, exposures of interest, and confounding variables on the same day. Models including additional history information are also considered as an effort to improve model fitting.
Table~\ref{tab:modelselection} shows a model comparison of state space models with varying extra history information and the original model, utilizing one step ahead predictions selection criteria \cite{rivers2002model}. 
The original model with the minimum sum of squared one-step prediction error is selected.

\begin{table}
\centering
\begin{tabular}{llc}
  \hline
\multicolumn{2}{l}{State space analytical model} & sum of squared one-step-prediction error \\ 
  \hline
\multicolumn{3}{l}{$y_t \sim \text{intercept} + Y_{t-1} +  A_{calls,t} + A_{calls,t-1} + A_{text,t} +  A_{texts,t-1} + \text{Temp}_t + \text{PA}_t $} \\ 
& original model & 2151.42 \\
& add $Y_{t-2}$ & 2162.21\\ 
& add $A_{calls,t-2}$ & 2161.45\\ 
& add $A_{texts,t-2}$ & 2158.88 \\ 
  \hline
\end{tabular}
\caption{State space model selection using sum or squared one-step-prediction error.}
\label{tab:modelselection}
\end{table}

\subsection{Estimated coefficients trajectory 
over time using SSMimpte}
Figure~\ref{fig:ssm_intercept} illustrates the estimated time-varying random walk intercept over time using SSMimpute. 

Figure~\ref{fig:ssm_rho} illustrate the estimated auto-correlation trajectory over time using SSMimpute, which is modeled as a time-invariant coefficient.

Figures~\ref{fig:ssm_calls} and \ref{fig:ssm_calls2} illustrate the estimated coefficient trajectories for the degree of call contacts on the same day and on the previous day, respectively, over time using SSMimpute, which are modeled as time-invariant coefficients.
 
Figure~\ref{fig:ssm_text2} illustrate the estimated coefficient trajectory for the degree of text contacts on the previous day over time using SSMimpute, which is modeled as a time-invariant coefficient.

Figure~\ref{fig:ssm_pa} illustrate the estimated coefficient trajectory for the averaged physical activity over time using SSMimpute, which is modeled as a time varying three-pieced coefficient with detected change points around day 162 and day 267.

Figure~\ref{fig:ssm_temp} illustrate the estimated coefficient trajectory for outdoor temperature over time using SSMimpute, which is modeled as a time-invariant coefficient.

\begin{figure}
\includegraphics[width=\linewidth]{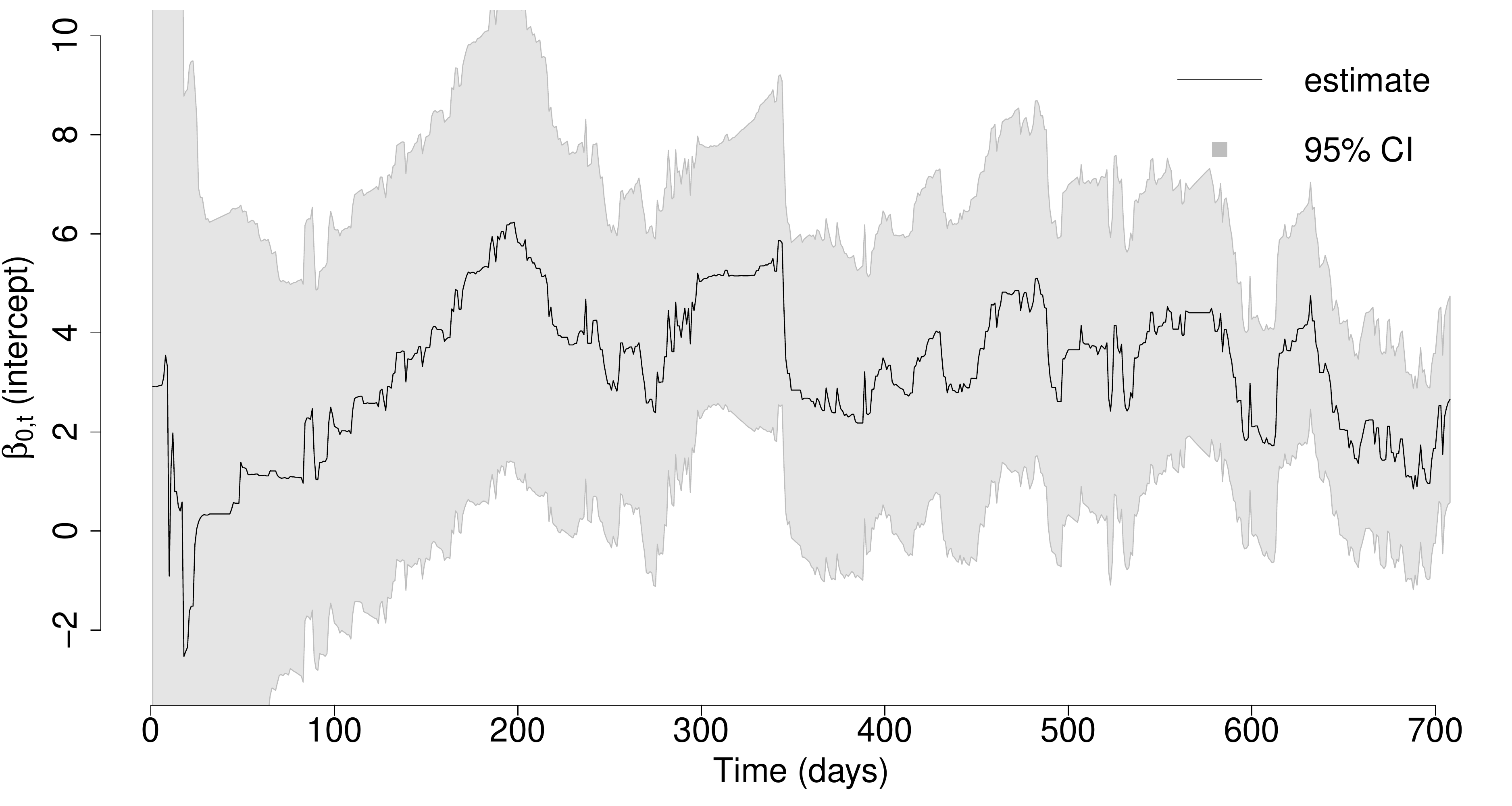}
\caption{Estimated intercept $\beta_{0,t}$ over time as a random walk using SSMimpute. The black line represents the estimated coefficient $\beta_{0,t}$ over time, while the grey area represents the $95\%$ confidence interval.}
\label{fig:ssm_intercept}
\end{figure}

\begin{figure}
\includegraphics[width=\linewidth]{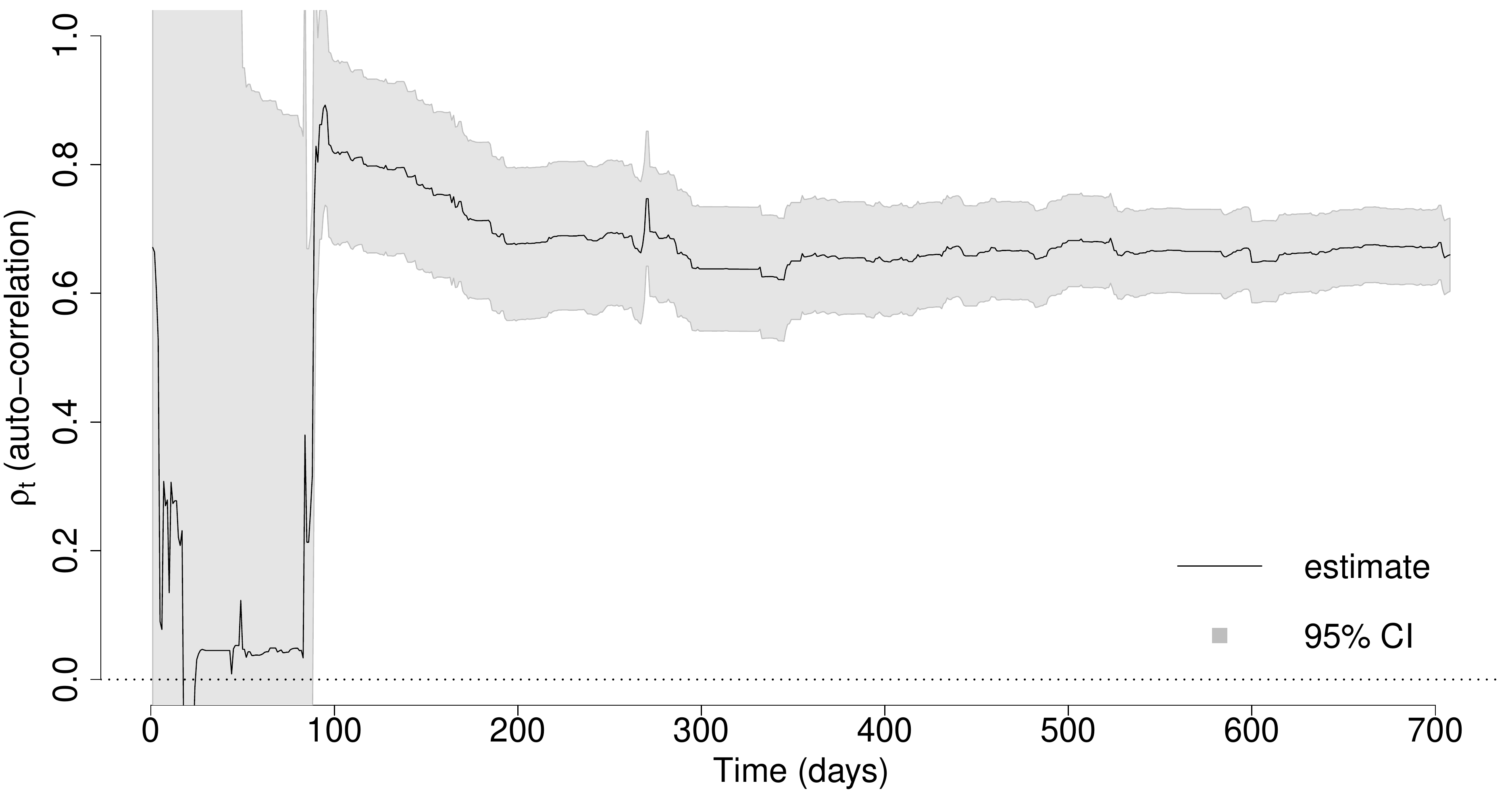}
\caption{Estimated auto-correlation $\rho_t$ between $Y_t$ and $Y_{t-1}$ over time using SSMimpute. The black line represents the estimated coefficient $\rho_t$ over time, while the grey area represents the $95\%$ confidence interval. The dashed horizontal line at 0 indicates null effect.}
\label{fig:ssm_rho}
\end{figure}

\begin{figure}
\includegraphics[width=\linewidth]{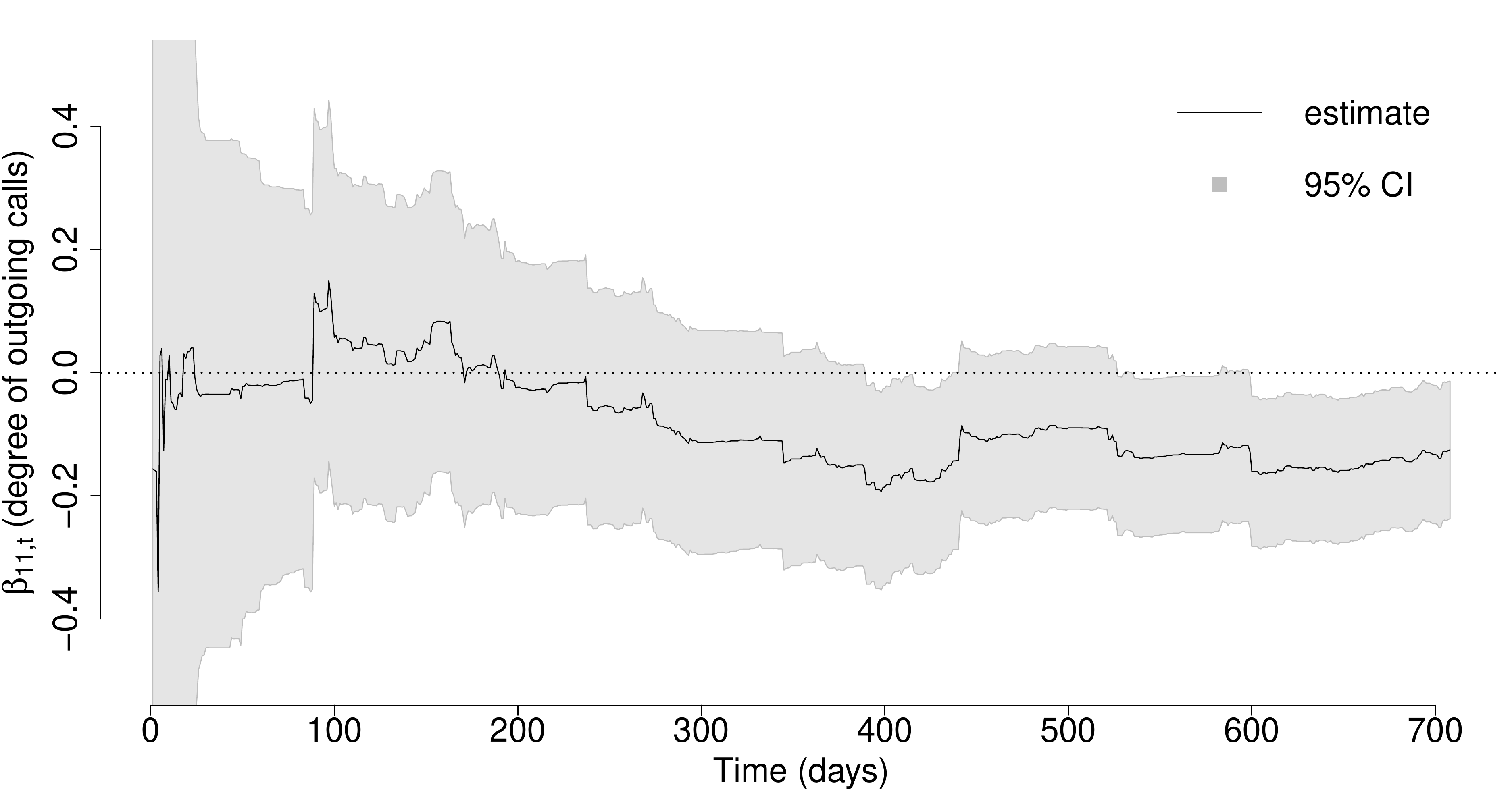}
\caption{Estimated coefficient trajectory over time for the degree of calls $\beta_{11,t}$ using SSMimpute. The black line represents the estimated coefficient over time, while the grey line represents the $95\%$ confidence interval. The dashed horizontal line at 0 indicates that there is no effect.}
\label{fig:ssm_calls}
\end{figure}

\begin{figure}
\includegraphics[width=\linewidth]{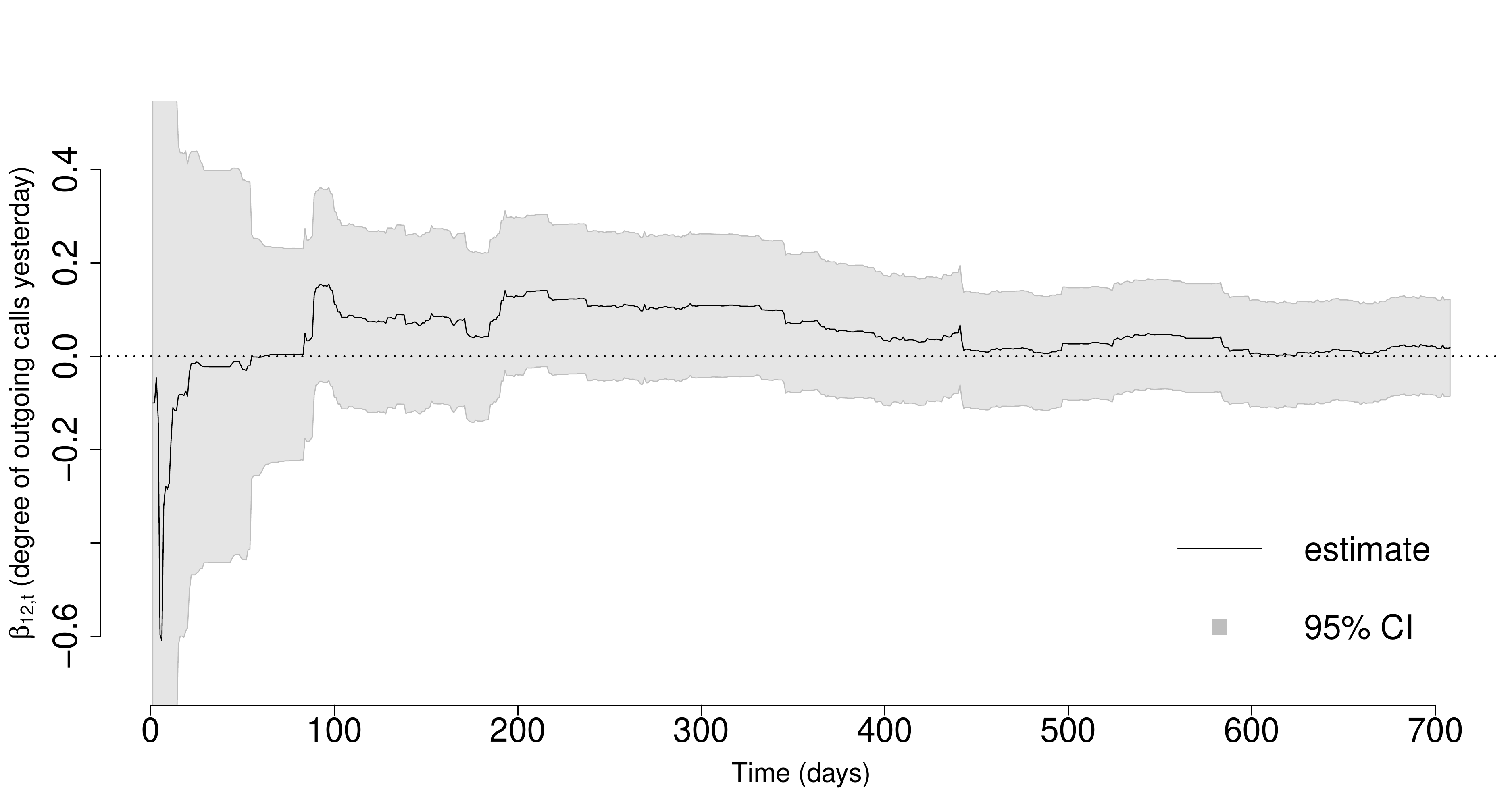}
\caption{Estimated coefficient trajectory over time for the degree of calls at the previous day $\beta_{12,t}$ using SSMimpute. The black line represents the estimated coefficient over time, while the grey line represents the $95\%$ confidence interval. The dashed horizontal line at 0 indicates that there is no effect.}
\label{fig:ssm_calls2}
\end{figure}

\begin{figure}
\includegraphics[width=\linewidth]{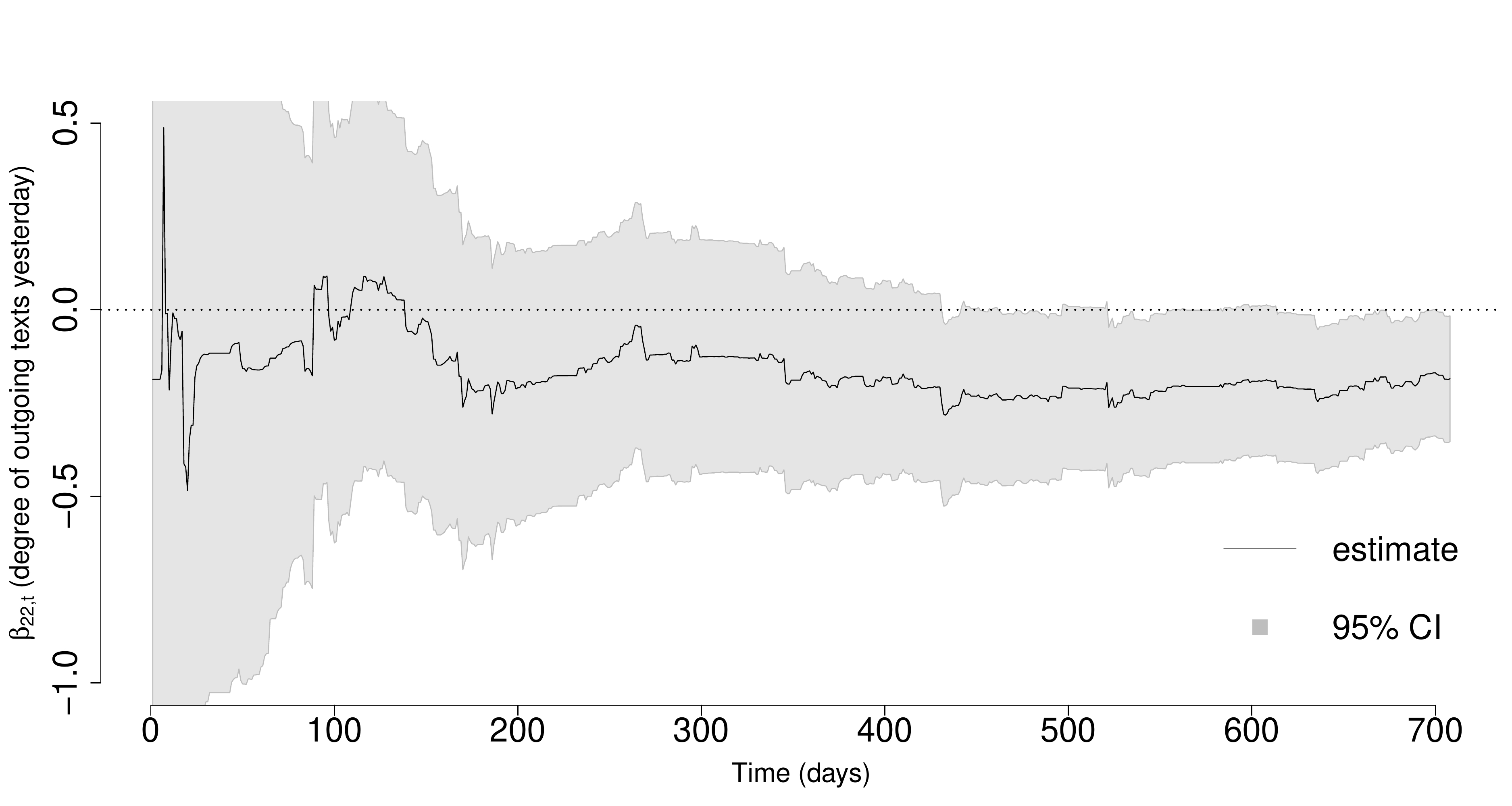}
\caption{Estimated coefficient trajectory over time for the degree of texts at the previous day $\beta_{22,t}$ using SSMimpute. The black line represents the estimated coefficient over time, while the grey line represents the $95\%$ confidence interval. The dashed horizontal line at 0 indicates that there is no effect.}
\label{fig:ssm_text2}
\end{figure}

\begin{figure}
\includegraphics[width=\linewidth]{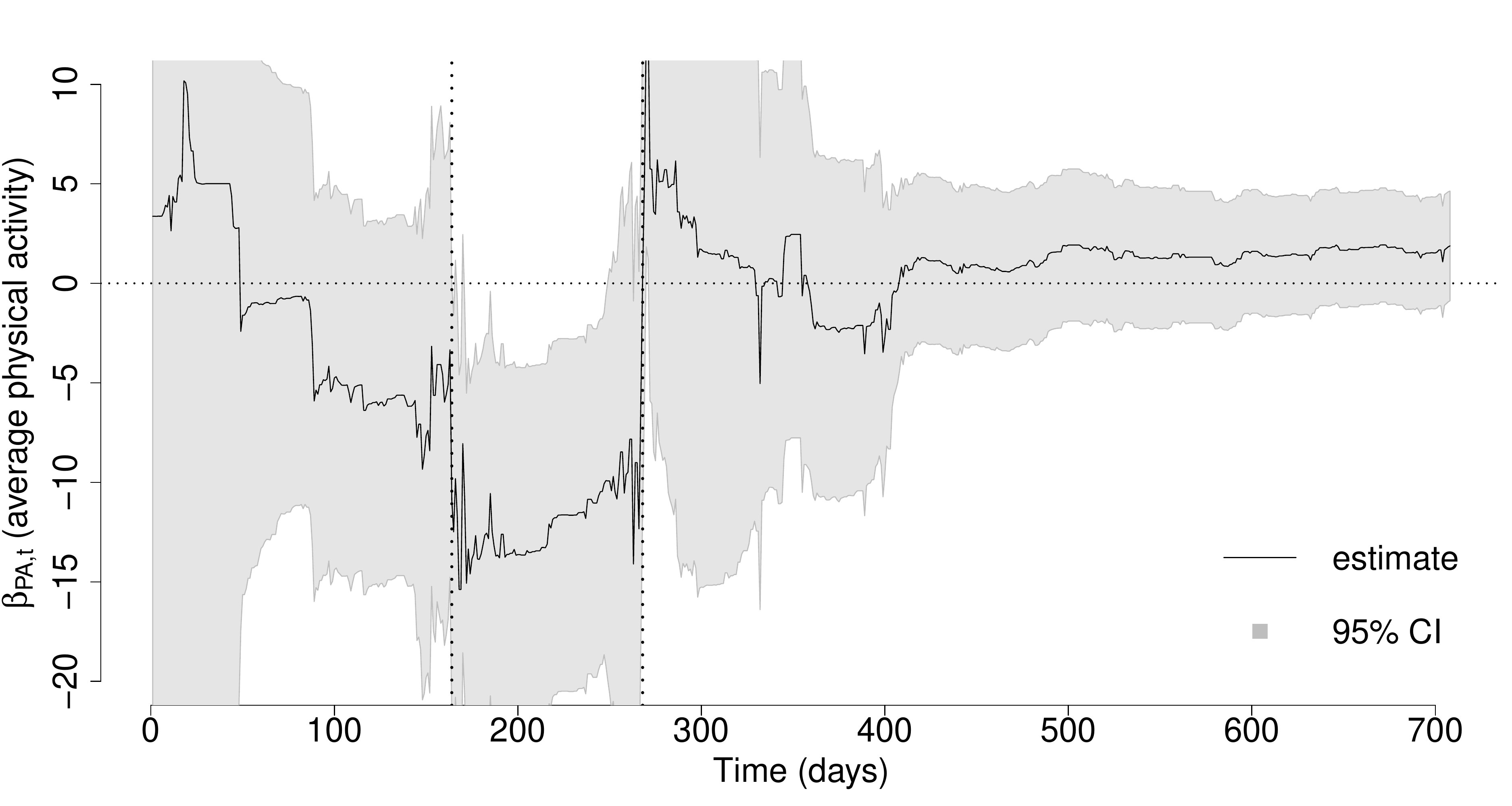}
\caption{Estimated coefficient trajectory over time for the averaged physical activity $\beta_\text{PA,t}$ using SSMimpute. The black line represents the estimated coefficient over time, while the grey line represents the $95\%$ confidence interval. The dashed horizontal line at 0 indicates that there is no effect.}
\label{fig:ssm_pa}
\end{figure}

\begin{figure}
\includegraphics[width=\linewidth]{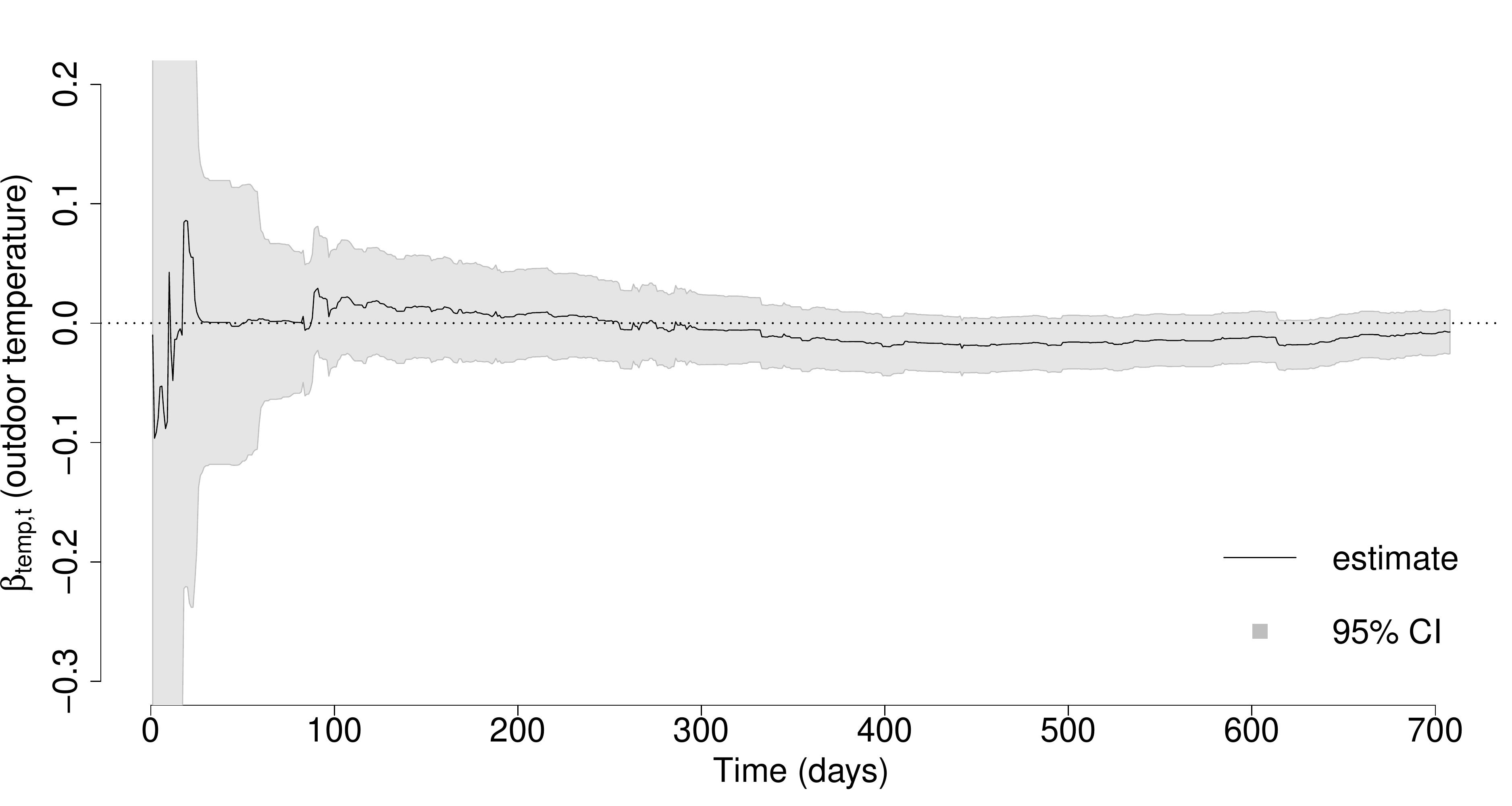}
\caption{Estimated coefficient trajectory over time for the outdoor temperature $\beta_\text{temp,t}$ using SSMimpute. The black line represents the estimated coefficient over time, while the grey line represents the $95\%$ confidence interval. The dashed horizontal line at 0 indicates that there is no effect.}
\label{fig:ssm_temp}
\end{figure}

\subsection{Posterior distribution for missing outcome imputation}
Figure~\ref{fig:convergence_process} shows the $90\%$ confidence interval of the point-wise posterior distribution from which missing outcomes are imputed. The confidence interval band contains observed outcomes adequately, validating the production of satisfactory candidates for missing outcomes in non-stationary multivariate time series.
\begin{figure}
\centering
\includegraphics[width=\linewidth]{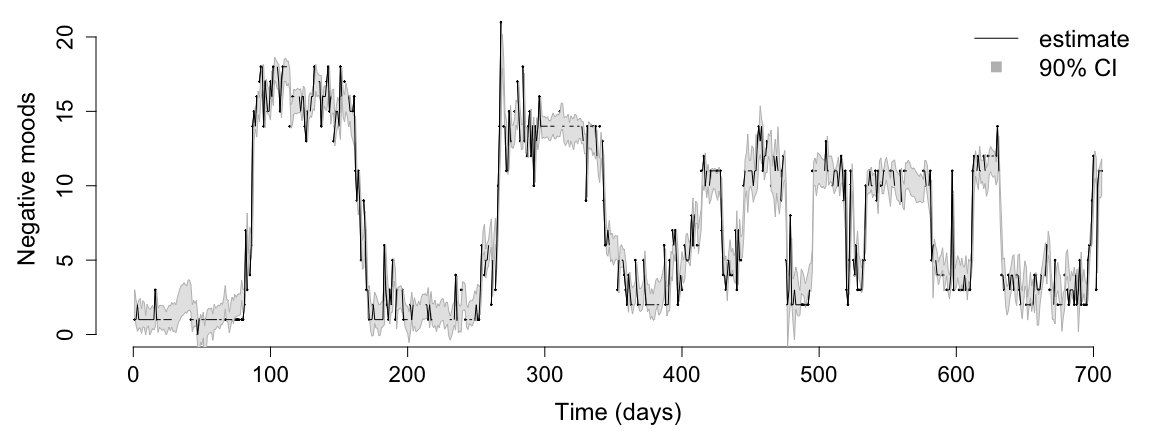}
\caption{Estimated $90\%$ confidence band of the missing outcome imputation posterior distribution.}
\label{fig:predictionband}
\end{figure}

\subsection{Illustration of convergence process}
The state space model structure of the SSMmp provides sufficient flexibility for data imputation in non-stationary time series, allowing for systemic changes over time: parameters to be estimated (e.g., coefficients for key variables and other nuisance structural parameters) can be modeled as a random walk, an auto-regressive process, or a periodic-stable process.
We allow a iterative learning process of the structure of state space model in the process of multiple imputation and iteration: we begin by assuming all parameters to change freely over time as a random walk; as iteration proceeds, the estimation trajectory for time-invariant parameters becomes a horizontal line, the estimation trajectory for periodic-stable parameters becomes a series of several horizontal lines, and the estimation trajectory for auto-regressive parameters becomes stable in level with positive variation, and the estimation trajectory for truly non-stable parameters exhibits randomness and level drift over time and are thus modeled as a random walk. Figure~\ref{fig:convergence_process} shows an example of convergence process for time-varying coefficients $\beta_{21,t}$ (coefficient for the degree of text) and $\beta_{\text{PA},t}$ (coefficient for physical activity level) in the BLS application.
\begin{figure}
    \centering 
\begin{subfigure}{0.49\textwidth}
\includegraphics[width=\linewidth]{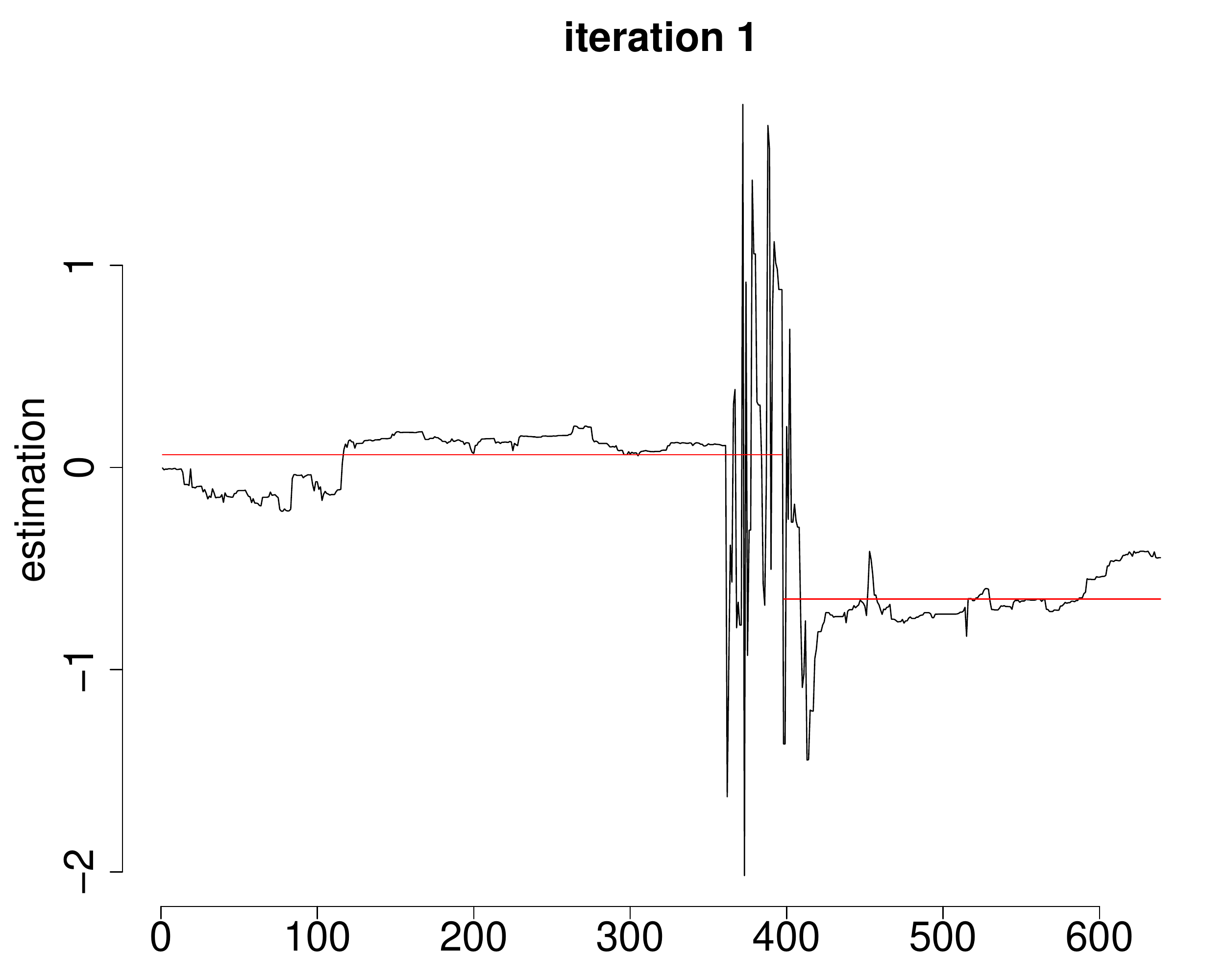}
\end{subfigure} 
\begin{subfigure}{0.49\textwidth}
  \includegraphics[width=\linewidth]{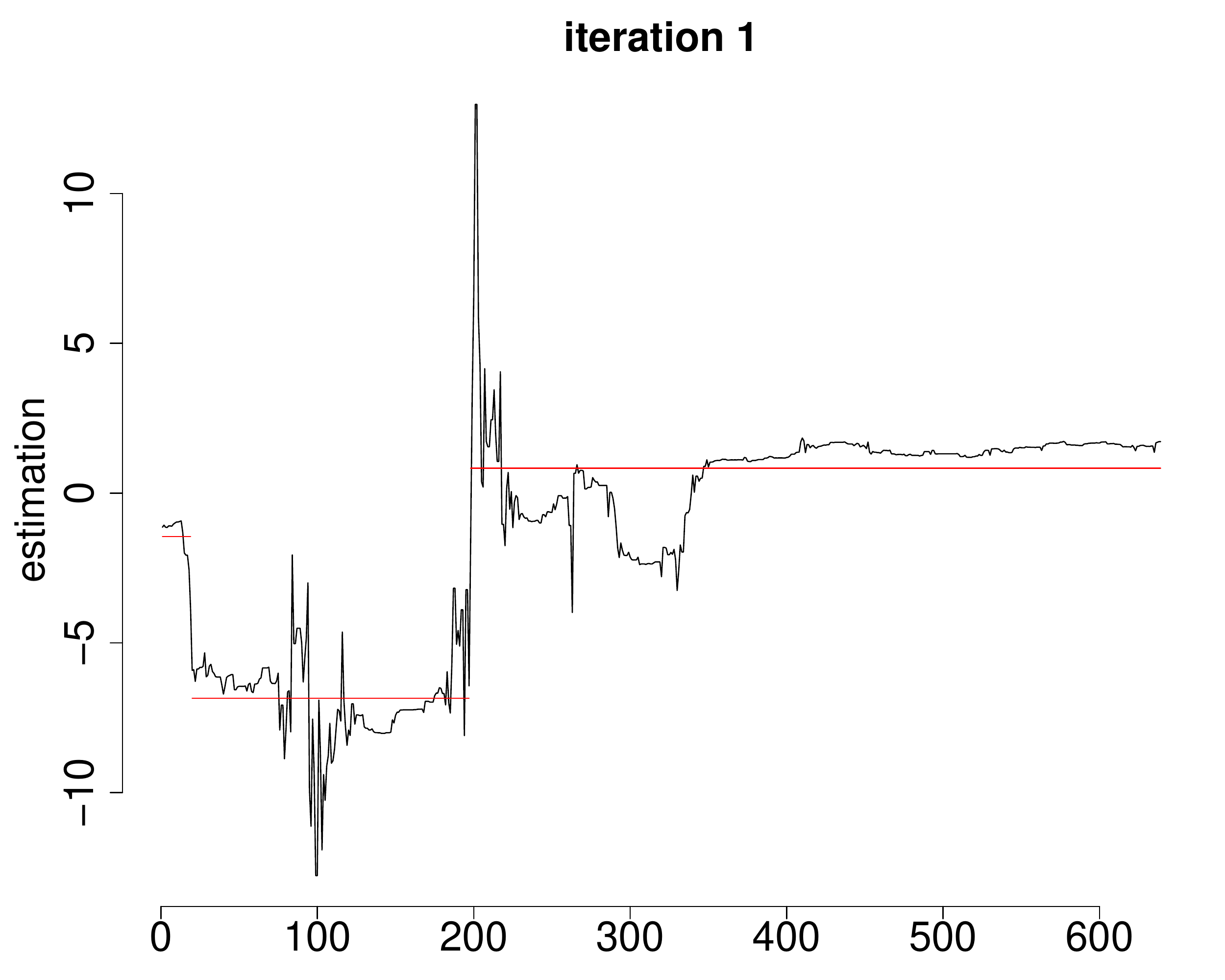}
\end{subfigure} \hfil 
\begin{subfigure}{0.49\textwidth}
  \includegraphics[width=\linewidth]{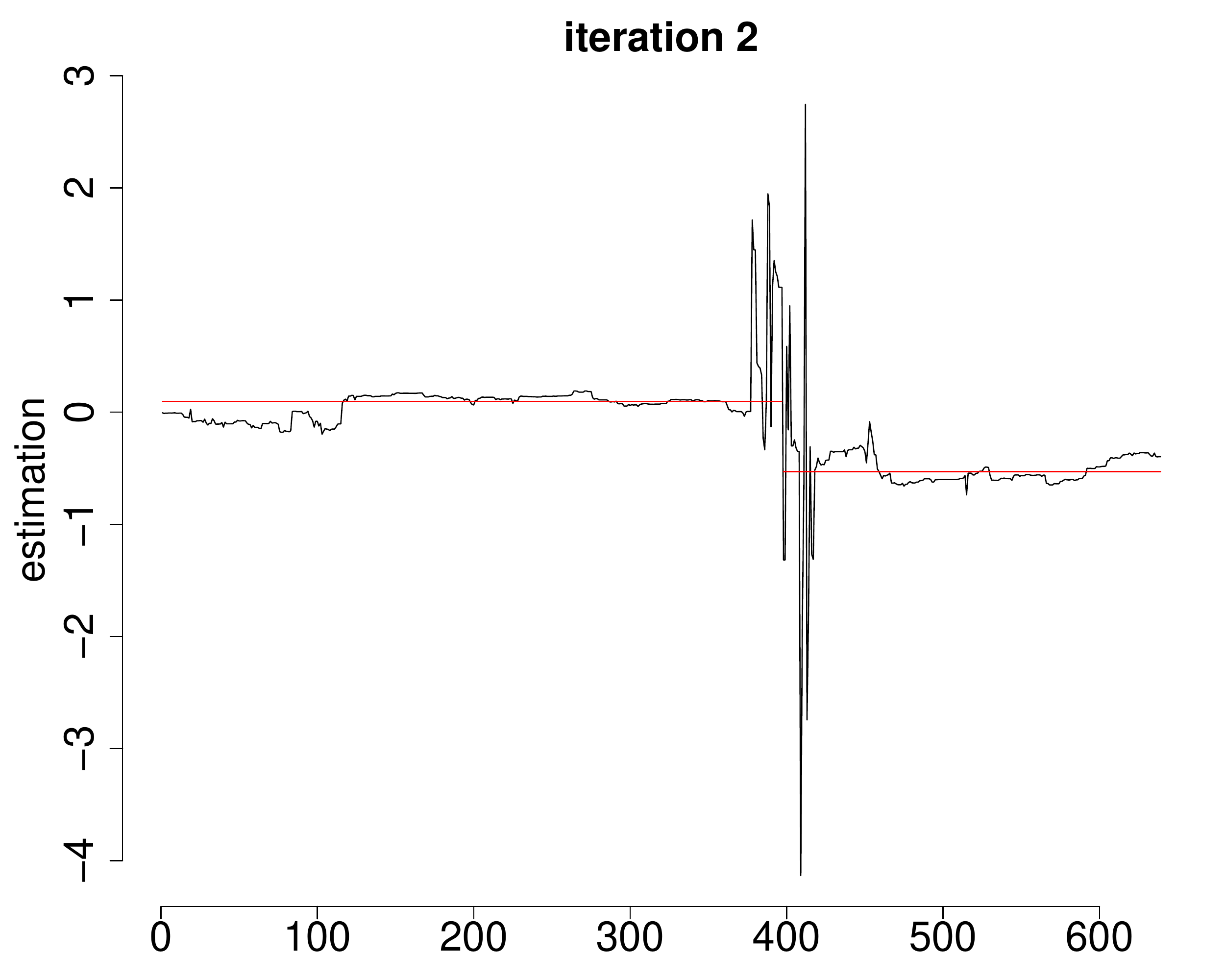}
\end{subfigure} 
\begin{subfigure}{0.49\textwidth}
  \includegraphics[width=\linewidth]{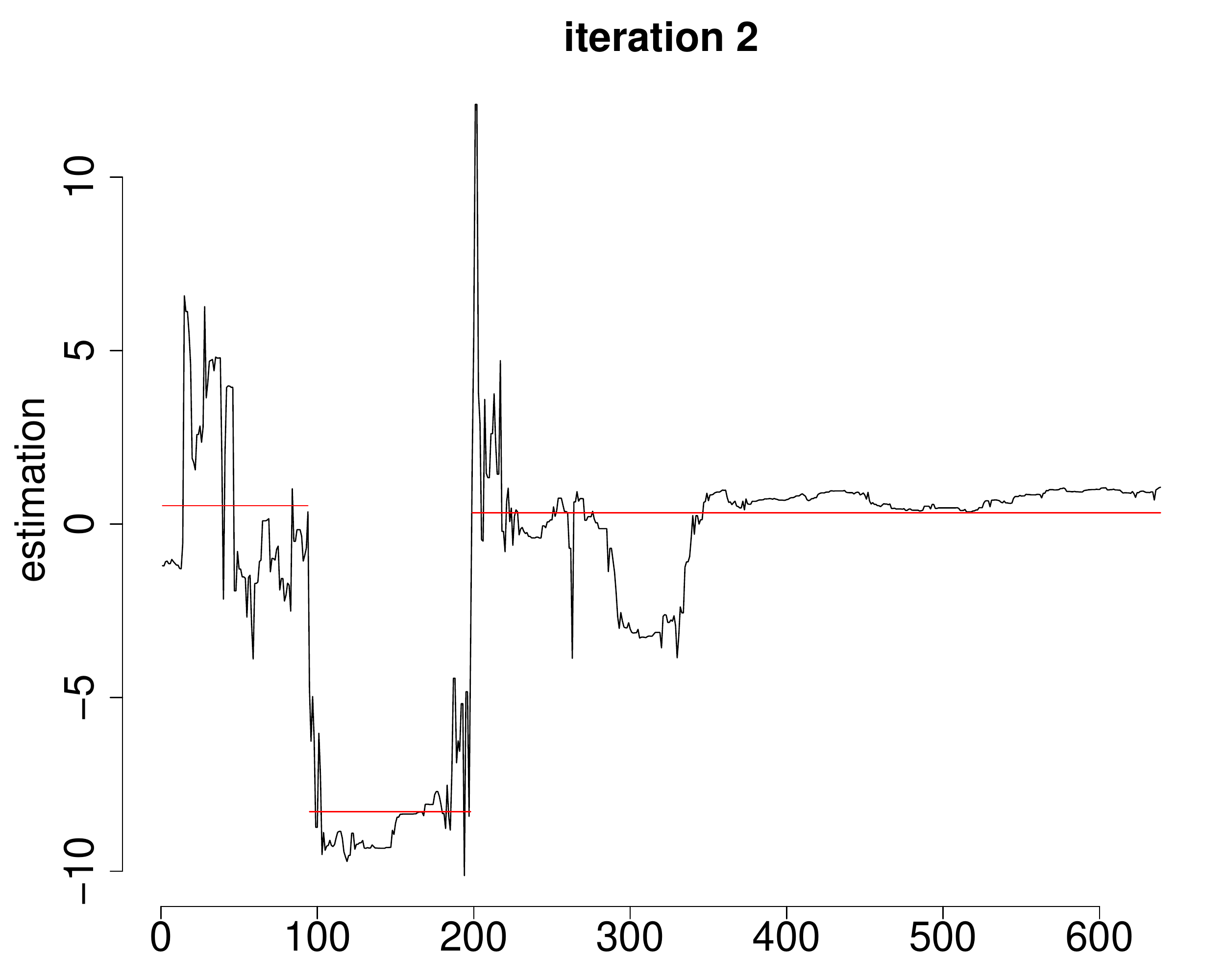}
\end{subfigure} \hfil
\begin{subfigure}{0.49\textwidth}
  \includegraphics[width=\linewidth]{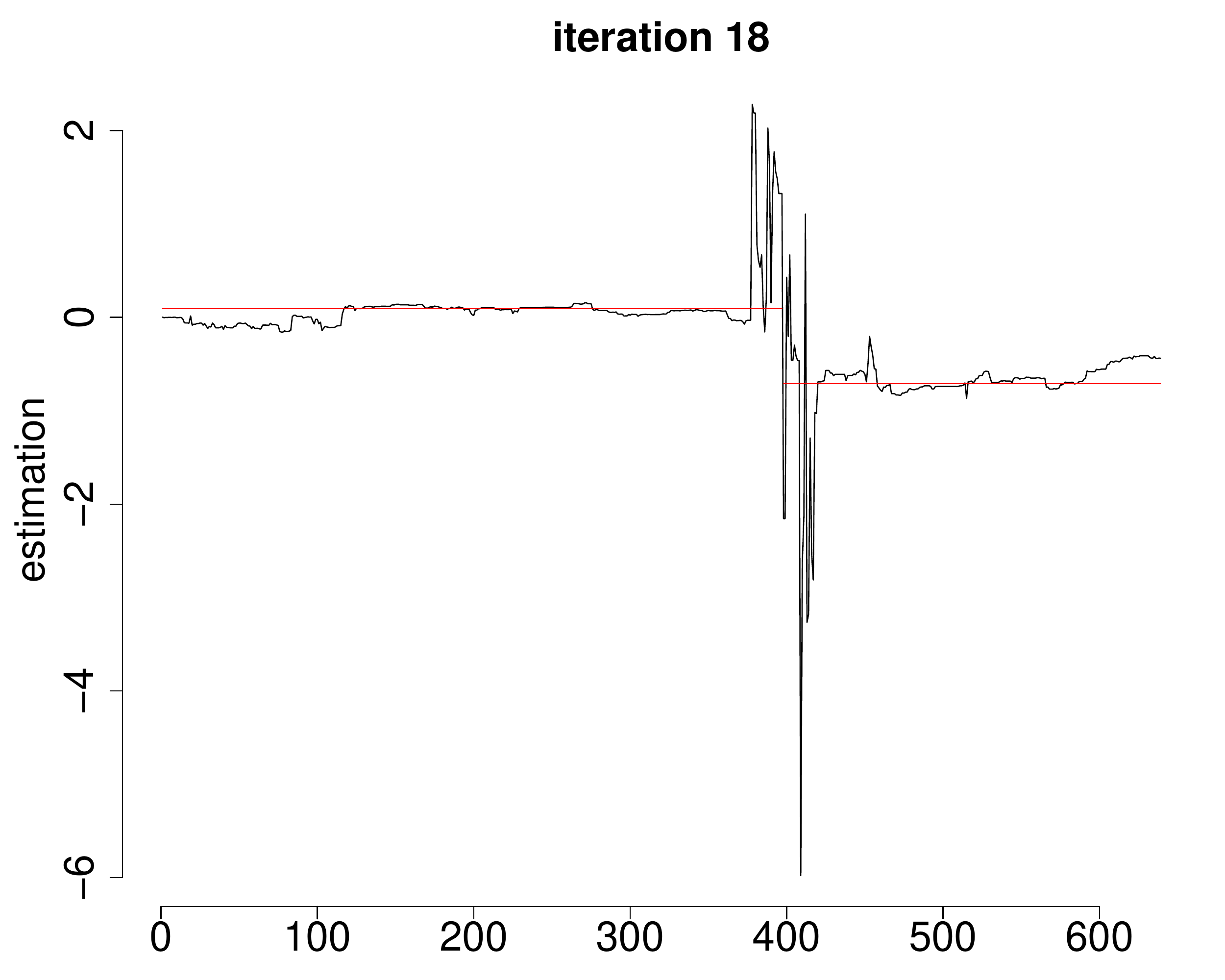}
\end{subfigure} 
\begin{subfigure}{0.49\textwidth}
  \includegraphics[width=\linewidth]{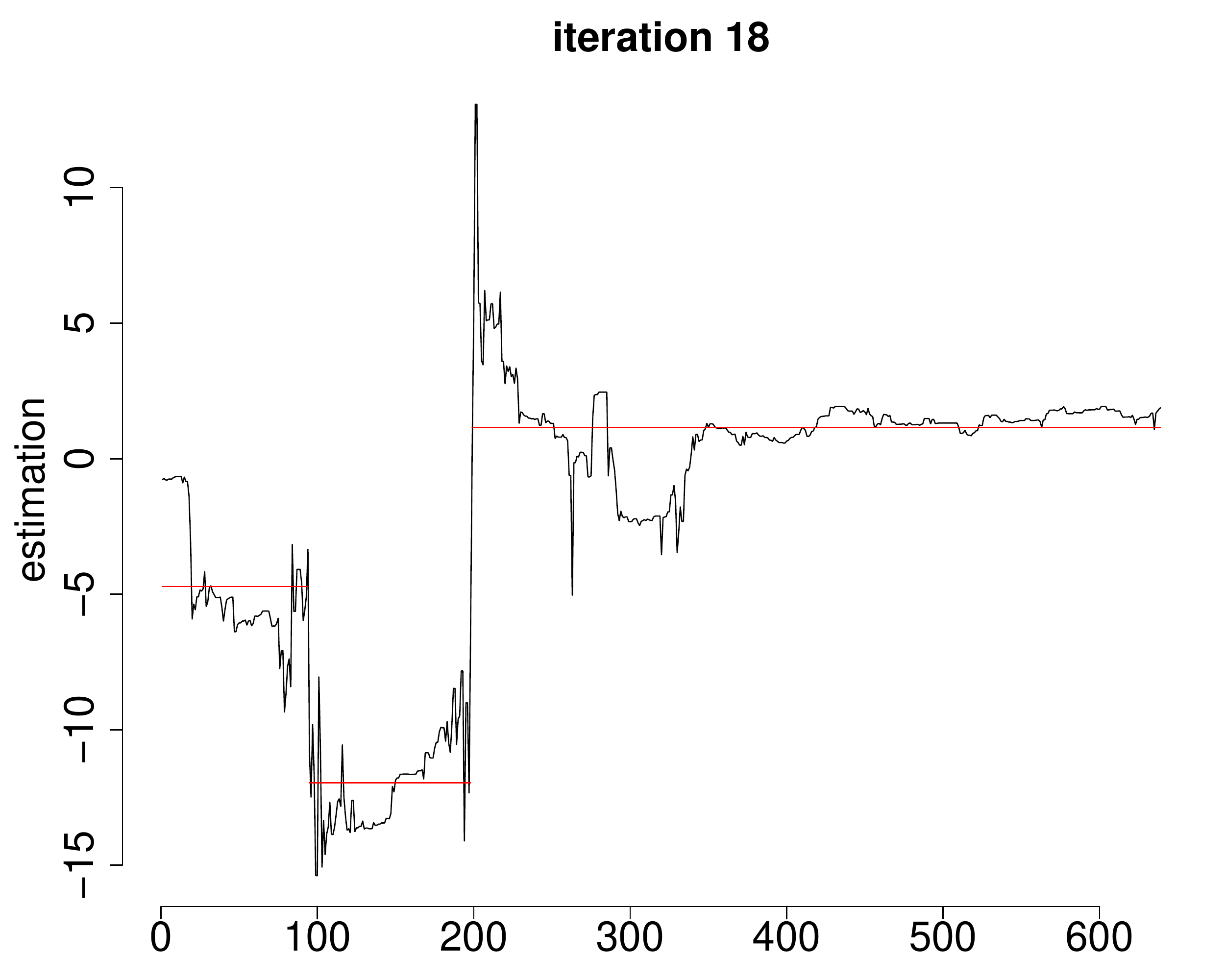}
\end{subfigure}\hfil 
\caption{Process of convergence for periodic-stable time-varying parameters $\beta_{21,t}$ (left panel) and $\beta_{\text{PA},t}$ (right panel) until convergence at iteration 18.}
\label{fig:convergence_process}
\end{figure}

\end{document}